\documentclass[11pt]{article}
\usepackage[utf8]{inputenc}
\usepackage[margin=1in]{geometry}

\usepackage{microtype}
\usepackage{subfigure}
\usepackage{booktabs}
\usepackage[round]{natbib}

\usepackage{amsmath,amsthm,amsfonts,amssymb,mathdots,array,mathrsfs,bm,bbm,stmaryrd,graphicx,subfigure,xcolor}

\usepackage[linesnumbered, ruled, vlined]{algorithm2e}

\let\oldnl\nl
\newcommand{\nonl}{\renewcommand{\nl}{\let\nl\oldnl}}

\usepackage{breakcites}

\usepackage[T1]{fontenc}
\usepackage{enumerate}
\usepackage{inputenc}

\usepackage{graphicx} 
\usepackage{subfigure}

\usepackage{booktabs,balance}
\usepackage{rotating}
\usepackage{boldline}
\usepackage{makecell}
\usepackage{multirow}
\usepackage{balance}
\usepackage{wrapfig}

\usepackage{tikz}

\newcommand{\eqdef}{\mathrel{\mathop:}=}

\usepackage[colorlinks=true,citecolor=blue]{hyperref}

\newtheorem{theorem}{Theorem}[section]

\newtheorem{lemma}[theorem]{Lemma}
\newtheorem{proposition}[theorem]{Proposition}

\newtheorem{definition}{Definition}[section]

\newtheorem{remark}{Remark}[section]

\newcommand{\bsmat}{\begin{bmatrix} }
\newcommand{\esmat}{\end{bmatrix} }

\DeclareMathOperator{\erf}{erf}

\usepackage{environ}
\NewEnviron{smallequation}{%
    \begin{equation}
    \scalebox{0.97}{$\BODY$}
    \end{equation}
    }
    \NewEnviron{smallalign}{%
    \begin{equation}
    \scalebox{0.97}{$\BODY$}
    \end{equation}
    }

\begin{document}

\title{\bf Differential Privacy with Random Projections\\ and Sign Random Projections \vspace{-0.1in}}

\author{
  \textbf{Ping Li, \ \ Xiaoyun Li  } \\
  LinkedIn Ads\\
  700 Bellevue Way NE, Bellevue, WA 98004, USA\\
  \texttt{ \{pinli,  xiaoyli\}@linkedin.com}
}
\date{\vspace{-0.36in}}
\maketitle

\begin{abstract}

\noindent In this paper, we develop a series of differential privacy (DP) algorithms from a family of random projections (RP), for general applications in machine learning, data mining, and information retrieval. Among the presented algorithms, \textbf{iDP-SignRP} is remarkably effective under the setting of ``individual differential privacy'' (iDP), based on sign random projections (SignRP). Also, \textbf{DP-SignOPORP} considerably improves existing algorithms in the literature under the standard DP setting, using ``one permutation + one random projection'' (OPORP), where OPORP is a variant of the celebrated count-sketch method with fixed-length binning and normalization. Without taking signs, among the DP-RP family, \textbf{DP-OPORP} achieves the best performance.

\vspace{0.1in}

\noindent The concept of iDP (individual differential privacy) is defined  only on a particular dataset of interest. While iDP is not strictly DP, it might be useful in certain applications, such as releasing a dataset (including sharing  embeddings across companies or countries). In our study, we find that \textbf{iDP-SignRP} is remarkably effective for search and machine learning applications, in that the utilities  are exceptionally good even at a  very small privacy parameter $\epsilon$ (e.g., $\epsilon<0.5$).

\vspace{0.1in}

\noindent Privacy can be protected at various stages of data life cycle. In this study, our methods can be applied as early as at data collection time. Instead of directly using the original $p$-dimensional vector $u$, we let $x_j = \sum_{i=1}^p u_i w_{ij}$, where $j=1$ to $k$ and $w_{ij}$ is sampled from the  Gaussian distribution. This is known as random projections (RP). In this paper, we assume the projection matrix $\{w_{ij}\}$ is also  released to the public, as a stronger privacy setting. The previous algorithm  named ``DP-RP-G'' can be improved quite considerably by setting the noise more carefully. The proposed new variant, named ``DP-RP-G-OPT'', compares favorably with directly adding noise to the original data. DP-RP-G-OPT can be further improved by ``DP-RP-G-OPT-B'' by replacing the Gaussian projections with Rademacher projections due to the smaller sensitivity. Finally, via the fixed-length binning mechanism, ``DP-OPORP'' slightly outperforms DP-RP-G-OPT-B. Empirical evidence shows that these DP-RP algorithms require $\epsilon\approx 20$ to achieve a good utility.

\vspace{0.12in}

\noindent Our key idea for improving DP-RP is to take only the signs, i.e., $sign(x_j) = sign\left(\sum_{i=1}^p u_i w_{ij}\right)$, of the projected data. The intuition is that the signs  often remain unchanged when the original data ($u$) exhibit small changes (according to the ``neighbor'' definition in DP). In other words, the aggregation and  quantization operations themselves provide good privacy protections. We develop a technique called \textbf{``smooth flipping probability''} that incorporates this intuitive privacy benefit of SignRPs and improves the standard DP bit flipping strategy. Based on this technique, we propose \textbf{DP-SignOPORP} which satisfies strict DP and outperforms other DP variants based on SignRP (and RP), especially when $\epsilon$ is not very large (e.g., $\epsilon = 5\sim10$). Moreover, if an application scenario accepts individual DP, then we immediately obtain an algorithm named \textbf{iDP-SignRP} which achieves excellent utilities even at small~$\epsilon$ (e.g., $\epsilon<0.5$).

\vspace{0.12in}

\noindent In practice, the main obstacle for deploying DP at the source of data is the severe degradation of performance  because the amount of required noise is typically high. We hope our proposed series of DP (and iDP) algorithms, e.g., \textbf{DP-OPORP}, \textbf{DP-SignOPORP}, \textbf{iDP-SignRP},  will help promote the industrial applications of DP  in search, retrieval, ranking, and AI in general.
\end{abstract}

\newpage

\hypersetup{
colorlinks,linkcolor=red,filecolor=blue,citecolor=blue,urlcolor=red, linktoc=page}

\tableofcontents

\hypersetup{
colorlinks,linkcolor=red,filecolor=blue,citecolor=blue,urlcolor=red}

\newpage

\section{Introduction}  \label{sec:intro}

With the rapid growth of capable electronic devices, personal data has been continuously  collected by  companies/organizations. Protecting data privacy has  become an urgent need and a trending research topic in computer science, statistics, applied mathematics, etc. Among many notions of privacy, the ``differential privacy'' (DP)~\citep{dwork2006calibrating} has gained tremendous attention in both the research community and industrial applications. The intuition of DP is straightforward:

\vspace{0.1in}
\leftskip=0.3in\rightskip=0.3in
\textit{We want to find a randomized data output procedure, such that a small change in the database can hardly be detected by an adversary based on the observation of the output.}
\vspace{0.1in}

\leftskip=1pt\rightskip=1pt

\vspace{0.1in}
\noindent Differential privacy provides a formal mathematical definition of privacy (see Section~\ref{sec:pre} for the details) that is powerful in the sense that it protects the sensitive data regardless of how the adversary utilizes the data subsequently. The strength of DP is usually characterized by two parameters, $\epsilon$ and $\delta$. If $\delta=0$, we say that the algorithm is ``pure DP'' (i.e., $\epsilon$-DP); if $\delta>0$, the algorithm is called ``approximate DP'' (i.e., $(\epsilon,\delta)$-DP). The typical technique of DP is to randomize/perturb the algorithm output by adding noise or probabilistic sampling. DP has been widely applied to numerous tasks such as clustering~\citep{feldman2009private,gupta2010differentially}, regression and classification~\citep{chaudhuri2008privacy,zhang2012functional}, DP-SGD~\citep{abadi2016deep,agarwal2018cpsgd,fang2023improved}, principle component analysis~\citep{ge2018minimax}, empirical risk minimization~\citep{chaudhuri2011differentially}, matrix completion~\citep{blum2005practical,jain2018differentially}, graph distance estimation~\citep{kasiviswanathan2013analyzing,fan2022distances,fan2022private}. In industry, DP has been broadly deployed to collect user data for (e.g.,) frequency or mean estimation at, for example, Google~\citep{erlingsson2014rappor}, Apple\footnote{\url{https://docs-assets.developer.apple.com/ml-research/papers/learning-with-privacy-at-scale.pdf}}, and Microsoft~\citep{ding2017collecting}.

The data to be protected may have various forms and come from various sources, for example: (i) the user ``profile'' vectors that contain sensitive attributes like personal information, click history, etc.; (ii) the embedding vectors learned from the user data, which may also contain rich information about the users and could be attacked by malicious adversaries to infer user attributes~\citep{beigi2020privacy,zhang2021graph}. In a machine learning model, there are several choices on when to deploy differential privacy (DP): (a) at the data collection/processing stage~\citep{zhang2014privbayes,yang2015bayesian,berlioz2015applying,cormode2018marginal}; (b) during the model training stage~\citep{abadi2016deep,wei2020federated}; and (c) on the model output (or summary data statistics)~\citep{dwork2009differential}. While each of them may find its need depending on the application scenarios, applying DP as early as at the data collection stage provides strong protection in the sense that the subsequent operations or outputs will become private by the post-processing property of differential privacy. Moreover, the data holder may need to release the data to third-parties. In this case, a private data publishing mechanism (at an early stage) becomes necessary.

In this work, we study differential privacy algorithms for protecting data vectors, based on the broad family of random projections (RP) and sign random projections (SignRP). Our algorithms can be applied to data publishing and processing for  downstream machine learning tasks. In this paper, besides the standard DP, we also consider a relaxation of DP called ``individual differential privacy'' (iDP)~\citep{comas2017individual}, which was proposed with the aim of improving the utility of private algorithms. It turns out that, when it is applied to SignRP, iDP indeed provides a much better utility, a phenomenon which might be interesting to the DP researchers and practitioners and provides a plausible option for balancing the privacy and utility in practice.

\subsection{Random Projections (RP) and Very Sparse Random Projections (VSRP)}

In practice,  compression and dimension reduction techniques can be crucial when dealing with massive (high-dimensional) data. The random projection (RP) method is an important fundamental algorithm. Denote $u\in \mathbb R^{p}$ as the data vector with $p$ features. With some $k$ (and typically $k\ll p$), we define the random projection of $u$ as
\begin{align}  \label{def:RP}
    X = \frac{1}{\sqrt{k}}W^Tu, \ \ W\in \mathbb R^{p\times k}.
\end{align}
The entries of the random matrix $W$ typically follow the Gaussian distribution or Gaussian-like distribution such as uniform. It has been well-understood that, with a sufficient number of projections ($k$), the  similarity between two vectors (using the same projection $W$) is preserved within a small multiplicative error with high probability. The dimensionality reduction and geometry preserving properties make RP widely useful in numerous applications, such as distance estimation, nearest neighbor search, clustering, classification, compressed sensing, etc.~\citep{johnson1984extensions,indyk1998approximate,dasgupta2000experiments,bingham2001random,achlioptas2003database,fern2003random,vempala2005random, candes2006robust,donoho2006compressed,li2006very,frund2007leanring,dasgupta2008random,dahl2013large,li2019generalization,rabanser2019failing,tomita2020sparse,zhang2020optimal, li2023oporp}.

\vspace{0.1in}
Instead of using dense projections, it is also common to sample the entries of $W$ from the following ``very sparse'' distribution, i.e., VSRP~\citep{li2006very}:
\begin{align}
\sqrt{s}\times \left\{\begin{array}{rrl}
-1 & \text{with prob.} &1/(2s)\\
0 & \text{with prob.} &1-1/s \ , \\
+1 & \text{with prob.} &1/(2s)
\end{array}\right.
\end{align}
which generalizes~\citet{achlioptas2003database} (for $s=1$ and $s=3$). Note that when $s=1$, it is also called the ``symmetric Bernoulli'' distribution or the ``Rademacher'' distribution. In a recent study by~\citet{li2023oporp}, the estimation accuracy of VSRP can be substantially improved via a normalization step.

\subsection{Sign Random Projections (SignRP)}

While RP is able to reduce the dimensionality to, e.g., hundreds, storing and transmitting the random projections might still be expensive for  large datasets. To this regard, one can further compress the RPs by quantization/discretization, where we only use a few bits to represent the projected data, instead of using full-precision (32 or 64 bits). Quantized random projections (QRP) and the extensions (e.g., to non-linear kernels) have many applications in various fields, e.g., compressed sensing, approximate nearest neighbor search, similarity estimation, large-scale machine learning~\citep{goemans1995improved,charikar2002similarity, datar2004locality,boufounos20081bit,dong2008asymmetric,zymnis2010compressed,leng2014random,li2014coding,knudson2016one,slawski2018trade,li2019sign,li2019random, xu2021locality,li2021quantization}. In particular, in this paper, we will consider the differential privacy of the extreme case of sign (1-bit) random projection (SignRP), also known as the ``SimHash'' in the  literature, where only the sign of the projected data is stored. Compared with storing the full-precision floats, SignRP  provides substantial memory and communication saving which could be beneficial for large-scale databases.

\newpage

From the sign random projections, one can estimate the angle/cosine between the data vectors. Consider two data vectors $u$ and $v$ with $\rho=\frac{u^Tv}{\|u\|\|v\|}$ and $\theta=\cos^{-1}(\rho)$. Denote $x=\frac{1}{\sqrt{k}}W^Tu$ and $y=\frac{1}{\sqrt{k}}W^Tv$, with $x,y\in\mathbb{R}^k$. The standard result~\citep{goemans1995improved,charikar2002similarity} gives the collision probability for a single sign random projection as
\begin{align}
    Pr\left(sign(x_j)=sign(y_j)\right)=1-\frac{\theta}{\pi},\ \ \forall j=1,...,k,  \label{eqn:signRP-collision-prob}
\end{align}
which leads to a straightforward estimator of the angle $\theta$ as
\begin{align}  \label{est:signRP}
    \hat\theta=\pi \left(1-\frac{1}{k}\sum_{j=1}^k \mathbbm 1\{sign(x_j)=sign(y_j)\} \right).
\end{align}
As $\mathbbm 1\{sign(x_j)=sign(y_j)\}$ follows a Bernoulli distribution, the mean and variance of $\hat\theta$ can be precisely computed:
\begin{align}  \label{eqn:var-signRP}
    \mathbb E[\hat\theta]=\theta,\quad Var(\hat\theta)=\frac{\theta(\pi-\theta)}{k}.
\end{align}

In addition to the mentioned benefits of achieving highly compact representations of the original data, SignRP has also been widely used for approximate near neighbor (ANN) search~\citep{friedman1975algorithm} by building hash tables from the SignRP bits~\citep{indyk1998approximate,charikar2002similarity,shrivastava2014defense}.

\subsection{Count-Sketch and OPORP: One Permutation + One Random Projection}

The celebrated count-sketch data structure~\citep{charikar2004finding} can be viewed as a special type of very sparse projections. Count-sketch is highly efficient because, as an option, it  only requires ``one permutation + one random projection'' (OPORP), as opposed to $k$ projections. The recent work by~\citet{li2023oporp} improved the original count-sketch in two aspects: (i) Using a fixed-length binning scheme makes the algorithms more convenient and also reduces the variance by a factor of $\frac{p-k}{p-1}$. (ii) The projected data (i.e., vectors in $k$ dimensions) should be normalized before they are used to estimate the original cosine. By doing so, the estimation variance can be  substantially reduced, essentially from $(1+\rho^2)/k$ (un-normalized)  to $(1-\rho^2)^2/k$ (normalized).

As the name OPORP suggests, the original data vector $u$ is first permuted by a random permutation, and then it is divided into $k$ bins. Within each bin, we apply a random projection of size $p/k$ to obtain one projected value. Thus, we still obtain $k$ projected values but the procedure is drastically more efficient than the standard random projections.  The count-sketch data structure and variants have been widely used in  applications. Recent examples include graph embedding~\citep{wu2019demo}, word \& image embedding~\citep{chen2017reading,zhang2020dynamic,abdullah2021finding,singhal2021federated,zhang2022kernelized},  model \& communication compression~\citep{weinberger2009feature, chen2015compressing,rothchild2020fetchsgd,haddadpour2020fedsketch,li2022gcwsnet} for linear models as well as deep neural nets, etc.


\subsection{DP Through RP and SignRP: Applications, Algorithms and Results} \label{sec:intro-setting}

In this paper, we propose to use random projection (RP) and sign random projection (SignRP) for data privatization. We name our methods ``differentially private random projections'' (DP-RP) and DP-SignRP. Additionally, we  develop a series of algorithms for OPORP, named DP-OPORP and DP-SignOPORP. Finally, we will present iDP-SignRP using individual differential privacy (iDP).

\newpage

The typical use cases of the proposed algorithms include:
\begin{itemize}
    \item In on-device data collection, each user can apply our methods locally (with shared projection matrix/vectors) before sending the data to the server in the data collection process.

    \item At the central data silo, the data curator can first apply our methods to the data attributes (or embeddings) and then publish the output (randomized data)  for data analysis or modeling.
\end{itemize}

By the geometry-preserving property of RP and SignRP, our approaches could be  naturally and effectively applied to a variety of distance/similarity based downstream tasks like (approximate) nearest search, clustering, classification, and regression. Additionally, with DP-RP and DP-SignRP, the data compression (especially for high-dimensional datasets) is also able to consequently reduce the computational cost and the communication overhead when transmitting the data.

More specifically, for a data vector $u\in\mathbb R^{p}$, our goal is to make its RP and SignRP, namely $x=W^Tu$ and $sign(X)$, differentially private. Then, for a database $U$ consisting of $n$ data vectors, applying this mechanism to all the points makes $U$ differentially private against a change in one data record. In the literature, there are works on randomized algorithms that assume the randomness of the algorithm (e.g., the hash keys or projection matrices) are ``internal'' and also kept private~\citep{blocki2012johnson,smith2020flajolet,dickens2022order}. We assume the projection matrix $W$ is known/public, which is stronger in privacy and enables broader applications of the released data. In addition, a ``known projection matrix'' is the typical case in the local DP model~\citep{kairouz2014extremal,cormode2018privacy}, since the projection matrix is created and shared among the users by the data aggregator.

\vspace{0.1in}

In summary, in this paper we make the following contributions:
\begin{itemize}
    \item In Section~\ref{sec:DP-RP}, for DP-RP, we first revisit/analyze the prior work~\citep{kenthapadi2013privacy} on the $(\epsilon,\delta)$-DP random projection based on Gaussian noise addition (called DP-RP-G). Then we propose an optimal Gaussian noise mechanism DP-RP-G-OPT, which improves the previous method. Then, we present DP-RP-G-OPT-B (which uses Rademacher projection and optimal noise) and DP-OPORP. We compare these variants of DP-RP with the strategy of adding (optimal) Gaussian noise to the raw data, in terms of the accuracy of inner product estimation.

    \item In Section~\ref{sec:DP-signRP}, we propose two algorithms for privatizing the (1-bit) DP-SignRP. The first method, DP-SignRP-RR, is based on the standard ``randomized response'' (RR). Then, we propose an improved method named DP-SignRP-RR-smooth based on a proposed concept of ``smooth flipping sensitivity'', thanks to the robustness brought by the ``aggregate-and-sign'' operation of SignRP. Finally, we extend the idea of smooth flipping probability to the OPORP variant, and show that DP-SignOPORP is more advantageous in terms of differential privacy.

    \item In Section~\ref{sec:experiment}, we conduct retrieval and classification experiments on benchmark datasets. Specifically, for DP-RP, the proposed optimal DP-RP-G-OPT method substantially improves DP-RP-G, which is the previous DP-RP method in the literature. DP-OPORP performs similarly to DP-RP-G-OPT-B. For the 1-bit variants, our proposed DP-SignOPORP outperforms DP-RP variants and DP-OPORP especially when $\epsilon$ is not large (e.g., $\epsilon = 5\sim10$).

    \item In Section~\ref{sec:individual DP}, we further design iDP-SignRP algorithms under the setting of individual differential privacy (iDP), which is a relaxed definition of the standard DP~\citep{comas2017individual}. Experiments demonstrate that iDP-SignRP is able to achieve excellent utilities (precision, recall, classification accuracy) even when $\epsilon$ is small (e.g.,  $\epsilon<0.5$). We anticipate that iDP-SignRP will be useful (e.g.,) as a data-publishing mechanism, among other applications.
\end{itemize}

\newpage

\section{Preliminaries}  \label{sec:pre}

Throughout the paper, $\|\cdot\|_1$ and $\|\cdot\|_2$ are the $l_1$ and $l_2$ norms, respectively. For convenience, $\|\cdot\|$ will denote the $l_2$ norm if there is no risk of confusion. Let $u\mathcal U\subset \mathbb R^p$ be a data vector of $p$ dimensions. We assume $\|u \|>0$, i.e., there is no all-zero data sample.

\subsection{Differential Privacy}

Differential privacy (DP) has been one of the standard tools in the literature to protect against (e.g.,) attribute inference attacks. The formal definition of DP is as follows.

\vspace{0.1in}

\begin{definition}[Differential Privacy~\citep{dwork2006calibrating}] \label{def:DP}
For a randomized algorithm $\mathcal M:\mathcal U\mapsto Range(\mathcal M)$, if for any two adjacent datasets $u$ and $u'$, it holds that
\begin{equation} \label{eq:DP-def}
    Pr[\mathcal M(u)\in O] \leq e^\epsilon Pr[\mathcal M(u')\in O]+\delta
\end{equation}
for $\forall O\subset Range(\mathcal M)$ and some $\epsilon,\delta\geq 0$,
then algorithm $\mathcal M$ is called $(\epsilon,\delta)$-differentially private. If $\delta=0$, $\mathcal M$ is called $\epsilon$-differentially private.
\end{definition}

Consider the classic noise addition approaches in DP. Given some definition of ``neighboring'', one important quantity is the ``sensitivity'' described below.

\begin{definition}[Sensitivity] \label{def:sensitivity}
Let $Nb(u)$ denote the neighbor set of $u$, and $\mathcal N(\mathcal U)=\{(u,u'):u'\in Nb(u), u,u'\in \mathcal U\}$ be the collection of all possible neighboring data pairs. The $l_1$-sensitivity and $l_2$-sensitivity of a function $f:\mathcal U\mapsto \mathbb R^k$ are defined respectively as
\begin{align*}
&\Delta_1=\max_{(u,u')\in \mathcal N(\mathcal U)}\|f(u)-f(u')\|_1,\quad \Delta_2=\max_{(u,u')\in \mathcal N(\mathcal U)}\|f(u)-f(u')\|_2.
\end{align*}
\end{definition}

In this work, our goal is to protect the data itself independent of subsequent tasks or operations. We follow the natural and standard definition in the DP literature of ``neighboring'' for data vectors, that two vectors are adjacent if they differ in one element (e.g.,~\citet{dwork2006calibrating,kenthapadi2013privacy,xu2013differentially,dwork2014algorithmic,smith2020flajolet,stausholm2021improved,dickens2022order,zhao2022differentially,li2023differentially}).


\vspace{0.1in}
\begin{definition}[$\beta$-adjacency] \label{def:neighbor}
Let $u\in [-1,1]^p$ be a data vector. A vector $u'\in [-1,1]^p$ is said to be $\beta$-adjacent to $u$ if $u'$ and $u$ differ in one dimension $i$, and $|u_i-u_i'|\leq \beta$.
\end{definition}

Data vectors $u$ and $u'$ satisfying Definition~\ref{def:neighbor} are called $\beta$-adjacent or $\beta$-neighboring. Several remarks follow:
\begin{itemize}
\item Compared with $l_1$ difference: $\|u-u'\|_1\leq \beta$, Definition~\ref{def:neighbor} (i.e., $u'$ and $u$ differ in one dimension $i$ with $|u_i-u_i'|\leq \beta$) considers the worst case of the $l_1$ difference in sensitivity analysis.

\item In the literature (e.g.,~\cite{kenthapadi2013privacy} for RP), $\beta=1$ is typically used. We allow a more general $\beta$ to accommodate the need from practical applications.

\item Assuming $u\in[-1,1]^p$ is merely for convenience, as one can easily extend the results to $u\in[-C, C]^p$ for an arbitrary $C>0$. Note that, coincidentally, the ``max normalization'' is a common data pre-processing strategy, where each feature dimension is normalized to $[-1,1]$.
\end{itemize}

\vspace{0.1in}
We make a more detailed explanation of the first point above. Since the neighboring data vector $u'$ and the raw data $u$ only differ in one attribute by at most $\beta$, by Definition~\ref{def:sensitivity}, it is easy to see that for RP the sensitivities are precisely $\frac{1}{\sqrt{k}}$ times the largest row norm of $W$:
\begin{align}
    \Delta_1&=\frac{1}{\sqrt{k}}\beta\max_{(u,u')\in \mathcal N(\mathcal U)} \|W^Tu-W^Tu'\|_1=\frac{1}{\sqrt{k}}\beta\max_{i=1,...,p} \|W_{[i,:]}\|_1, \label{eq:l1-sensitivity-known} \\
    \Delta_2&=\frac{1}{\sqrt{k}}\beta\max_{(u,u')\in \mathcal N(\mathcal U)} \|W^Tu-W^Tu'\|_2=\frac{1}{\sqrt{k}}\beta\max_{i=1,...,p} \|W_{[i,:]}\|_2. \label{eq:l2-sensitivity-known}
\end{align}
where $W_{[i,:]}$ is the $i$-th row of $W$. It is a known fact  that Definition~\ref{def:neighbor} for neighboring is the worst case for the alternative neighboring definition using the $l_1$ difference: ``$u$ and $u'$ are neighboring if $\|u-u'\|_1=\sum_{i=1}^p |u_i-u_i'|\leq \beta$.'' Particularly, in our work, we have $f(u) = \frac{1}{\sqrt{k}}W^T u$, and thus
\begin{align*}
    \frac{1}{\sqrt{k}}\max_{\|u-u'\|_1\leq \beta}\|W^Tu-W^Tu'\|_1&=\frac{1}{\sqrt{k}}\max_{\|u-u'\|_1\leq \beta}\|\sum_{i=1}^p (u_i-u_i')W_{[i,:]}  \|_1 \notag \\
    &\leq \frac{1}{\sqrt{k}}\max_{\|u-u'\|_1\leq \beta}\sum_{i=1}^p |u_i-u_i'|\|W_{[i,:]} \|_1 \notag \\
    &=\frac{1}{\sqrt{k}}\beta\max_{i=1,...,p} \|W_{[i,:]}\|_1 = \Delta_1.
\end{align*}
Note that the last line equals (\ref{eq:l1-sensitivity-known}) corresponding to the sensitivity based on Definition~\ref{def:neighbor}. Since the neighboring condition in Definition~\ref{def:neighbor} also satisfies $\|u-u'\|_1\leq \beta$, we know that the equality can be exactly attained. Similarly, for the $l_2$ sensitivity, we have
\begin{align*}
    \frac{1}{\sqrt{k}}\max_{\|u-u'\|_1\leq \beta}\|W^Tu-W^Tu'\|_2&=\frac{1}{\sqrt{k}}\max_{\|u-u'\|_1\leq \beta}\|\sum_{i=1}^p (u_i-u_i')W_{[i,:]}\|_2 \nonumber\\ &=\frac{1}{\sqrt{k}}\beta\max_{i=1,...,p} \|W_{[i,:]}\|_2 = \Delta_2.
\end{align*}

Therefore, for the sensitivity analysis, the ``$\beta$-adjacency'' in Definition~\ref{def:neighbor} in fact covers the definition of neighboring based on the $l_1$ difference.

\vspace{0.1in}

Once the sensitivities can be identified/computed, the following two standard noise addition mechanisms can be applied to make an algorithm differentially private.

\begin{theorem}[Laplace mechanism~\citep{dwork2006calibrating}] \label{prop:laplace mechanism}
Given any function $f:\mathcal U\mapsto \mathbb R^k$ and $\epsilon>0$, the randomized algorithm $\mathcal M(u)=f(u)+(Z_1,...,Z_k)$ is $\epsilon$-DP, where $Z_1,...,Z_k$ follow iid $Laplace(\Delta_1/\epsilon)$. The centered Laplace distribution with scale $\lambda$ has density function $g(x)=\frac{1}{2\lambda}e^{-\frac{|x|}{\lambda}}$.
\end{theorem}

\vspace{0.1in}

\begin{theorem}[Gaussian mechanism~\citep{dwork2014algorithmic}] \label{prop:gaussian mechanism}
Given any function $f:\mathcal U\mapsto \mathbb R^k$, $0<\epsilon<1$, $0<\delta<1$, the randomized algorithm $\mathcal M(u)=f(u)+(Z_1,...,Z_k)$ is $(\epsilon,\delta)$-DP, where $Z_1,...,Z_k$ follow iid $N(0,\sigma^2)$ with $\sigma=\Delta_2\frac{\sqrt{2\log(1.25/\delta)}}{\epsilon}$.
\end{theorem}

We remark that Theorem~\ref{prop:gaussian mechanism} requires $\epsilon<1$. The behavior of the Gaussian noise for general $\epsilon$ will be discussed shortly in Section~\ref{sec:DP-RP}. Furthermore, one of the nice facts about DP is that, multiple DP algorithms can be composed together, where the final output is still DP.

\vspace{0.1in}

\begin{theorem}[Composition Theorem~\citep{dwork2014algorithmic}]\label{theo:composition}
Let $\mathcal M_j:\mathcal U\rightarrow \mathbb R$ be an $(\epsilon_j,\delta_j)$-DP algorithm for $j=1,...,k$. Then $\mathcal M(u)=(\mathcal M_1(u), ..., \mathcal M_k(u))$ is $(\sum_{j=1}^k \epsilon_i,\sum_{j=1}^k \delta_j)$-DP.
	
\end{theorem}

\subsection{Local Sensitivity and Smooth Sensitivity}

The sensitivity in Definition~\ref{def:sensitivity} is a ``global'' worst-case bound over all possible neighboring data $(u,u')$ in the data domain, which in many cases leads to overly large noise added to the algorithm output. Intuitively, for each specific data $u$, to make $\mathcal M(u)$ and $\mathcal M(u')$ indistinguishable in the sense of DP, we may not always need the worst-case noise. This is grasped by the concept of ``local sensitivity'' as defined below.

\begin{definition}[Local sensitivity~\citep{nissim2007smooth}] \label{def:local-sensitivity}
For $u\in\mathcal U$ and a function $f:\mathcal U\mapsto \mathbb R^k$, the local sensitivities of $f$ at $u$ are
\begin{align*}
   LS_1(u)=\max_{u'\in Nb(u)} \|f(u)-f(u')\|_1,\quad  LS_2(u)=\max_{u'\in Nb(u)} \|f(u)-f(u')\|_2.
\end{align*}
\end{definition}

Instead of bounding the worst-case change in the output, the local sensitivity at $u$ just focuses on the neighborhood of each specific $u$. By the definitions, it is clear that $LS_1(u)\leq \Delta_1$ and $LS_2(u)\leq \Delta_2$ for all $u\in\mathcal U$, i.e., the global sensitivity is an upper bound of local sensitivity. That said, the noise level calibrated to the local sensitivity is always smaller than the noise magnitude calculated by the global sensitivity, which implies (substantially) better utility.

One however cannot simply add noise according to the instance-specific local sensitivity to achieve DP, for several reasons. (A) Only ensuring ``local'' indistinguishability at each specific data instance does not satisfy the rigorous DP condition (Definition~\ref{def:DP}), which requires the indistinguishability between all possible adjacent data pairs. (B) The noise level itself, when computed by the local sensitivity, is a function of $u$, which may reveal the information of the database. Intuitively, the issue of local sensitivity is that, although $u$ has small local sensitivity, its neighbor $u'$ may have very large local sensitivity. If the injected noise is determined according to the local sensitivity, then the very different noise scale may be detected by an adversary to distinguish $u$ from $u'$, violating the goal of DP. A simple example is provided in~\citet{nissim2007smooth} for reporting the median. In that same paper, the authors proposed ``smooth sensitivity'' which leverages local sensitivity in the noise calibration, at the same time also satisfies the DP definition. The intuition is that, whatever “local” notion of sensitivity at $u$ we employ must account for the more sensitive data which are nearby, but at a discounted factor depending on how far they are from $u$.

\begin{definition}[Smooth sensitivity~\citep{nissim2007smooth}] \label{def:smooth-sensitivity}
For $u,u'\in\mathcal U$, denote $d(u,u')$ as the minimal number of ``neighboring operations'' needed to derive $u'$ from $u$~\footnote{For example, if $u$ and $u'$ are two databases that differ in $t$ sample records, then $d(u,u')=t$.}. The $\epsilon$-smooth sensitivity of $f:\mathcal U\mapsto \mathbb R^k$ at $u$ is defined as
\begin{align*}
    SS(u)=\max_{u'\in\mathcal U} e^{-\epsilon d(u,u')}LS(u'),
\end{align*}
where $LS(u')$ is the local sensitivity (in proper metric) of $f$ at $u'$ in Definition~\ref{def:local-sensitivity}.
\end{definition}

The smooth sensitivity $SS(u)$ is an $\epsilon$-smooth upper bound of the local sensitivity $LS(u)$ satisfying two conditions: (i) $SS(u)\geq LS(u)$, $\forall u\in \mathcal U$; (ii) $SS(u)\leq e^\epsilon SS(u')$, $\forall u,u'\in\mathcal U$ with $d(u,u')=1$. Moreover, $SS(u)$ is the smallest upper bound $\forall u$ among all such $\epsilon$-smooth upper bounds\footnote{The global sensitivity (Definition~\ref{def:sensitivity}) is also an $\epsilon$-smooth upper bound of the local sensitivity.}. \citet{nissim2007smooth} also proposed several noise addition mechanisms tailored to the (instance-specific) smooth sensitivity, with rigorous DP guarantees. In particular, the smooth sensitivity usually works very well for robust statistics/estimators (e.g., the median, the interquartile range), in which the local sensitivity in the neighborhood of most data does not change much.

For DP-RP studied in our paper, it is easy to check that both the local sensitivity and smooth sensitivity, at any $u\in [-1,1]^p$, are the same as the global sensitivity by (\ref{eq:l1-sensitivity-known}) and (\ref{eq:l2-sensitivity-known}). Thus, for DP-RP, these two concepts may not help much. Fortunately, for DP-SignRP, the ``aggregate-and-sign'' operation brings robustness against small data changes. Therefore, we can leverage the idea of local or smooth sensitivity to develop novel and better DP algorithms than the standard bit/sign flipping approach (Section~\ref{sec:DP-SignRP-RR-smooth} and Section~\ref{sec:individual DP}).

\subsection{Relaxation: Individual Differential Privacy (iDP)}

Besides the smooth sensitivity framework, many extensions or relaxations of DP have also been proposed to improve the utility of DP mechanisms. Examples include concentrated differential privacy~\citep{dwork2016concentrated}, R{\'{e}}nyi differential privacy~\citep{mironov2017renyi}, and Gaussian differential privacy~\citep{dong2022gaussian}. These alternatives provide less strict requirement on the ``indistinguishability'' (\ref{eq:DP-def}) for privacy and better composition properties than those of the standard ``worst-case'' DP, thus reducing the noise needed~\citep{jayaraman2019evaluating}. Another possible direction to elevate the empirical performance of DP is to relax the DP definition by constraining the scope of neighboring datasets depending on the specific use case of DP~\citep{comas2017individual,desfontaines2020sok}. In this paper, we also consider the concept called ``individual differential privacy'' (iDP), also known as ``data-centric DP'', as follows.

\begin{definition}[Individual DP~\citep{comas2017individual}] \label{def:idp}
Given a dataset $u$, an algorithm $\mathcal M$ satisfies $(\epsilon,\delta)$-iDP for $u$ if for any dataset $u'$ that is adjacent to $u$, it holds that
\begin{align*}
    Pr[\mathcal M(u)\in O] \leq e^\epsilon Pr[\mathcal M(u')\in O]+\delta, \\
    Pr[\mathcal M(u')\in O] \leq e^\epsilon Pr[\mathcal M(u)\in O]+\delta.
\end{align*}
\end{definition}

We should emphasize that, individual DP does not satisfy the rigorous DP definition, as iDP only focuses on the ``point-wise'' guarantee of privacy. It protects the neighborhood of a specific dataset of interest, instead of fulfilling DP requirements for all possible adjacent databases. While deviating from the standard DP definition, we discuss it in our work because iDP might be sufficient for some applications, and the utility gain of iDP can be substantial.

\vspace{0.1in}

The intuition of iDP is that, while the standard DP (Definition~\ref{def:DP}) requires indistinguishability between any pair of neighboring databases, in some practical scenarios, the data custodian only holds one ``ground truth'' data vector $u$ that needs to be protected. Limiting the scope of the neighborhood could be reasonable in certain practical scenarios. For example, for the centralized data server to (non-interactively) publish the user data matrix or user embedding vectors to the public, we essentially care about preserving the privacy of the true user data at hand, as no information other than the perturbed database would be released. In this paper, while we will mainly focus on the standard DP definition, we also find that using iDP for sign-based RP methods provides good utilities. Therefore, we will also discuss iDP-based variants in Section~\ref{sec:individual DP}, as it may provide another direction/option for balancing the trade-off between privacy and utility in practice, based on specific applications.

\newpage

\section{DP-RP: Differentially Private Random Projections}  \label{sec:DP-RP}

In this section, we study the differentially private random projections (DP-RP). We first revisit the noise addition mechanisms reported in~\citet{kenthapadi2013privacy} for Gaussian random projections, an algorithm which we name ``DP-RP-G''. Then we apply the optimal Gaussian mechanism~\citep{balle2018improving} to DP-RP to obtain an improved algorithm named ``DP-RP-G-OPT''. In Section~\ref{sec:experiment}, our experiments confirm that DP-RP-G-OPT considerably improves DP-RP-G. Moreover, by replacing the Gaussian projection matrix with the Rademacher projection matrix, we obtain a further improved version named ``DP-RP-G-OPT-B''. Furthermore, one can obtain another variant called ``DP-OPORP'' by replacing the dense Rademacher  projections with OPORP, which is a variant of count-sketch and is computationally much more efficient than dense projections.

We also analytically compare those variants of DP-RP with the basic DP strategy of directly adding Gaussian noise to each dimension of the original data vectors. Our analysis  explains the improvements by using  DP-RP in high-dimensional data.

\subsection{Gaussian Noise Mechanism for DP-RP}  \label{sec:gaussian-known}

\begin{algorithm}[h]
	{
		\textbf{Input:} Data $u\in[-1,1]^p$, privacy parameters $\epsilon>0$, $\delta\in (0,1)$, number of projections~$k$
		
		\textbf{Output:}  $(\epsilon,\delta)$-differentially private random projections $\tilde{x}\in\mathbb{R}^k$
		
		Apply RP  $x=\frac{1}{\sqrt{k}}W^Tu$, where  $W\in\mathbb R^{p\times k}$ has iid entries sampled from $N(0,1)$
		
		Compute the sensitivity $\Delta_2$ by (\ref{eq:l2-sensitivity-known})
		
		Generate the random noise vector $G\in\mathbb R^k$ whose entries are iid samples from $N(0,\sigma^2)$ where $\sigma$ is obtained by Theorem~\ref{theo:kenthapadi} (DP-RP-G) or Theorem~\ref{theo:gauss-optimal} (DP-RP-G-OPT)

		Return  $\tilde x =x+G$
	}
	\caption{DP-RP-G and DP-RP-G-OPT}
	\label{alg:DP-RP-gaussian-noise}
\end{algorithm}

The general Gaussian noise mechanism for DP-RP is detailed in Algorithm~\ref{alg:DP-RP-gaussian-noise}. We present the algorithms for a single data point $u$, and the procedure is applied to every data vector in the database (with same $W$ but independent noise). Given the projection matrix $W\in \mathbb R^{d\times k}$, we first compute the $l_2$-sensitivity $\Delta_2$ of the RPs according to (\ref{eq:l2-sensitivity-known}). Then we generate the random Gaussian noise vector following iid $N(0,\sigma^2)$ where $\sigma$ is derived by Theorem~\ref{theo:kenthapadi} below, and output the projected vector plus the noise. The following Gaussian mechanism for random projections is known in the literature. We call it the ``DP-RP-G'' method.

\vspace{0.1in}
\begin{theorem}[DP-RP-G~\citep{kenthapadi2013privacy}] \label{theo:kenthapadi}
Let $\Delta_2$ be defined in (\ref{eq:l2-sensitivity-known}). For any $\epsilon>0$ and $0<\delta<\frac{1}{2}$, DP-RP-G in Algorithm~\ref{alg:DP-RP-gaussian-noise} is $(\epsilon,\delta)$-DP if $\sigma\geq \Delta_2 \frac{\sqrt{2(\log(1/\delta)+\epsilon)}}{\epsilon}$.
\end{theorem}

According to Theorem~\ref{theo:kenthapadi}, one will choose $\sigma$ based on the observed $\Delta_2$ which is calculated according to~\eqref{eq:l2-sensitivity-known} for a given realization of the projection matrix $W$. For the convenience of theoretical analysis, it is often desirable to derive the criterion for choosing $\sigma$ directly in terms of the input parameters $\epsilon$, $\delta$, $p$, and $k$. We call the corresponding algorithm to be ``Analytic DP-RP-G''. This is achieved by bounding $\Delta_2$ in high probability.

\vspace{0.1in}
\begin{lemma}[\citet{laurent2000adaptive}]\label{lemma:chi-square tail}
Let $Y_1,...,Y_n$ be iid standard Gaussian random variables, and denote $Z_i=\sum_{i=1}^n Y_i^2$. Then for any $t>0$,
\begin{align*}
    Pr(Z\geq n+2\sqrt{nt}+2t)\leq \exp(-t).
\end{align*}
\end{lemma}

\begin{lemma}[Bounding $\Delta_2$] \label{lemma:l2 sensitivity bound}
Let $W\in \mathbb R^{p\times k}$ be a random matrix with iid $N(0,1)$ entries, and $\Delta_2$ be defined in (\ref{eq:l2-sensitivity-known}). Then for any $0<\delta<1$,
\begin{align*}
    Pr\left(\Delta_2\geq \beta\sqrt{1+2\sqrt{\log(p/\delta)/k}+2\log(p/\delta)/k} \right) \leq \delta.
\end{align*}
\end{lemma}
\begin{proof}
Let $e_i$ be the unit base vector with the $i$-th dimension being $1$ and all other dimensions being zero. For each $i=1,...,p$, $W^T e_i$ follows iid chi-square distribution with $k$ degree of freedom. Applying Lemma~\ref{lemma:chi-square tail} we have that for any $\delta'>0$,
\begin{align*}
    Pr\left(\|W^T e_i\|^2\geq k+2\sqrt{k\log(1/\delta')}+2\log(1/\delta') \right)\leq \delta',\ \ \forall i=1,...,p.
\end{align*}
Applying union bound gives
\begin{align*}
    Pr\left(\frac{1}{k}\beta^2\max_{i=1,...,p} \|W^T e_i\|^2\geq \beta^2\left(1+2\sqrt{\log(1/\delta')/k}+2\log(1/\delta')/k \right) \right)\leq p\delta'.
\end{align*}
Setting $\delta=p\delta'$ and taking the square root on both sides completes the proof.
\end{proof}

\vspace{0.1in}

With the high probability bound on the $l_2$-sensitivity $\Delta_2$, we have the following result.

\begin{theorem}[Analytic DP-RP-G]\label{theo:analytic-DP-RP-G}
For any $\epsilon>0$ and $0<\delta<\frac{1}{2}$, Algorithm~\ref{alg:DP-RP-gaussian-noise} achieves $(\epsilon,\delta)$-DP when $\sigma\geq \frac{ \beta\sqrt{1+2\sqrt{\log(2p/\delta)/k}+2\log(2p/\delta)/k} \sqrt{2(\log(2/\delta)+\epsilon)} }{\epsilon}$.
\end{theorem}
\begin{proof}
By Lemma~\ref{lemma:l2 sensitivity bound}, with probability $\delta/2$, $\Delta_2$ is bounded by $\beta\sqrt{1+2\sqrt{\log(2p/\delta)/k}+2\log(2p/\delta)/k}$. In this event, by Theorem~\ref{theo:kenthapadi}, Algorithm~\ref{alg:DP-RP-gaussian-noise} is $(\epsilon,\delta/2)$-DP when $\sigma\geq \Delta_2\frac{\sqrt{2(\log(2/\delta)+\epsilon)}}{\epsilon}$. Therefore, the approach is $(\epsilon,\delta)$-DP.
\end{proof}

We remark that, since this analytic version is built upon a high probability bound on the sensitivity, in most cases, the noise variance in Theorem~\ref{theo:analytic-DP-RP-G} will be slightly larger than the one computed by the realized/exact sensitivity in Theorem~\ref{theo:kenthapadi}.




\subsection{The Optimal Gaussian Noise Mechanism for DP-RP}

In Theorem~\ref{theo:kenthapadi} (as well as the classical Gaussian mechanism, e.g.,~\citet{dwork2014algorithmic}), the analysis on the noise level is based on upper bounding the tail of Gaussian distribution\footnote{More specifically, the analysis of the classic Gaussian mechanism is based on bounding the probability of $X=\log\frac{p_{\mathcal M(u)}(y)}{p_{\mathcal M(u')}(y)}\geq \epsilon$, where $u$ and $u'$ are neighboring databases and $y$ is an output. $X$, called the privacy loss random variable, is shown to follow a Gaussian distribution~\citep{dwork2016concentrated} when the injected noise is Gaussian. The tail probability of $X$ is then bounded by standard concentration inequalities.}. While the noise level in Theorem~\ref{theo:kenthapadi} is rate optimal (see Remark~\ref{remark:optimal}), it can be  improved by computing the exact tail probabilities instead of the upper bounds. \citet[Theorem 8]{balle2018improving} proposed the optimal Gaussian mechanism based on this idea, using which allows us to improve Theorem~\ref{theo:kenthapadi} and obtain the optimal DP-RP-Gaussian-Noise (DP-RP-G-OPT) method as follows.

\vspace{0.1in}

\begin{theorem}[DP-RP-G-OPT]\label{theo:gauss-optimal}
Suppose $\Delta_2$ is defined as (\ref{eq:l2-sensitivity-known}). For any $\epsilon>0$ and $0<\delta<1$, DP-RP-G-OPT in Algorithm~\ref{alg:DP-RP-gaussian-noise} achieves $(\epsilon,\delta)$-DP if $\sigma\geq \sigma^*$ where $\sigma^*$ is the solution to the equation
\begin{align}
    \Phi\left(\frac{\Delta_2}{2\sigma}-\frac{\epsilon \sigma}{\Delta_2}\right) - e^\epsilon \Phi\left(-\frac{\Delta_2}{2\sigma}-\frac{\epsilon \sigma}{\Delta_2}\right)= \delta,  \label{eqn:optimal DP-RP}
\end{align}
where $\Phi(\cdot)$ is the cumulative distribution function (cdf) of the standard normal distribution.
\end{theorem}

\vspace{0.1in}

\begin{remark} \label{remark:optimal}
For both Theorem~\ref{theo:gauss-optimal} and Theorem~\ref{theo:kenthapadi}, the Gaussian noise level is $\sigma = \mathcal O(\frac{\Delta_2}{\epsilon})= \mathcal O(\frac{1}{\epsilon})$ when $\epsilon$ is small (i.e., $\epsilon\rightarrow 0$), and $\sigma=\mathcal O(\frac{1}{\sqrt\epsilon})$ when $\epsilon$ is large (i.e., $\epsilon\rightarrow\infty$). This is rate optimal as Gaussian mechanism must incur a noise $\sigma=\Omega(\frac{\Delta_2}{\sqrt\epsilon})$ when $\epsilon\rightarrow\infty$ \citep{balle2018improving}. Theorem~\ref{theo:gauss-optimal} reduces the variance of Theorem~\ref{theo:kenthapadi} by analyzing the tail probability exactly.
\end{remark}

\begin{figure}[h]
    \centering
\mbox{
    \includegraphics[width=3.2in]{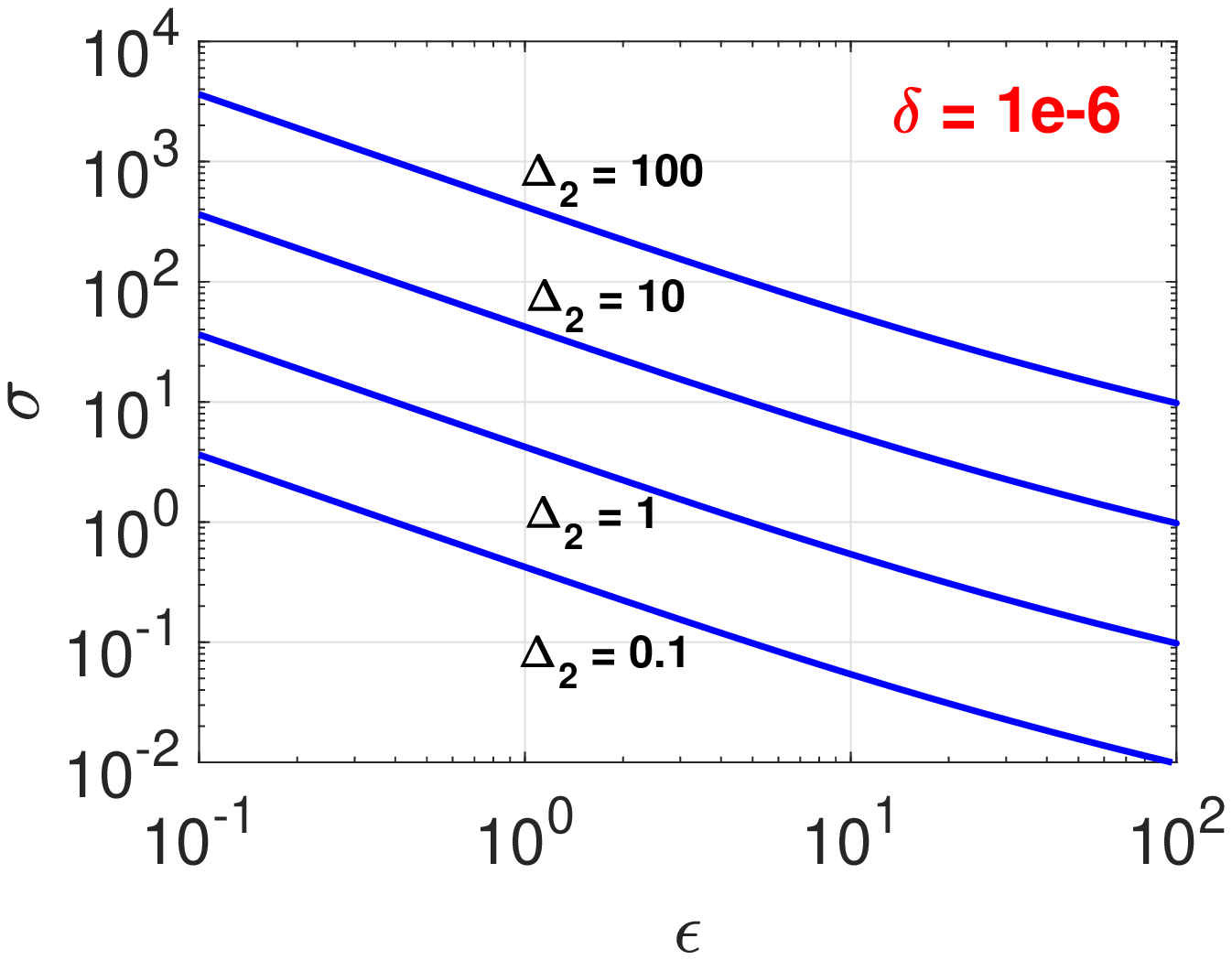}
    \includegraphics[width=3.2in]{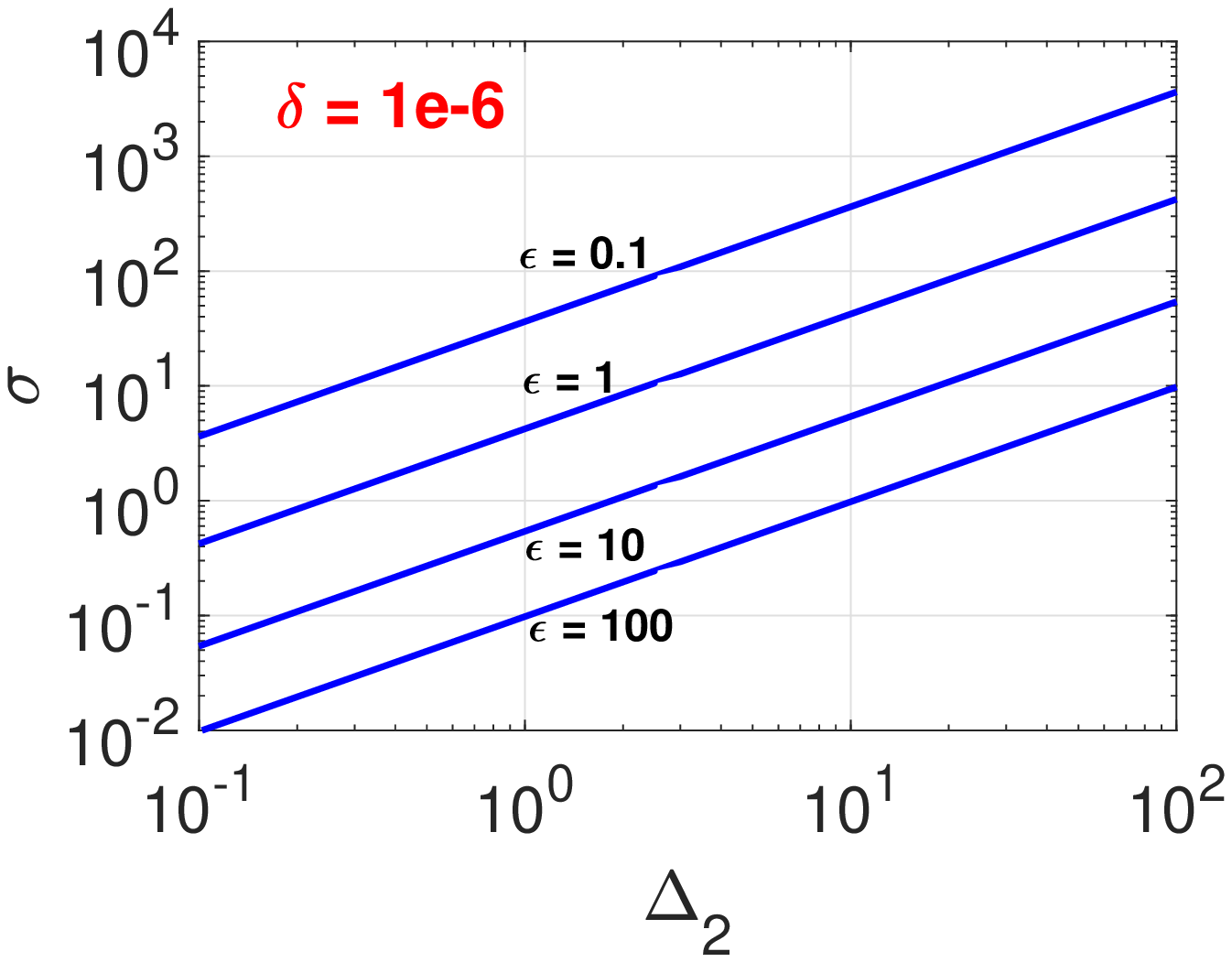}
}

\vspace{-0.1in}

    \caption{Left panel: the optimal Gaussian noise $\sigma$ versus $\epsilon$ for a series of $\Delta_2$ values, by solving the nonlinear equation~\eqref{eqn:optimal DP-RP} in Theorem~\ref{theo:gauss-optimal}, for $\delta = 10^{-6}$. Right panel: the  optimal Gaussian noise $\sigma$ versus $\Delta_2$ for a series of $\epsilon$ values. }
    \label{fig:optimal DP-RP}
\end{figure}

\subsection{The Laplace Noise Mechanism for DP-RP}  \label{sec:DP-RP laplace}

\begin{algorithm}[h]
	{
		\textbf{Input:} Data $u\in[-1,1]^p$, privacy parameters $\epsilon>0$, number of projections~$k$
		
		\textbf{Output:} $\epsilon$-differentially private random projections $\tilde{x}\in\mathbb{R}^k$
		
		Apply RP $x=\frac{1}{\sqrt{k}}W^Tu$, where $W\in\mathbb R^{p\times k}$ is a random $N(0,1)$ matrix
		
		Compute sensitivity $\Delta_1$ by (\ref{eq:l1-sensitivity-known})
		
		Generate iid random noise vector $L\in\mathbb R^k$ following $Laplace(\frac{\Delta_1}{\epsilon})$

		Return  $\tilde x =x+L$
	}
	\caption{DP-RP-L}
	\label{alg:DP-RP-laplace-noise}
\end{algorithm}

By Theorem~\ref{prop:laplace mechanism}, one may add $Laplace(\lambda)$ noise to the projected data with $\lambda=\Delta_1/\epsilon$ to achieve $\epsilon$-DP. For completeness, we detail the approach in Algorithm~\ref{alg:DP-RP-laplace-noise}, which is the same as Algorithm~\ref{alg:DP-RP-gaussian-noise} except for the noise added. $\Delta_1$ is the $l_1$-sensitivity for a given $W$ given in (\ref{eq:l1-sensitivity-known}): $\Delta_1=\frac{1}{\sqrt{k}}\beta\max_{i=1,...,p} \|W_{[i,:]}\|_1$. Each term in the maximum operator is a sum of $k$ iid standard half-normal random variables ($Y=|X|$ where $X$ is standard normal). The readers might be interested in a comparison between the noise level of the two noise addition mechanisms, i.e., DP-RP-G and DP-RP-L. To observe the dependence of the noise level $\lambda$ on $\epsilon$, $k$ and $p$, we derive a tail bound on the $l_1$-sensitivity. We first establish the following lemma.

\begin{lemma}[Half-normal tail bound] \label{lemma:half-normal-tail}
Let $X_1,...,X_n$ be iid $N(0,1)$ random variables and let $Y_i=|X_i|$, $\forall i$. Denote $Z=\sum_{i=1}^n Y_i$. Then for any $t>0$, it holds that
\begin{align*}
    Pr(Z\geq \sqrt{2n^2\log 2+2nt})\leq \exp(-t).
\end{align*}
\end{lemma}

\vspace{0.1in}

Using Lemma~\ref{lemma:half-normal-tail}, we obtain a high probability bound on the $l_1$-sensitivity as follows. The proof is omitted since it is similar to the proof of Lemma~\ref{lemma:l2 sensitivity bound}.

\begin{lemma}[Bounding $\Delta_1$] \label{lemma:l1 sensitivity bound}
Let $W\in \mathbb R^{p\times k}$ be a random matrix with iid $N(0,1)$ entries, and $\Delta_1$ be defined in (\ref{eq:l1-sensitivity-known}). Then for any $0<\delta<1$,
\begin{align*}
    Pr\left(\Delta_1\geq \beta\sqrt{2k\log 2+2\log(p/\delta)} \right) \leq \delta.
\end{align*}
\end{lemma}

It is important to point out that, unlike in the Gaussian mechanism, simply inserting the upper bound in Lemma~\ref{lemma:l1 sensitivity bound} to replace the (exact) $\Delta_1$ in the noise level of $Laplace(\Delta_1/\epsilon)$ in Algorithm~\ref{alg:DP-RP-laplace-noise} does not provide pure $\epsilon$-DP. In other words, we must compute the actual $\Delta_1$ by (\ref{eq:l1-sensitivity-known}) for a realization of the projection matrix $W$. This is because, using the high probability bound of $\Delta_1$ would introduce a failure probability (i.e., when the actual $\Delta_1$ is larger than the bound) which violates $\epsilon$-DP.

\vspace{0.1in}
\noindent\textbf{Comparison with DP-RP-G-OPT.} Lemma~\ref{lemma:l1 sensitivity bound} provides an estimate of the noise level required by the Laplace mechanism. We can roughly compare DP-RP-L with DP-RP-OPT in terms of the noise variances, which equal $\sigma^2$ for Gaussian and $2\lambda^2$ for Laplace, respectively. Lemma~\ref{lemma:l1 sensitivity bound} shows that in most cases, the Laplace noise requires $\lambda=\mathcal O(\frac{\sqrt k}{\epsilon})$. Recall that in DP-RP-G and DP-RP-G-OPT, the Gaussian noise level $\sigma=\mathcal O(\frac{1}{\epsilon})$ when $\epsilon\rightarrow 0$ and $\sigma=\mathcal O(\frac{1}{\sqrt\epsilon})$ when $\epsilon\rightarrow \infty$. Thus, the Gaussian noise would be smaller with relatively small $\epsilon$, and the Laplace noise would be smaller only when $\epsilon=\Omega(k)$, which is usually too large to be useful since in practice $k$ is usually hundreds to thousands. As a result, we would expect DP-RP-G methods to perform better than DP-RP-L in common scenarios. Therefore, in the remaining parts of the paper, we will mainly focus on the Gaussian mechanism for noise addition.

\subsection{DP-RP-G-OPT-B: DP-RP with Rademacher Projections}  \label{sec:DP-RP Rade}

We have considered the Gaussian random projections. We can also adopt other types of projection matrices which might even work better for DP.  The following distributions of $w_{ij}$ are popular:

\begin{itemize}

\item The uniform distribution, $\sqrt{3}\times unif[-1,1]$. The $\sqrt{3}$ factor is placed here to have $\mathbb E(w_{ij}^2)=1$ by following the convention in the practice of random projections.

\item The ``very sparse'' distribution, as used in~\citet{li2006very}:
\begin{align}\label{eqn:vsrp}
w_{ij} = \sqrt{s}\times \left\{\begin{array}{rrl}
-1 & \text{with prob.} &1/(2s)\\
0 & \text{with prob.} &1-1/s,\\
+1 & \text{with prob.} &1/(2s)
\end{array}\right.
\end{align}
which generalizes~\citet{achlioptas2003database} (for $s=1$ and $s=3$). Note that when $s=1$, it is also called the ``symmetric Bernoulli'' distribution or the ``Rademacher'' distribution.
\end{itemize}

From Theorem~\ref{theo:kenthapadi} and Theorem~\ref{theo:gauss-optimal}, it is clear that the noise magnitude of Gaussian noise in DP-RP directly depends on the $l_2$-sensitivity $\Delta_2$, which, according to (\ref{eq:l2-sensitivity-known}), equals the largest row norm of the projection matrix $W$. Among the above-mentioned distributions, the dense Rademacher projection ($s=1$ in (\ref{eqn:vsrp})) has $\Delta_2=\frac{1}{\sqrt k}\beta\times \sqrt{k}=\beta$ which is independent of $p$. This could be much smaller than the dense Gaussian projection (i.e., DP-RP-G-OPT) as can be seen from the bound in Lemma~\ref{lemma:l2 sensitivity bound}).

\newpage

\begin{figure}[h]
    \centering
    \includegraphics[width=3.2in]{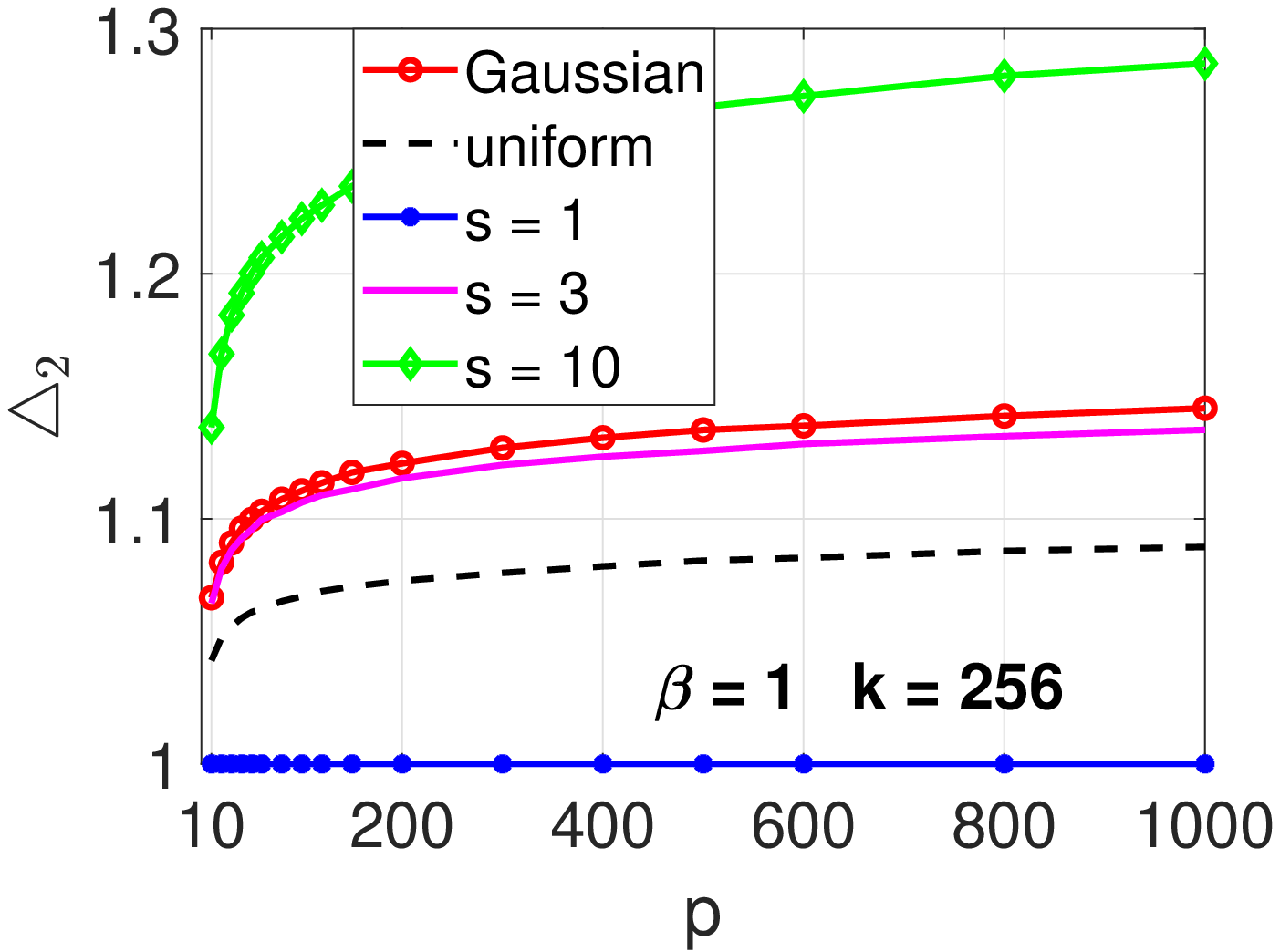}
    \includegraphics[width=3.2in]{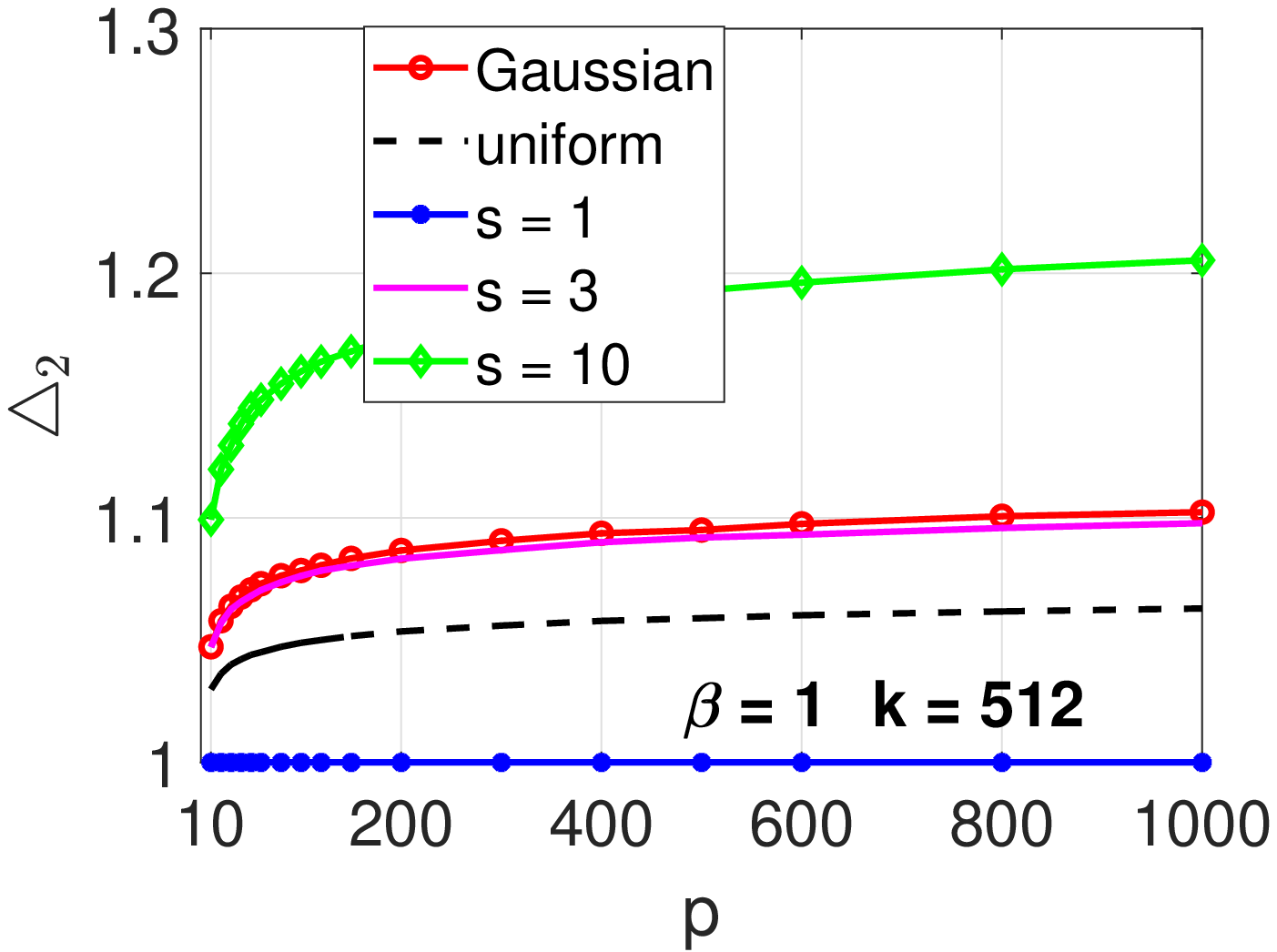}

\vspace{-0.1in}

    \caption{The $l_2$-sensitivity $\Delta_2$ (\ref{eq:l2-sensitivity-known}) for different types of random projection matrices against the data dimensionality $p$, at $k=256$ and $k=512$, respectively. $\beta=1$. }
    \label{fig:row-norm}
\end{figure}

In Figure~\ref{fig:row-norm}, we numerically simulate the $\Delta_2$ of different projection matrices, which shows that the Rademacher projection produces the smallest sensitivity. This, when plugged into the optimal Gaussian mechanism (Theorem~\ref{theo:gauss-optimal}), leads to smaller Gaussian noise variance needed. For clarity, we summarize the DP-RP based on Rademacher projections in Algorithm~\ref{alg:DP-RP-G-OPT-B}. We name it ``DP-RP-G-OPT-B'', where ``G'' stands for the \textbf{G}aussian noise mechanism and ``B'' stands for ``symmetric \textbf{B}ernoulli'' projections.

\begin{algorithm}[h]
	{
		\textbf{Input:} Data $u\in[-1,1]^p$, privacy parameters $\epsilon>0$, $\delta\in (0,1)$, number of projections~$k$
		
		\textbf{Output:} $(\epsilon,\delta)$-differentially private random projections $\tilde x\in \mathbb R^k$
		
		Apply RP $x=\frac{1}{\sqrt{k}}W^Tu$, where $W\in\mathbb R^{p\times k}$ is a random Rademacher matrix
		
		Generate iid random noise vector $G\in\mathbb R^k$ following $N(0,\sigma^2)$ where $\sigma$ is obtained by Theorem~\ref{theo:gauss-optimal} with $\Delta_2=\beta$

		Return  $\tilde x =x+G$
	}
	\caption{DP-RP-G-OPT-B}
	\label{alg:DP-RP-G-OPT-B}
\end{algorithm}

Does the reduction in noise magnitude translate into better utility? We first provide an approximate assessment as follows. Note that when the data $u\in [-1,1]^p$ is not too sparse, the distribution of $\sum_{i=1}^p w_iu_i$, where $w_i$ follows Rademacher distribution, is asymptotically normal according to the Central Limit Theorem (CLT), with mean $0$ and variance $\|u\|^2$. This distribution is the same as the distribution of $\sum_{i=1}^p w_iu_i$ when $w_i$ is from $N(0,1)$. Therefore, we may expect the projected data of Gaussian RP and Rademacher RP to have similar magnitude. Hence, the signal-to-noise ratio (SNR) of DP-RP-G-OPT-B using Rademacher projections would be higher than that of DP-RP-G-OPT using Gaussian projections, as a result of the smaller noise level. An analytical comparison in terms of inner product estimation variance will be given in Section~\ref{sec:variance comparison}.

\subsection{DP-OPORP: A More Efficient Alternative}

The standard random projection (\ref{def:RP}) requires $k$ projections which takes $O(kp)$ complexity. We can reduce the cost to $O(p)$ by the idea of ``binning + RP''. Basically, we split the data entries into $k$ bins, and apply RP in each bin to generate $k$ samples. \cite{zhao2022differentially} studied noise injection mechanism for count-sketch. Recently, \citet{li2023oporp} proposed OPORP (One Permutation + One Random Projection), an improved variant of count-sketch~\citep{charikar2004finding}, by using fixed-length binning and normalization. The steps of OPORP for a data vector $u$ are summarized in Algorithm~\ref{alg:OPORP} (for convenience, assume $p$ can be integer divided by $k$).

\begin{algorithm}[h]
	{
		\textbf{Input:} Data vector $u\in\mathbb R^p$, number of projected samples~$k$
		
		\textbf{Output:} $k$ OPORP samples
		
		Apply a permutation $\Pi:[p]\mapsto [p]$ to $u$ to get $u_{\Pi}$
		
		Split the $p$ permuted data columns (features) into $k$ consecutive fixed-length bins each containing $p/k$ features: $u_\Pi=[u_\Pi^{(1)},...,u_\Pi^{(k)}]$
		
		Generate one projection vector $w\in\mathbb R^p$ following a proper probability distribution (see Section~\ref{sec:DP-RP Rade}). Denote it as a concatenation of bins $w=[w^{(1)},...,w^{(k)}]$

		Return $k$ projected samples by $x_i={w^{(i)}}^T u_\Pi^{(i)}$, for $i=1,...,k$.
	}
	\caption{OPORP: count-sketch with fixed-length binning and normalization.}
	\label{alg:OPORP}
\end{algorithm}

Note that the same permutation $\Pi$ and projection vector $w$ should be used to process all the data vectors. Let $x_i$ be an OPORP of $u$, and $y_i$ be an OPORP of $v$. It can be shown that $\mathbb E[x_iy_i]=u^T v$, i.e., OPORP provides an unbiased estimator of the inner product. Regarding the estimation variance, note that in step 2, OPORP deploys a fixed-length binning scheme; in prior works, the lengths of the bins are random and follow a multinomial distribution. \citet{li2023oporp} showed that fixed-length binning can reduce the inner product and cosine estimation variance of variable-length binning, by a factor of $\frac{p-k}{p-1}$, when the projection vector follows the symmetric Bernoulli (Rademacher) distribution (see (\ref{eqn:vsrp})). Furthermore, the estimation variance can be substantially reduced by normalizing the output vector of OPORP. It was also shown by~\citet{li2023oporp} that one can substantially improve the estimates of ``very sparse random projections'' via the same normalization technique.

\vspace{0.1in}

Next, we consider the differential privacy of  OPORP and propose DP-OPORP analogously, as given in Algorithm~\ref{alg:DP-OPORP}. Since we have shown that Gaussian noise would be better than Laplace noise in most cases for DP-SignRP (see Section~\ref{sec:DP-RP laplace}), we present the optimal Gaussian noise mechanism here for conciseness. Also, while one may choose any projection vector $w$ in OPORP (Algorithm~\ref{alg:OPORP}), \citet{li2023oporp} showed that the Rademacher projection gives the smallest cosine and inner product estimation variance. Hence, we adopt the Rademacher projection in DP-OPORP.

\begin{algorithm}[h]
	{
		\textbf{Input:} Data $u\in[-1,1]^p$, privacy parameters $\epsilon>0$, $\delta\in (0,1)$, number of projections~$k$
		
		\textbf{Output:}   Differentially private OPORP
		
		Apply Algorithm~\ref{alg:OPORP} with a random Rademacher projection vector to obtain the OPORP $x$
		
		Set sensitivity $\Delta_2=\beta$
		
		Generate iid random vector $G\in\mathbb R^k$ following $N(0,\sigma^2)$ where $\sigma$ is computed by Theorem~\ref{theo:gauss-optimal}

		Return  $\tilde x =x+G$
	}
	\caption{DP-OPORP}
	\label{alg:DP-OPORP}
\end{algorithm}

In Algorithm~\ref{alg:DP-OPORP}, we formalize the DP-OPORP method. It is a variant of DP-RP-G-OPT-B, where Gaussian noise is added to the OPORP. Note that the sensitivity $\Delta_2=\beta$, due to the nature of OPORP: changing the data $u$ by $\beta$ (Definition~\ref{def:neighbor}) would lead to at most $\beta$ change in $x$ in term of $l_2$ distance.  Note that, if we directly add noise to each dimension of the original data, then the sensitivity by definition  $\Delta_2=\max_{(u,u')\in \mathcal N(\mathcal U)}\|f(u)-f(u')\|_2$ is also $\beta$.

\subsection{Comparison: Inner Product Estimation}   \label{sec:variance comparison}

We now analytically compare different DP algorithms, in terms of inner product estimation. Here for simplicity, we assume the data are normalized, i.e., the data vector has $l_2$ norm equal to 1. In this case, the inner product is also the cosine. The baseline method is the most straightforward: we add optimal Gaussian noise to each dimension of the original data (Raw-data-G-OPT). As mentioned earlier, the sensitivity is also $\Delta_2=\beta$. This means, when we compare all three methods: Raw-data-G-OPT, DP-RP-G-OPT-B, and DP-OPORP, the noise level $\sigma$ is the same. This makes it convenient to conduct the comparisons, from which we can gain valuable insights.

\begin{theorem}[Raw-data-G-OPT, i.e., adding optimal Gaussian noise on raw data]  \label{theo:original-data-inner}
Let $\sigma$ be the solution to (\ref{eqn:optimal DP-RP}) with $\Delta_2=\beta$. For any $u, v\in\mathcal U$, let $\tilde u_i=u_i+a_i$ and $\tilde v_i=v_i+b_i$ be the DP noisy vectors, with $a_i, b_i\sim N(0,\sigma^2)$ i.i.d.  Then, denote $\hat g_{org}=\sum_{i=1}^p \tilde u_i\tilde v_i$. we have
\begin{align}
\mathbb E[\hat g_{org}]=\sum_{i=1}^p u_iv_i,  \hspace{0.3in} Var\left(\hat g_{org}\right)
= \sigma^2\sum_{i=1}^p\left(u_i^2+v_i^2\right) + p\sigma^4.
\end{align}
\end{theorem}

\begin{proof}
To add Gaussian noise to the original data, it suffices to find the sensitivity, which, by Definition~\ref{def:neighbor},
is $\Delta_2=\beta$. Thus, the approach is $(\epsilon,\delta)$-DP according to the optimal Gaussian mechanism (Theorem~\ref{theo:gauss-optimal}). To compute the mean and variance, consider some $i\in [p]$. We have
\begin{align}\notag
&\mathbb E\left[\left(u_i + a_i\right)\left(v_i+b_i\right)\right]
= \mathbb E[u_iv_i+a_iv_i+b_iu_i+a_ib_i]
=u_iv_i.
\end{align}
Thus, taking the sum implies $\mathbb E[\hat g_{org}]=\sum_{i=1}^p u_iv_i$. For the variance, \begin{align}\notag
&\mathbb E\left[\left(u_i + a_i\right)\left(v_i+b_i\right)\right]^2
= \mathbb E[u_iv_i+a_iv_i+b_iu_i+a_ib_i]^2
=u_i^2v_i^2+\sigma^2\left(u_i^2+v_i^2\right) + \sigma^4,
\end{align}
which leads to
\begin{align}\notag
&Var\left(\left(u_i + a_i\right)\left(v_i+b_i\right)\right)
= \sigma^2\left(u_i^2+v_i^2\right) + \sigma^4.
\end{align}
Therefore, by independence,
\begin{align}\notag
Var\left(\hat g_{org}\right) = Var\left(\sum_{i=1}^p\left(u_i + a_i\right)\left(v_i+b_i\right)\right)
&= \sigma^2\sum_{i=1}^p\left(u_i^2+v_i^2\right) + p\sigma^4,
\end{align}
which proves the claim.
\end{proof}

For DP-RP-G-OPT-B  and DP-OPORP, we have the following results.

\begin{theorem}[DP-RP-G-OPT-B inner product estimation]  \label{theo:DP-RP-inner}
Let $\sigma$ be the solution to (\ref{eqn:optimal DP-RP}) with $\Delta_2=\beta$. In Algorithm~\ref{alg:DP-RP-G-OPT-B}, let $W\in\{-1,1\}^{p\times k}$ be a Rademacher random matrix. Denote $x=\frac{1}{\sqrt k}W^Tu$, $y=\frac{1}{\sqrt k}W^Tv$, and $a$, $b$ are two random Gaussian noise vectors following $N(0,\sigma^2)$. Let $\hat g_{rp}=\sum_{j=1}^k (x_j+a_j)(y_j+b_j)$. Then, $\mathbb E[\hat g_{rp}]=\sum_{i=1}^p u_iv_i$, and
\begin{align}
    Var\left(\hat g_{rp}\right)
=\sigma^2\sum_{i=1}^p\left(u_i^2+v_i^2\right) + k\sigma^4 + \frac{1}{k}\left(\sum_{i=1}^pu_i^2\sum_{i=1}^pv_i^2+\left(\sum_{i=1}^pu_iv_i\right)^2-2\sum_{i=1}^pu_i^2v_i^2\right).
\end{align}
\end{theorem}

\begin{proof}
The conditional mean and variance can be computed as
\begin{align}\notag
&\mathbb E\left[\sum_{j=1}^k\left(x_j + a_j\right)\left(y_j+b_j\right)|x_j, y_j, j = 1, ..., k\right]
=\sum_{j=1}^kx_jy_j,
\end{align}
\begin{align}\notag
&Var\left(\sum_{j=1}^k\left(x_j + a_j\right)\left(y_j+b_j\right)|x_j, y_j, j = 1, ..., k\right)
=\sigma^2\sum_{j=1}^k (x_j^2+y_j^2) +  k\sigma^4,
\end{align}
where the variance calculation follows from Theorem~\ref{theo:original-data-inner}. Hence, we have
\begin{align}\notag
&\mathbb E\left[\sum_{j=1}^k\left(x_j + a_j\right)\left(y_j+b_j\right)\right]
=\mathbb E\left[\sum_{j=1}^kx_jy_j\right] = \sum_{i=1}^p u_iv_i,
\end{align}
\begin{align}\notag
Var\left(\hat g_{rp}\right)
=&\mathbb E\left[\sigma^2\sum_{j=1}^k (x_j^2+y_j^2) +  k\sigma^4\right]+ Var\left(\sum_{j=1}^kx_jy_j\right)\\\label{eqn:var_g_rp2}
=&\sigma^2\sum_{i=1}^p\left(u_i^2+v_i^2\right) + k\sigma^4 + \frac{1}{k}\left(\sum_{i=1}^pu_i^2\sum_{i=1}^pv_i^2+\left(\sum_{i=1}^pu_iv_i\right)^2-2\sum_{i=1}^pu_i^2v_i^2\right).
\end{align}
In the above calculation, the formula of $Var\left(\sum_{j=1}^kx_jy_j\right)$ is from the result in~\citet{li2006very} with $s=1$ for Rademacher distribution.
\end{proof}

\begin{theorem}[DP-OPORP inner product estimation]  \label{theo:DP-OPORP-inner}
Let $\sigma$ be the solution to (\ref{eqn:optimal DP-RP}) with $\Delta_2=\beta$. Let $w\in\{-1,1\}^p$ be a Rademacher random vector. In Algorithm~\ref{alg:DP-OPORP}, let $x$ and $y$ be the OPORP of $u$ and $v$, and $a$, $b$ be two random Gaussian noise vectors following $N(0,\sigma^2)$. Denote $\hat g_{oporp}=\sum_{j=1}^k (x_j+a_j)(y_j+b_j)$. Then, $\mathbb E[\hat g_{oporp}]=\sum_{i=1}^p u_iv_i$, and
\begin{align}
    Var\left(\hat g_{oporp}\right)
=\sigma^2\sum_{i=1}^p\left(u_i^2+v_i^2\right) + k\sigma^4 + \frac{1}{k}\left(\sum_{i=1}^pu_i^2\sum_{i=1}^pv_i^2+\left(\sum_{i=1}^pu_iv_i\right)^2-2\sum_{i=1}^pu_i^2v_i^2\right)\frac{p-k}{p-1}.
\end{align}
\end{theorem}

\begin{proof}
The proof is similar to that of Theorem~\ref{theo:DP-RP-inner}, with the help of the result in~\citet{li2023oporp}.
\end{proof}

\vspace{0.2in}

The variance reduction factor $\frac{p-k}{p-1}$ can be quite beneficial when $p$ is not very large. Also, see~\citet{li2023oporp} for the normalized estimators for both OPORP and VSRP (very sparse random projections). The normalization steps can substantially reduce the estimation variance.

\newpage

\noindent\textbf{Comparison.} For the convenience of comparison, let us assume that the data are row-normalized, i.e., $\|u\|^2=1$ for all $u\in\mathcal U$. Let $\rho=\sum_{i=1}^p u_iv_i$. We have
\begin{align}\notag
&Var\left(\hat g_{org}\right)
= 2\sigma^2 + p\sigma^4,\\
 &Var\left(\hat g_{rp}\right)
=2\sigma^2 + k\sigma^4 + \frac{1}{k}\left(1+\rho^2-2\sum_{i=1}^pu_i^2v_i^2\right), \notag\\
&Var\left(\hat g_{oporp}\right)
=2\sigma^2 + k\sigma^4 + \frac{1}{k}\left(1+\rho^2-2\sum_{i=1}^pu_i^2v_i^2\right)\frac{p-k}{p-1}.  \notag
\end{align}
For high-dimensional data (large $p$), we see that $\hat g_{rp}$ and $\hat g_{oporp}$ has roughly the same variance, approximately $2\sigma^2+k\sigma^4+\frac{1}{k}$. We would like to compare this with $Var\left(\hat g_{org}\right)= 2\sigma^2 + p\sigma^4$ the variance for adding noise directly to the original data.

\vspace{0.1in}

Let's define the ratio of the variances:
\begin{align}\label{eqn:ratio}
R =  \frac{2\sigma^2+p\sigma^4}{2\sigma^2+k\sigma^4+\frac{1}{k}} \sim  \frac{p\sigma^4}{k\sigma^4} = \frac{p}{k} \ \  (\text{if } p \text{ is  large or } \sigma \text{ is high} )
\end{align}
to illustrate the benefit of RP-type algorithms (DP-RP and DP-OPORP) in protecting the privacy of the (high-dimensional) data. If $\frac{p}{k} = 100$, then it is possible that the ratio of the variances can be roughly 100. This would be a huge advantage. Figure~\ref{fig:ratio} plots the ratio $R$ for $p=1000$ and $p=10000$ as well as a series of $k/p$ values, with respect to $\sigma$.

\begin{figure}[h]
\mbox{
\includegraphics[width=3.2in]{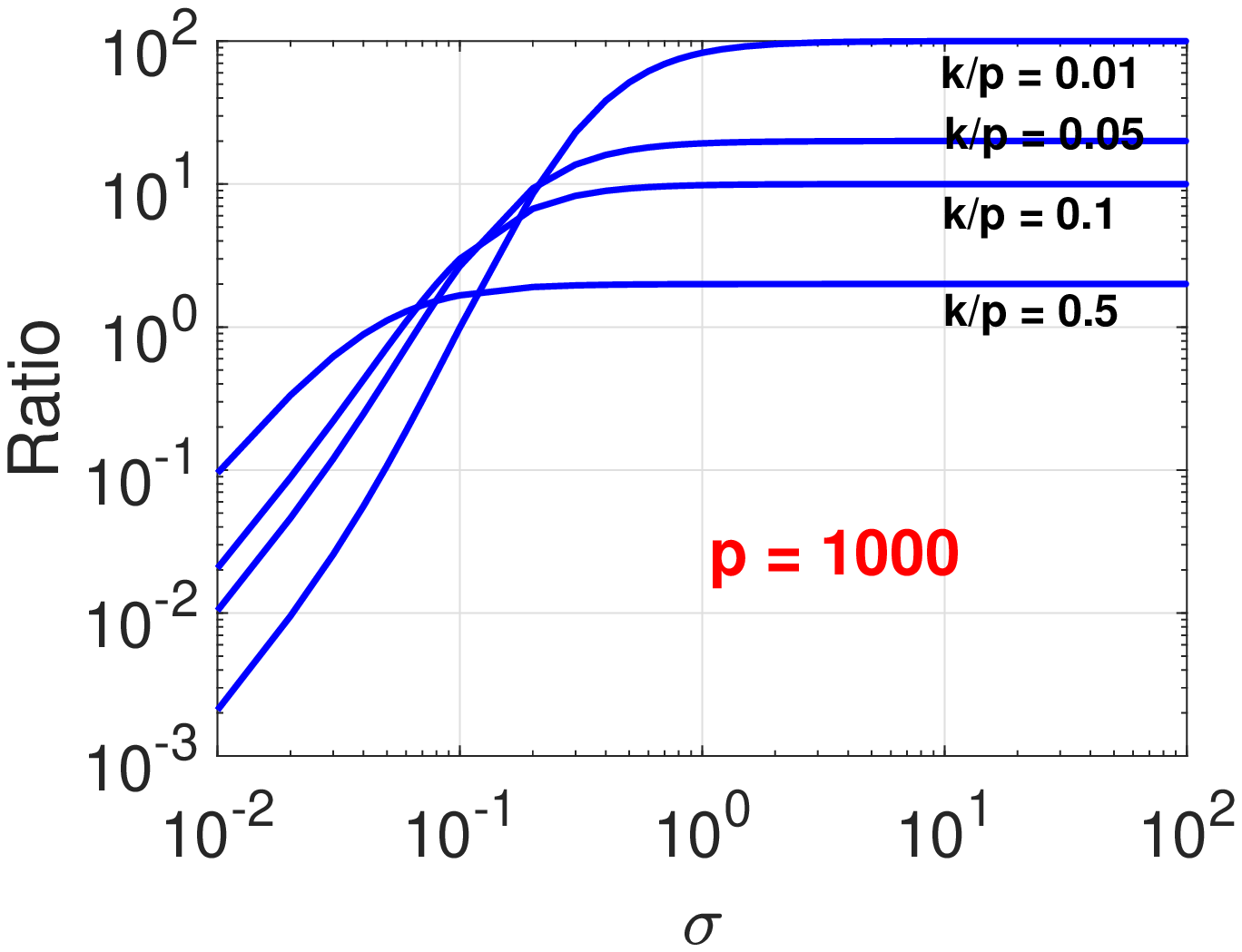}
\includegraphics[width=3.2in]{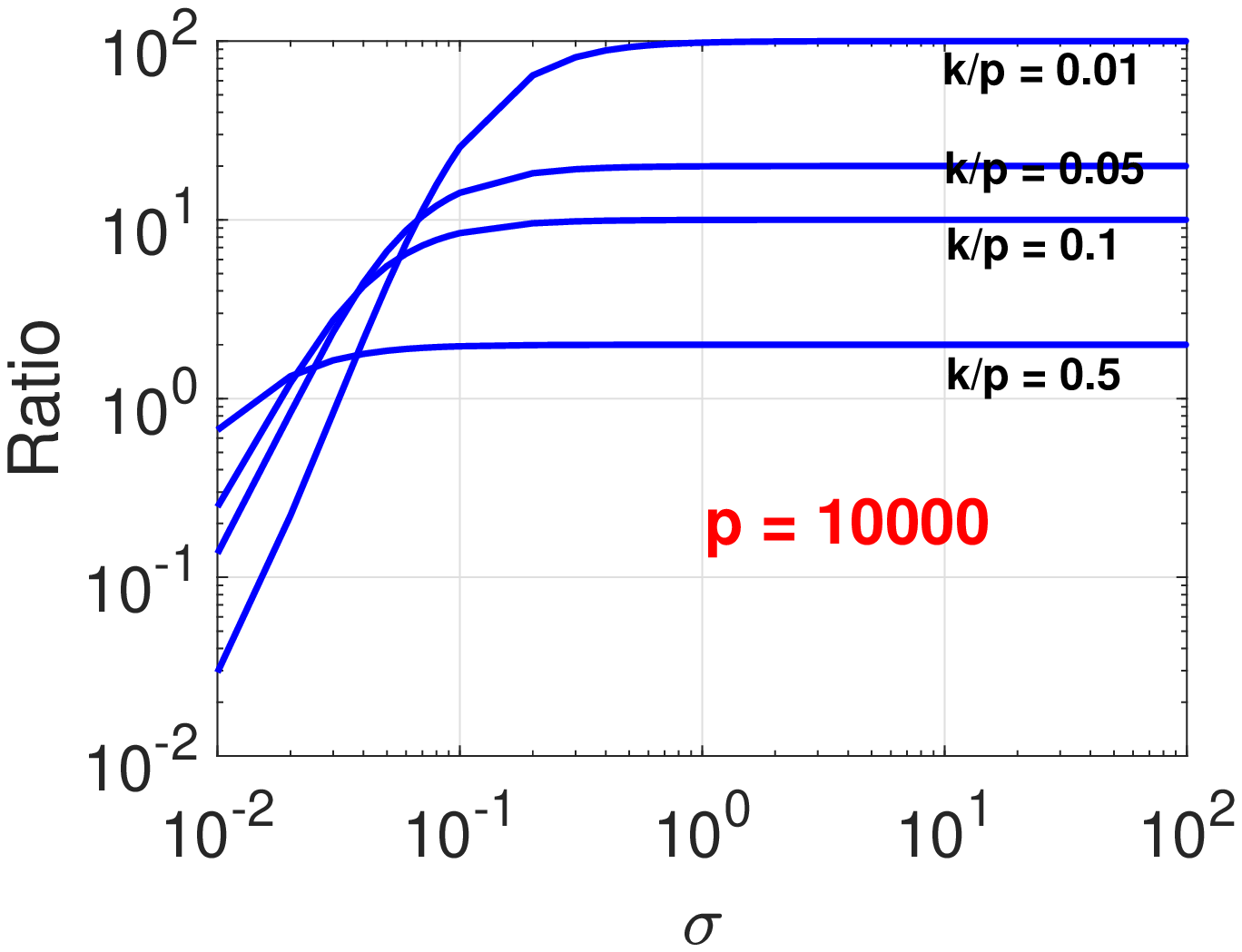}
}

\vspace{-0.1in}

    \caption{We plot the ratio of variances in~\eqref{eqn:ratio} for $p=1000$ and $p=10000$. We choose $k$ values with $k/p \in\{0.01, 0.05, 0.1, 0.5\}$. Then for any $\sigma$ value, we are able to compute the ratio $R$. For larger $\sigma$, we have $R\sim \frac{p}{k}$ as expected. See Figure~\ref{fig:optimal DP-RP}  for the relationship among $\sigma$, $\Delta$, and $\epsilon$ (and $\delta$).    }
    \label{fig:ratio}\vspace{0.2in}
\end{figure}

Figure~\ref{fig:ratio} also illustrates when it might be a good strategy to directly add noise to the original data. For example, when $p=1000$, the ratio can be below 1 if $\sigma<0.1$. Figure~\ref{fig:optimal DP-RP}  depicts the relationship among $\sigma$, $\Delta$, and $\epsilon$ (and $\delta$). One can verify that, in order for $\sigma<0.1$ at $\Delta_2=\beta=1$, we need $\epsilon>100$. In the literature, many DP applications typically require a much smaller $\epsilon$, such as $\epsilon\in [0.1,20]$ (e.g., \citet{haeberlen2014differential,kenny2021use}).

\newpage

\section{DP-SignRP: Differentially Private Sign Random Projections} \label{sec:DP-signRP}

We have reviewed  in the Introduction the advantage of sign random projections (SignRP) (in terms of storage, communications, and computations).  In this section, we propose algorithms that output SignRP with differential privacy guarantees. The  straightforward approach is directly taking the signs of DP-RP methods, e.g., from Algorithm~\ref{alg:DP-RP-gaussian-noise} or Algorithm~\ref{alg:DP-RP-G-OPT-B}. By the post-processing property of DP, the resultant signs are also DP. In our study, we focus on algorithms that randomize the signs (bit-vector) of the non-private RPs. That is, we first generate non-private SignRP, and then perturb the bits.

In the following, we first study the DP-SignRP from Gaussian random projections. Then we demonstrate the use of  Rademacher projections for improving the performance. Finally, we propose algorithms for DP-SignOPORP and discuss its advantages in terms of privacy protection.

\subsection{DP-SignRP-RR by Randomized Response} \label{sec:DP-signRP-RR}

\begin{algorithm}[h]
{
    \vspace{0.05in}
    \textbf{Input:} Data $u\in[-1,1]^p$; $\epsilon>0$, $0<\delta<1$, number of projections~$k$, norm lower bound $m$

    \vspace{0.05in}

    \textbf{Output:}   Differentially private sign random projections

    \vspace{0.05in}

    Apply RP by $x=\frac{1}{\sqrt k}W^Tu$, where $W\in\mathbb R^{p\times k}$ is a random $N(0,1)$ matrix

	Let $N_+(m,\delta,k,p)$ be computed as (\ref{eqn:N+ bound}) in Proposition~\ref{prop:num-of-changes}

    Compute $\tilde s_j=\begin{cases}
    sign(x_j), & \text{with prob.}\ \frac{e^{\epsilon'}}{e^{\epsilon'}+1}\\
    -sign(x_j), & \text{with prob.}\ \frac{1}{e^{\epsilon'}+1}
    \end{cases}$ for $j=1,...,k$, with $\epsilon'=\epsilon/N_+(m,\delta,k,p)$

    Return $\tilde s$ as the DP-SignRP of $u$
    }
    \caption{DP-SignRP-RR}
    \label{alg:DP-signRP-RR}
\end{algorithm}

Firstly, we develop a DP algorithm for SignRP based on the classic randomized response (RR) technique~\citep{warner1965randomized,dwork2014algorithmic} and call it DP-SignRP-RR. As summarized in Algorithm~\ref{alg:DP-signRP-RR}, after we apply random projection $x=\frac{1}{\sqrt k}W^Tu$, we take $s=sign(x)$. Then, for each $s_j$, we keep the sign with probability $\frac{e^{\epsilon'}}{e^{\epsilon'}+1}$ and flip the sign with probability $\frac{1}{e^{\epsilon'}+1}$. Here $\epsilon'=\epsilon/N_+$, where $N_+(m,\delta,k,p)$ is an upper bound on the number of different signs (among $k$ signs) of $x=\frac{1}{\sqrt k}W^Tu$ and $x=\frac{1}{\sqrt k}W^Tu'$ for any $\beta$-adjacent $(u,u')$, which will be derived later in Proposition~\ref{prop:num-of-changes}. The final output is $\tilde s$ after perturbing $s$ by the above procedure. Note that, in Algorithm~\ref{alg:DP-signRP-RR}, in order to achieve differential privacy, we additionally assume that there exists a lower bound on the $l_2$ norm of the data, i.e., $\|u\|\geq m$ for all $u$ in the data domain. We have the following privacy guarantee as a Theorem.

\vspace{0.1in}

\begin{theorem} \label{theo:DP-SignRP-RR privacy}
Algorithm~\ref{alg:DP-signRP-RR} is $(\epsilon,\delta)$-DP.
\end{theorem}
\begin{remark}  \label{remark:DP-SignRP-RR}
If we do not assume a lower bound on the data norm (i.e., $m=0$), then $N_+=k$. That is, we need to set $\epsilon'=\epsilon/k$ in Algorithm~\ref{alg:DP-signRP-RR}, and the algorithm becomes $\epsilon$-DP.
\end{remark}

\begin{proof}
For any data point $u\in \mathcal U$ and its $\beta$-adjacent neighbor $u'$, denote $s=sign(W^Tu)\in \{-1,+1\}^k$, $s'=sign(W^Tu')\in \{-1,+1\}^k$, and let $\tilde s$ and $\tilde s'$ be the corresponding randomized sign vectors output by Algorithm~\ref{alg:DP-signRP-RR}. Denote $S=\{i:s_j\neq s_j'\}$ and $S^c=[k]\setminus S$. For any vector $y\in \{-1,+1\}^k$, define $S_0=\{j\in S:s_j=y_j\}$, $S_1=\{j\in S:s_j\neq y_j\}$, $S^c_0=\{j\in S^c:s_j=y_j\}$ and $S^c_1=\{j\in S^c:s_j\neq y_j\}$. By Proposition~\ref{prop:num-of-changes}, we know that the event $\{|S|\leq N_+(m,\delta,k,p)\}$ happens with probability at least $1-\delta$. In this event, we have
\begin{align*}
    \log\frac{Pr(\tilde s=y)}{Pr(\tilde s'=y)}&=\log\frac{\prod_{j\in S^c_0}\frac{e^{\epsilon'}}{e^{\epsilon'}+1}\prod_{j\in S^c_1}\frac{1}{e^{\epsilon'}+1}\prod_{j\in S_0}\frac{e^{\epsilon'}}{e^{\epsilon'}+1}\prod_{j\in S_1}\frac{1}{e^{\epsilon'}+1}}{\prod_{j\in S^c_0}\frac{e^{\epsilon'}}{e^{\epsilon'}+1}\prod_{j\in S^c_1}\frac{1}{e^{\epsilon'}+1}\prod_{j\in S_0}\frac{1}{e^{\epsilon'}+1}\prod_{j\in S_1}\frac{e^{\epsilon'}}{e^{\epsilon'}+1}} \\
    &\leq \log\frac{\prod_{j\in S}\frac{e^{\epsilon'}}{e^{\epsilon'}+1}}{\prod_{j\in S}\frac{1}{e^{\epsilon'}+1}}=|S|\epsilon'\leq N_+ \epsilon'=\epsilon.
\end{align*}
Since this event occurs with probability at least $1-\delta$, the overall procedure is $(\epsilon,\delta)$-DP.
\end{proof}

\subsubsection{The flipping probability and calculation of $N_+$}

As in the proof of Theorem~\ref{theo:DP-SignRP-RR privacy}, $N_+$ is the upper bound on $S=\{i:s_j\neq s_j'\}$, where $s=sign(W^Tu)\in \{-1,+1\}^k$ and $s'=sign(W^Tu')\in \{-1,+1\}^k$ for neighboring data vectors $(u,u')$. $N_+$ also determines the flipping probability in DP-SignRP-RR, which is important for the utility. Next, we analyze $N_+$. The following is a useful lemma that describes some statistical properties of the absolute value of bivariate normal random variables, which might be of independent interest. The proof can be found in Appendix.

\vspace{0.1in}

\begin{lemma}  \label{lemma:conditional-gaussian}
Let $\begin{pmatrix}
X \\ Y
\end{pmatrix}\sim N\begin{pmatrix}
\sigma_x^2 & \rho\sigma_x\sigma_y \\ \rho\sigma_x\sigma_y & \sigma_y^2
\end{pmatrix}$. Denote $r=\sigma_x/\sigma_y$. Then we have:

1. $Pr(|X|>|Y|)=\frac{1}{\pi}\left[ \tan^{-1}\left( \frac{r-\rho }{\sqrt{1-\rho^2}} \right)+\tan^{-1}\left( \frac{r+\rho }{\sqrt{1-\rho^2}} \right) \right]$. When $r\leq 1$, the maximum is achieved at $\rho=0$, i.e., $\max_{\rho} Pr(|X|<|Y|)=\frac{2}{\pi}\tan^{-1}(r)$.

2. The conditional expectation:
\begin{align}
    \mathbb E\left[|X|\big| |X|>|Y|\right]=\sigma_x\sqrt{\frac{\pi}{2}}\cdot \frac{\frac{r-\rho}{\sqrt{1+r^2-2r\rho}}+\frac{r+\rho}{\sqrt{1+r^2+2r\rho}}}{\tan^{-1}\left( \frac{r-\rho }{\sqrt{1-\rho^2}} \right)+\tan^{-1}\left( \frac{r+\rho }{\sqrt{1-\rho^2}} \right)}.  \label{eqn:conditional-expectation}
\end{align}

3. The conditional tail probability: for any $r>0$, $\rho\in (-1,1)$, for any $t>0$,
\begin{align}
    Pr(|X|>t\big| |X|>|Y|)\leq \exp\left( -\frac{t^2}{2\sigma_x^2} \right).  \label{eqn:conditional-tail-prob}
\end{align}
\end{lemma}

\vspace{0.1in}

With Lemma~\ref{lemma:conditional-gaussian}, we have the following result on the maximum of a group of Gaussian variables.

\begin{lemma}  \label{lemma:p-vars-prob-bound}
Let $X_1,...,X_p$ be iid $N(0,\sigma_x^2)$ variables. Let $Y\sim N(0,\sigma_y^2)$ be another Gaussian random variable with arbitrary dependence structure with $X_i$'s. Let $r=\sigma_x/\sigma_y\leq 1.$ Then
\begin{align} \label{eqn:p-vars-prob-bound}
    P_+(r,p)\eqdef Pr(\max_{i=1,...,p}|X|>|Y|)\leq \int_0^\infty 2p[2\Phi(t)-1]^{p-1}[2\Phi(rt)-1]\phi(t) dt,
\end{align}
where $\phi(x)$ and $\Phi(x)$ are the standard Gaussian pdf and cdf, respectively.
\end{lemma}
\begin{proof}
By Lemma~\ref{lemma:conditional-gaussian} and independence, we know that $Pr(\max_{i=1,...,p}|X|>|Y|)$ reaches its maximum when $X_i$ is independent of $Y$, i.e., when $\rho(X_i,Y)=0$, $\forall i=1,...,p$. Since $|X_i|$ follows a half-normal distribution with cdf being $\erf(\frac{x}{\sqrt 2 \sigma_x})$, we have
\begin{align*}
    Pr(\max_{i=1,...,p}|X_i|\leq t) = \erf\left(\frac{t}{\sqrt 2\sigma_x}\right)^p=\left[2\Phi\left(\frac{x}{\sigma_x}\right)-1 \right]^p,
\end{align*}
and probability density function $g(x)=2p[\Phi(\frac{x}{\sigma_x})-1]^{p-1}\frac{1}{\sqrt{2\pi}\sigma_x}e^{-\frac{x^2}{2\sigma_x^2}}$. When $Y$ is independent of all $X_i$'s (which gives the upper bound), we have
\begin{align*}
    Pr(\max_{i=1,...,p}|X|>|Y|)&=\int_0^\infty 2p\left[\Phi \left(\frac{x}{\sigma_x}\right)-1\right]^{p-1}\frac{1}{\sqrt{2\pi}\sigma_x}e^{-\frac{x^2}{2\sigma_x^2}}Pr(|Y|<x) dx \\
    &=\int_0^\infty 2p\left[\Phi \left(\frac{x}{\sigma_x}\right)-1\right]^{p-1}\frac{1}{\sqrt{2\pi}\sigma_x}e^{-\frac{x^2}{2\sigma_x^2}}\erf\left(\frac{x}{\sqrt 2\sigma_y}\right) dx \\
    &=\int_0^\infty 2p[2\Phi(t)-1]^{p-1}[2\Phi(rt)-1]\phi(t) dt,
\end{align*}
with a proper change of variables. This gives an upper bound as shown above.
\end{proof}

\vspace{0.1in}

In Lemma~\ref{lemma:p-vars-prob-bound}, $P_+(r,p)$ is an increasing function in both $r$ and $p$. In Figure~\ref{fig:P+},
we plot the $P_+(r,p)$ with different $r$ and $p$ values. We see that when $r$ is as small as 0.1, $P_+(r,p)$ is smaller than 0.4 when $p$ is as large as $10^6$.

\begin{figure}[h]

\centering

\includegraphics[width=3.2in]{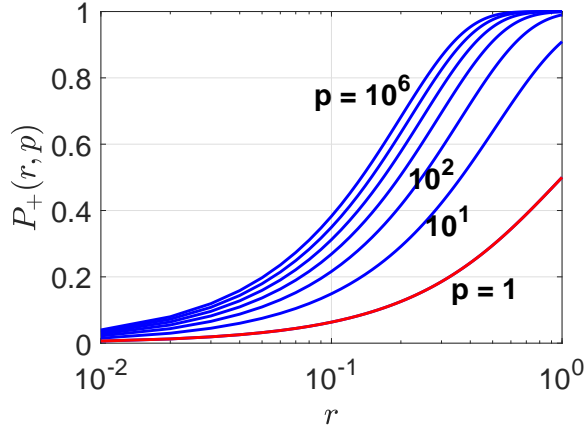}

\vspace{-0.15in}

\caption{$P(r,p)$ in (\ref{eqn:p-vars-prob-bound}) in Lemma~\ref{lemma:p-vars-prob-bound} against $r\leq 1$ at different $p$. The red curve ($p=1$) also equals $Pr(|X|>|Y|)$ in Lemma~\ref{lemma:conditional-gaussian}. }
\label{fig:P+}
\end{figure}

Equipped with Lemma~\ref{lemma:p-vars-prob-bound}, we can bound $N_+$ in Algorithm~\ref{alg:DP-signRP-RR}, the largest number of projected signs that are possible to change when $u$ is replaced by any neighboring data vector $u'$.

\begin{lemma}[Binomial tail bound] \label{lemma:binomial-tail}
Let $X\sim Binomial(n,p)$ and denote $\mu=\mathbb E[X]=np$. For any $\eta>0$, it holds that
\begin{align*}
    Pr(X\geq (1+\eta)\mu)\leq \exp(-\frac{\eta^2\mu}{\eta+2}).
\end{align*}
\end{lemma}

\begin{proposition}[Bounding $N_+$] \label{prop:num-of-changes}
Suppose $u\in [-1,1]^p$ and $\beta\leq \|u\|$. Denote $r=\frac{\beta}{\|u\|}$ and $F_{\|u\|,p}=P_+(\frac{\beta}{\|u\|},p)$ as (\ref{eqn:p-vars-prob-bound}) in Lemma~\ref{lemma:p-vars-prob-bound}. Denote $s=sign(W^Tu)\in \{-1,+1\}^k$ and $s'=sign(W^Tu')\in \{-1,+1\}^k$ for $u'$ that is $\beta$-neighboring to $u$, and $S=\{i:s_j\neq s_j'\}$. Then, with probability $1-\delta$,
\begin{align} \label{eqn:N+ bound}
    |S| \leq N_+(\|u\|,\delta,k,p)= \min\{ F_{\|u\|,p} k+\frac{1}{2}\Big[\log(1/\delta)+\sqrt{(\log(1/\delta))^2+8F_{\|u\|,p} k \log(1/\delta)}\Big], k\}.
\end{align}
\end{proposition}

\vspace{0.1in}

\begin{remark}  \label{remark:N_+}
Since $F_{\|u\|,p}=P_+(r,p)$ is an increasing function in $r=\beta/\|u\|$, $N_+$ would be smaller if $\beta/\|u\|$ is smaller. That is, the signs of $W^Tu$ and $W^Tu'$ are less likely to be different when the difference between neighboring data, $\beta$, is relatively small compared with the data norm $\|u\|$.
\end{remark}

\begin{figure}[t]
\begin{center}
    \includegraphics[width=3.2in]{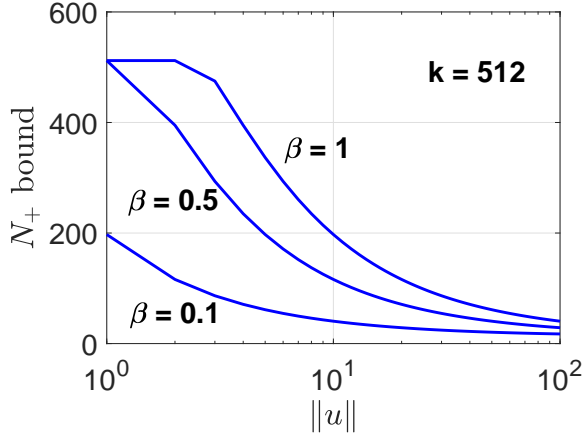}
\end{center}

\vspace{-0.2in}

\caption{The upper bound on $N_+(\|u\|,\delta,k,p)$ in Proposition~\ref{prop:num-of-changes} against data norm $\|u\|$ at different $\beta$. In the simulation, $p=1024$, $k=512$ and $\delta=10^{-6}$. }
\label{fig:N+}
\end{figure}

\begin{proof}
Consider a single Gaussian projection vector $w$ with iid $N(0,1)$ entries. Since $w^Tu=\sum_{i=1}^p u_iw_i$ and each $w_i\sim N(0,1)$, we know that $\begin{pmatrix} \beta w_i \\ x \end{pmatrix}\sim N\begin{pmatrix} \beta^2 & \rho_i\beta\|u\| \\ \rho_i\beta\|u\| & \|u\|^2 \end{pmatrix}$ where $\rho_i=\frac{u_i}{\|u\|}$ is the correlation coefficient. Since $|w^T(u-u')|\leq \beta\max_{i=1,...,p}|w_i|$ by Definition~\ref{def:neighbor} of $\beta$-neighboring (and more generally, when $\|u-u'\|_1\leq \beta$), we have
\begin{align*}
    Pr(\max_{u'\in Nb(u)}|w^T(u-u')|\geq |w^Tu|)= Pr(\beta\max_{i=1,...,p}|w_i|\geq |w^Tu|).
\end{align*}
Note that, $\beta\max_{i=1,...,p}|w_i|\geq |w^Tu|$ is a necessary condition for the event that there exists a neighbor such that $sign(w^Tu)\neq sign(w^Tu')$. Denote $I=\mathbbm 1\{\beta\max_{i=1,...,p}|w_i|\geq |w^Tu|\}$. Applying Lemma~\ref{lemma:p-vars-prob-bound} with $r=\beta/\|u\|\leq 1$ yields
\begin{align}
    \mathbb E[I]=Pr(\beta\max_{i=1,...,p}|w_i|\geq |w^Tu|)\leq F_{\|u\|,p}=\int_0^\infty 2p[2\Phi(t)-1]^{p-1}[2\Phi(rt)-1]\phi(t) dt   \label{eqn:E[I]}
\end{align}
as given by (\ref{eqn:p-vars-prob-bound}). Let $I_j$ be the corresponding indicator function w.r.t. each column in the projection matrix $W$. Denote $N_+=\sum_{j=1}^k I_j$, and by the above reasoning, we know that $|S|\leq N_+$ where $S$ is defined in the theorem. Since the columns of $W$ are independent, $N_+$ follows a $Binomial(k,\mathbb E[I])$ distribution with $k$ trials and success probability $\mathbb E[I]$ bounded as above.  Applying Chernoff's bound on binomial variable (Lemma~\ref{lemma:binomial-tail}), we obtain
\begin{align*}
    Pr(N_+\geq (1+\eta) F_{\|u\|,p} k)\leq \exp(-\frac{\eta^2F_{\|u\|,p} k}{\eta+2}).
\end{align*}
Setting the RHS to $\delta$ gives $\eta=\frac{\log(1/\delta)+\sqrt{(\log(1/\delta))^2+8F_{\|u\|,p} k \log(1/\delta)}}{2F_{\|u\|,p} k}$. Therefore, with probability $1-\delta$,
\begin{align*}
    N_+(\|u\|,\delta,k,p)\leq F_{\|u\|,p} k+\frac{1}{2}\Big[\log(1/\delta)+\sqrt{(\log(1/\delta))^2+8F_{\|u\|,p} k \log(1/\delta)} \Big].
\end{align*}
In addition, $N_+\leq k$ trivially. The proof is complete.
\end{proof}

\vspace{0.1in}

In Figure~\ref{fig:N+}, we plot the bound on $N_+(\|u\|,\delta,k,p)$ derived in Proposition~\ref{prop:num-of-changes} against $\|u\|$ with various $\beta$ and fixed $p$, $k$ and $\delta$. As expected, $N_+(\|u\|,\delta,k,p)$ shrinks rapidly as the data norm $\|u\|$ increases. Since $N_+(\|u\|,\delta,k,p)$ is decreasing in $\|u\|$ (because $F_{\|u\|,p}$ decreases with $\|u\|$), we know that for all $u$ with $\|u\|\geq m$, the $|S|$ in Proposition~\ref{prop:num-of-changes} would be upper bounded by $N_+(m,\delta,k,p)$ with high probability, as used in Algorithm~\ref{alg:DP-signRP-RR} and the proof of Theorem~\ref{theo:DP-SignRP-RR privacy}. Also, in DP-SignRP-RR, the probability of flipping the true SignRP, $\frac{1}{e^{\epsilon/N_+}+1}$, would be smaller when the data norm (or, its lower bound $m$) is large compared with $\beta$.

\subsubsection{Utility in angle estimation by DP-SignRP-RR}

Define the DP-SignRP-RR estimator of the angle between two data points $u$ and $v$ as
\begin{align}
&\hat\theta_{RR} = \pi(1-\hat P_{RR}), \label{est:DP-signRP-RR} \\
\text{where}\ &\hat P_{RR}=\frac{(e^{\epsilon'}+1)^2}{(e^{\epsilon'}-1)^2} \frac{1}{k}\sum_{j=1}^k \mathbbm 1\{\tilde s_{1j}=\tilde s_{2j}\}-\frac{2 e^{\epsilon'}}{(e^{\epsilon'}-1)^2}.  \nonumber
\end{align}
We have the following  result on the utility.

\vspace{0.1in}

\begin{theorem}  \label{theo:DP-signRP-utility}
Let $\rho=\cos(u,v)$ and $\theta=\cos^{-1}(\rho)$. Run Algorithm~\ref{alg:DP-signRP-RR} with $N_+$ given in Proposition~\ref{prop:num-of-changes} and define the DP-SignRP-RR angle estimator by (\ref{est:DP-signRP-RR}). We have $\mathbb E[\hat\theta_{RR}]=\theta$. As $k\rightarrow \infty$, the following holds:
\begin{align*}
    \hat\theta_{RR}\rightarrow N(\theta, \frac{V_{RR}}{k}),
\end{align*}
where $$V_{RR}=\theta(\pi-\theta)+\frac{2\pi^2e^{\epsilon/N_+}}{(e^{\epsilon/N_+}-1)^2}+\frac{4\pi^2e^{2\epsilon/N_+}}{(e^{\epsilon/N_+}-1)^4}.$$
\end{theorem}

\begin{proof}
Let $s=sign(W^Tu)\in \{-1,+1\}^k$, $s'=sign(W^Tu')\in \{-1,+1\}^k$. We denote the collision probability of non-private SignRP as $$P_{SRP}=Pr(s_{1j}=s_{2j})=1-\frac{\cos^{-1}(\rho)}{\pi}=1-\frac{\theta}{\pi}.$$
Hence, the collision probability of DP-SignRP-RR can be computed as
\begin{align*}
    \tilde P\eqdef Pr(\tilde s_{1j}&=\tilde s_{2j})=Pr(s_{1j}=s_{2j}, \text{both change sign or not change sign}) \\
    &\hspace{2.2in} + Pr(s_{1j}\neq s_{2j}, \text{one sign changes})\\
    &=P_{SRP}[(\frac{e^{\epsilon'}}{e^{\epsilon'}+1})^2+(\frac{1}{e^{\epsilon'}+1})^2]+2(1-P_{SRP})\frac{e^{\epsilon'}}{(e^{\epsilon'}+1)^2}\\
    &=P_{SRP}\frac{(e^{\epsilon'}-1)^2}{(e^{\epsilon'}+1)^2}+\frac{2 e^{\epsilon'}}{(e^{\epsilon'}+1)^2},
\end{align*}
which increases linearly in $P_{SRP}$. Thus, it holds that
\begin{align*}
    \mathbb E[\hat P_{RR}]=\frac{(e^{\epsilon'}+1)^2}{(e^{\epsilon'}-1)^2} \tilde P-\frac{2 e^{\epsilon'}}{(e^{\epsilon'}-1)^2} =P_{SRP}=1-\frac{\theta}{\pi},
\end{align*}
which implies $\mathbb E[\hat\theta_{RR}]=\pi\left(1-(1-\frac{\theta}{\pi})\right)=\theta$. To compute the variance, we first estimate $\theta=\cos^{-1}(\rho)$~by
\begin{align*}
    \hat\theta=\pi(1-\hat P_{RR}).
\end{align*}
Then according to the Central Limit Theorem (CLT), for the sample mean of iid Bernoulli's, as $k\rightarrow \infty$, we have
\begin{align*}
    \frac{1}{k}\sum_{j=1}^k \mathbbm 1\{\tilde s_{1j}=\tilde s_{2j}\} \rightarrow N(\tilde P, \frac{\tilde P(1-\tilde P)}{k}).
\end{align*}
As a result, we have $\hat\theta\rightarrow N(\theta,\frac{V_{RR}}{k})$, where
\begin{align*}
    V_{RR}&=\frac{\pi^2(e^{\epsilon'}+1)^4}{(e^{\epsilon'}-1)^4}\Big[(1-\frac{\theta}{\pi})\frac{(e^{\epsilon'}-1)^2}{(e^{\epsilon'}+1)^2}+\frac{2 e^{\epsilon'}}{(e^{\epsilon'}+1)^2}  \Big] \Big[ \frac{e^{2\epsilon'}+1}{(e^{\epsilon'}+1)^2}-(1-\frac{\theta}{\pi})\frac{(e^{\epsilon'}-1)^2}{(e^{\epsilon'}+1)^2} \Big]\\
    &=\frac{\pi^2(e^{\epsilon'}+1)^4}{(e^{\epsilon'}-1)^4}\Big[(1-\frac{\theta}{\pi})\frac{(e^{\epsilon'}-1)^2}{(e^{\epsilon'}+1)^2}+\frac{2 e^{\epsilon'}}{(e^{\epsilon'}+1)^2}  \Big] \Big[ \frac{\theta}{\pi}\frac{(e^{\epsilon'}-1)^2}{(e^{\epsilon'}+1)^2}+\frac{2 e^{\epsilon'}}{(e^{\epsilon'}+1)^2} \Big] \\
    &=\frac{\pi^2\theta}{\pi}(1-\frac{\theta}{\pi})+(1-\frac{\theta}{\pi})\frac{2e^{\epsilon'}}{(e^{\epsilon'}-1)^2}+\frac{\theta}{\pi}\frac{2e^{\epsilon'}}{(e^{\epsilon'}-1)^2}+\frac{4e^{2\epsilon'}}{(e^{\epsilon'}-1)^4} \\
    &=\theta(\pi-\theta)+\frac{2\pi^2e^{\epsilon'}}{(e^{\epsilon'}-1)^2}+\frac{4\pi^2e^{2\epsilon'}}{(e^{\epsilon'}-1)^4}.
\end{align*}
We conclude the proof by replacing $\epsilon'=\epsilon/N_+$.
\end{proof}

\vspace{0.1in}

Theorem~\ref{theo:DP-signRP-utility} says that $\hat\theta_{RR}$ is an unbiased estimator of $\theta$ and asymptotically normal. Compared with the variance in (\ref{eqn:var-signRP}) of the estimator from SignRP, we see that $\hat\theta_{RR}$ incurs an extra variance (i.e., utility loss) of $\frac{\pi^2}{k}\Big[ \frac{2e^{\epsilon/N_+}}{(e^{\epsilon/N_+}-1)^2}+\frac{4e^{2\epsilon/N_+}}{(e^{\epsilon/N_+}-1)^4}\Big]$. This quantity increases as $\epsilon$ gets smaller, illustrating the utility-privacy trade-off of DP-SignRP-RR.

\vspace{0.1in}

\noindent\textbf{Optimal projection dimension $k^*$}. In Theorem~\ref{theo:DP-signRP-utility}, note that the variance $V_{RR}$ does not always decrease with larger $k$. In particular, in Proposition~\ref{prop:num-of-changes} since $N_+\asymp F_{\|u\|,p} k$, the term $\Big[ \frac{2e^{\epsilon/N_+}}{(e^{\epsilon/N_+}-1)^2}+\frac{4e^{2\epsilon/N_+}}{(e^{\epsilon/N_+}-1)^4}\Big]$ increases with $k$. Therefore, there exists an optimal $k^*$ that minimizes the estimation variance. When $k$ is sufficiently large, we have the approximation
\begin{align*}
    \frac{V_{RR}}{k}\approx \frac{\theta(\pi-\theta)}{k}+\frac{2\pi^2F_{\|u\|,p}^2 k}{\epsilon^2}+\frac{8\pi^2F_{\|u\|,p}^4 k^3}{\epsilon^4},
\end{align*}
using the approximation that $e^x-1\approx x$ when $x$ is small. To find the optimal $k$ minimizing this expression, we compute the derivative and set it as zero:
\begin{align*}
    -\frac{\theta(\pi-\theta)}{k^2}+\frac{2\pi^2F_{\|u\|,p}^2}{\epsilon^2}+\frac{24\pi^2F_{\|u\|,p}^4 k^2}{\epsilon^4}=0\ \ \Longrightarrow \ \ k^*\asymp \frac{\epsilon\theta(\pi-\theta)}{F_{\|u\|,p}}. 
\end{align*}
This analysis suggests that theoretically the optimal $k^*$ is larger when: (1) $\epsilon$ is large; (2) $F_{\|u\|,p}$ is small, which is typically true when the norm of the data is relatively large.

\newpage

\subsection{DP-SignRP-RR-Smooth Using Smooth Flipping Probability}  \label{sec:DP-SignRP-RR-smooth}

The DP-SignRP-RR algorithm is not satisfactory in terms of privacy preservation. For instance, when the norm of the data is not very large (e.g., $m=10$ and $\beta=1$), for $k=512$, $N_+$ can be as large as $200$ (see Figure~\ref{fig:N+}). This implies that, we need $\epsilon=400$ in order for the probability of flipping a true SignRP to be $\frac{1}{e^{\epsilon/N_+}+1}\approx 11.9\%$, which is still fairly large to achieve a good utility in practice. Thus, the privacy protection of DP-SignRP-RR might be too weak. Next, we develop a new algorithm that could reduce the flipping probability of DP-SignRP-RR, thanks to the ``robustness'' of SignRP brought by the ``aggregate-and-sign'' operation.

\begin{algorithm}[h]
{
    \vspace{0.05in}
    \textbf{Input:} Data $u\in[-1,1]^p$; $\epsilon>0$, number of projections~$k$

    \vspace{0.05in}

    \textbf{Output:}   Differentially private sign random projections

    \vspace{0.05in}

    Apply RP by $x=\frac{1}{\sqrt k}W^Tu$, where $W\in\mathbb R^{p\times k}$ is a random $N(0,1)$ matrix

    Compute $L_j=\lceil \frac{|x_j|}{\beta\max_{i=1,...,p} |W_{ij}|} \rceil$ for $j=1,...,k$

    Compute $\tilde s_j=\begin{cases}
    sign(x_j), & \text{with prob.}\ \frac{e^{\epsilon_j'}}{e^{\epsilon_j'}+1}\\
    -sign(x_j), & \text{with prob.}\ \frac{1}{e^{\epsilon_j'}+1}
    \end{cases}$ for $j=1,...,k$, with $\epsilon_j'=\frac{L_j}{k}\epsilon$

    Return $\tilde s$ as the DP-SignRP of $u$
    }
    \caption{DP-SignRP-RR-smooth using smooth flipping probability}
    \label{alg:DP-signRP-RR-smooth}
\end{algorithm}

\vspace{0.1in}
\noindent\textbf{Smooth flipping probability.}  We propose a novel sampling approach based on \textit{``smooth flipping probability''}. Recall the motivation of smooth sensitivity (Definition~\ref{def:smooth-sensitivity}): we may allow the use of local sensitivity at $u$ as long as we account for the local sensitivities of all $u'\in\mathcal U$, discounted exponentially by a factor $e^{-\epsilon d(u,u')}$ where $d(u,u')$ is the minimal number of ``taking-a-neighbor'' operations to reach $u'$ from $u$. Conceptually, smooth sensitivity says that, the farther $u$ is from a high local sensitivity point, the less noise is needed for $u$ to achieve DP.


What is the ``local perturbation'' for the 1-bit SignRP? Consider one projection $x_j=W[:,j]^T u$ and $s_j=sign(x_j)$. For a neighbor $u'$ of $u$, denote $x_j'=W[:,j]^T u'$ and $s_j'=sign(x_j')$. Obviously, by $\beta$-adjacency (Definition~\ref{def:neighbor}), $s_j'$ is possible to be different from $s_j$ only if $|x_j|< \beta\max_{i=1,...,p} |W_{ij}|$, because $|x_j-x_j'|\leq \beta\max_{i=1,...,p} |W_{ij}|$. Intuitively, this implies that projected values near zero would have larger ``local sensitivity''. More concretely,
\begin{itemize}
    \item When $|x_j|< \beta\max_{i=1,...,p} |W_{ij}|$, the ``local sign flipping'' (analogue to the local sensitivity for noise addition) is required, since $s_j'$ may change the sign of $s_j$. A typical choice is the standard randomized response (RR) strategy: keep the sign with probability $\frac{e^{\epsilon/k}}{e^{\epsilon/k}+1}$ and flip otherwise.

    \item When $|x_j|\geq \beta\max_{i=1,...,p} |W_{ij}|$, there does not exist $u'\in Nb(u)$ such that $s_j'\neq s_j$. Thus, the ``local sign flipping'' is not needed (i.e., locally, we do not need to apply random sign flips). However, due to the same issue as the local sensitivity, if we do not perturb $s_j$ in this case, then the algorithm does not satisfy DP. Instead, we leverage the idea of smooth sensitivity and propose \textit{``smooth flipping probability''} which applies less perturbation for points far away from the ``high-sensitivity region'' (i.e., the regime with $|x_j|< \beta\max_{i=1,...,p} |W_{ij}|$). Specifically, we define $L_j=\lceil \frac{|x_j|}{\beta\max_{i=1,...,p} |W_{ij}|} \rceil$. The probability of keeping the sign $s_j$ is $P_j^{(u)}=\frac{e^{\epsilon_j'}}{e^{\epsilon_j'}+1}$ where $\epsilon_j'=\frac{L_j}{k}\epsilon$. In Theorem~\ref{theo:privacy-DPSignRP-RR-smooth}, we will see that this flipping probability is ``smooth'': for $\forall u,u'$ that are neighbors, $P_j^{(u)}\leq e^{\epsilon/k}P_j^{(u')}$, which justifies its name.
\end{itemize}

\begin{figure}[h]
  \begin{center}
    \includegraphics[width=3.5in]{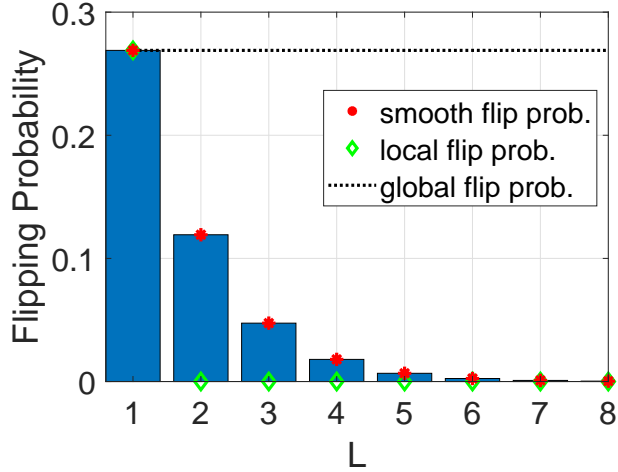}
  \end{center}
  \vspace{-0.25in}
  \label{fig:smooth_flip_prob}
  \caption{Illustration of the smooth flipping probability, ``local flipping probability'', and the global upper bound for DP-SignRP. $\epsilon=1$, $k=1$. $L$ is computed in Algorithm~\ref{alg:DP-signRP-RR-smooth}.}
  \vspace{0.1in}
\end{figure}

The concrete steps of DP-SignRP-RR-smooth are summarized in Algorithm~\ref{alg:DP-signRP-RR-smooth}. Note that, the above two cases can be unified into one, since the first case corresponds to $L_j=1$ which exactly recovers the flipping probability $\frac{e^{\epsilon/k}}{e^{\epsilon/k}+1}$. In Figure~\ref{fig:smooth_flip_prob}, we provide an illustrative example of the smooth flipping probability (red) and the ``local flipping probability'' (green) discussed above. We see that the local flipping probability is only non-zero at $L=1$, i.e., when $s_j$ and $s_j'$ are possible to be different. The proposed smooth flipping probability is non-zero for all $L>0$, but keeps shrinking as $L$ is larger, i.e., as the projected value is farther from $0$.

\begin{theorem} \label{theo:privacy-DPSignRP-RR-smooth}
Algorithm~\ref{alg:DP-signRP-RR-smooth} is $\epsilon$-DP.
\end{theorem}

\begin{proof}
Let us consider a single projection vector $w_j=W_{[:,j]}$. Denote $x_j=w_j^T u$ and $x_j'=w_j^T u'$ for a neighboring data $u'$ of $u$, and $s_j=sign(x_j)$, $s_j'=sign(x_j')$. Also, let $L_j=\lceil \frac{|x_j|}{\beta\max_{i=1,...,p} |W_{ij}|} \rceil$ and $L_j'=\lceil \frac{|x_j'|}{\beta\max_{i=1,...,p} |W_{ij}|} \rceil$. W.l.o.g., we can assume $s_j=1$ by the symmetry of random projection and the symmetry of DP. Consider two cases:
\begin{itemize}
    \item Case I: $L_j\geq 2$. In this case, we know that $s_j'=s_j$, i.e., the change from $u$ to $u'$ will not change the sign of the projection. Thus, in Algorithm~\ref{alg:DP-signRP-RR-smooth}, we have
    \begin{align*}
        \frac{Pr(\tilde s_j=1)}{Pr(\tilde s_j'=1)}=\exp(\frac{L_j-L_j'}{k}\epsilon) \frac{\exp(\frac{L_j'}{k}\epsilon)+1}{\exp(\frac{L_j}{k}\epsilon)+1}.
    \end{align*}
    By the definition of $\beta$-adjacency, $|L_j-L_j'|$ equals either $0$ or $1$. When $L_j=L_j'$, $\frac{Pr(\tilde s_j=1)}{Pr(\tilde s_j'=1)}=1$. When $L_j-L_j'=1$, we have
    \begin{align*}
        \frac{Pr(\tilde s_j=1)}{Pr(\tilde s_j'=1)}=\frac{\exp(\frac{L_j}{k}\epsilon)+\exp(\frac{1}{k}\epsilon)}{\exp(\frac{L_j}{k}\epsilon)+1}.
    \end{align*}
    Hence, we have $1\leq\frac{Pr(\tilde s_j=1)}{Pr(\tilde s_j'=1)}\leq e^{\frac{\epsilon}{k}}$ by the numeric identity $1\leq \frac{a+c}{b+c}\leq \frac{a}{b}$ for $a\geq b>0$ and $c>0$. Thus, by symmetry, $e^{-\frac{\epsilon}{k}}\leq \frac{Pr(\tilde s_j=1)}{Pr(\tilde s_j'=1)}\leq e^{\frac{\epsilon}{k}}$. On the other hand,
    \begin{align*}
        \frac{Pr(\tilde s_j=-1)}{Pr(\tilde s_j'=-1)}=\frac{\exp(\frac{L_j'}{k}\epsilon)+1}{\exp(\frac{L_j}{k}\epsilon)+1}.
    \end{align*}
    Similarly, when $L_j=L_j'$, the ratio equals $1$. When $L_j=L_j'-1$, we have $\frac{Pr(\tilde s_j=-1)}{Pr(\tilde s_j'=-1)}\leq\exp(\frac{L_j'}{k}\epsilon-\frac{L_j}{k}\epsilon)=e^{\frac{\epsilon}{k}}$. By symmetry we obtain $e^{-\frac{\epsilon}{k}}\leq \frac{Pr(\tilde s_j=-1)}{Pr(\tilde s_j'=-1)}\leq e^{\frac{\epsilon}{k}}$.

    \item Case II: $L_j=1$. In this case, $s_j$ might be different from $s_j'$. First, if $L_j'=2$, then the above analysis also applies that $\frac{Pr(\tilde s_j=1)}{Pr(\tilde s_j'=1)}$ and $\frac{Pr(\tilde s_j=-1)}{Pr(\tilde s_j'=-1)}$ are both lower and upper bounded by $e^{-\frac{\epsilon}{k}}$ and $e^{\frac{\epsilon}{k}}$, respectively. It suffices to examine the case when $L_j'=1$. In this case, if $s_j'=s_j=1$ then the probability ratios simply equal 1. If $s_j'=-1$, we have
    \begin{align*}
        \frac{Pr(\tilde s_j=1)}{Pr(\tilde s_j'=1)}=\frac{\frac{\exp(\frac{\epsilon}{k})}{\exp(\frac{\epsilon}{k})+1}}{\frac{1}{\exp(\frac{\epsilon}{k})+1}}=e^{\frac{\epsilon}{k}},\quad \frac{Pr(\tilde s_j=-1)}{Pr(\tilde s_j'=-1)}=\frac{\frac{1}{\exp(\frac{\epsilon}{k})+1}}{\frac{\exp(\frac{\epsilon}{k})}{\exp(\frac{\epsilon}{k})+1}}=e^{-\frac{\epsilon}{k}}.
    \end{align*}
\end{itemize}
Combining two cases, we have that $\log \frac{Pr(\tilde s_j=t)}{Pr(\tilde s_j'=t)}\leq \frac{\epsilon}{k}$, for $t=-1,1$, and for all $j=1,...,k$. That is, each single perturbed sign achieves $\frac{\epsilon}{k}$-DP. Since the $k$ projections are independent, by Theorem~\ref{theo:composition}, we know that the output bit vector $\tilde s=[\tilde s_1,...,\tilde s_k]$ is $\epsilon$-DP as claimed.
\end{proof}

\noindent\textbf{Comparison with DP-SignRP-RR (Algorithm~\ref{alg:DP-signRP-RR}).} In terms of algorithm design and privacy, DP-SignRP-RR-smooth (Algorithm~\ref{alg:DP-signRP-RR-smooth}) has the following two advantages over the standard DP-SignRP-RR method (Algorithm~\ref{alg:DP-signRP-RR}):
\begin{enumerate}
    \item DP-SignRP-RR-smooth does not require (assume) a lower bound $m$ on the data norms;

    \item DP-SignRP-RR-smooth achieves $\epsilon$-DP, while DP-SignRP-RR can only guarantee $(\epsilon,\delta)$-DP.
\end{enumerate}
Regarding the second point, as mentioned in Remark~\ref{remark:DP-SignRP-RR}, Algorithm~\ref{alg:DP-signRP-RR} can also achieve $\epsilon$-DP if we set the flipping probability to be $\frac{1}{e^{\epsilon/k}+1}$. Note that this equals the flipping probability in Algorithm~\ref{alg:DP-signRP-RR-smooth} when $L_j=1$, which is the ``worst'' global bound on the flipping probability. In many cases, a good proportion of $L_j$ $j=1,...,k$ would be greater than 1. Therefore, when both are $\epsilon$-DP, DP-SignRP-RR-smooth requires a much smaller flipping probability than DP-SignRP-RR.

\vspace{0.1in}
\subsection{DP-SignRP with Rademacher Projections}  \label{sec:rademacher-SignRP}
\vspace{0.1in}

At this point, it should be clear that the flipping probability of DP-SignRP (both DP-SignRP-RR and DP-SignRP-RR-smooth) essentially depends on how concentrated the projected data is around zero. Particularly, $N_+$ in Algorithm~\ref{alg:DP-signRP-RR}, as given in Proposition~\ref{prop:num-of-changes}, is a high probability upper bound on a Binomial random variable with success probability $Pr(\beta\max_{i=1,...,p}|w_i|\geq |w^Tu|)$ with $w\sim N(0,1)$. In Algorithm~\ref{alg:DP-signRP-RR-smooth}, $L_j=\lceil \frac{|w_j^Tu|}{\beta\max_{i=1,...,p} |W_{ij}|} \rceil$. For both quantities, a smaller value leads to a smaller sign flipping probability and thus better utility.

\vspace{0.2in}

\noindent\textbf{$N_+$ in DP-SignRP-RR.} We first consider the $N_+$ in Algorithm~\ref{alg:DP-signRP-RR}, which determines the flipping probability $\frac{1}{e^{\epsilon/N_+}+1}$. Particularly, $N_+$ in Algorithm~\ref{alg:DP-signRP-RR}, as given in Proposition~\ref{prop:num-of-changes}, is a high probability upper bound on a Binomial random variable with success probability \begin{align}\label{eqn:prob_flip}
P_+ = Pr\left(\beta\max_{i=1,...,p}|w_i|\geq |w^Tu|\right),
\end{align}
where $w$ is the $p$-dimensional projection vector. When $w_i$ is sampled from the Rademacher distribution, i.e., $w_i\in\{-1, +1\}$ with equal probabilities, the probability calculation can be simplified:
\begin{align}\label{eqn:prob_flip_b}
&P_{+,b} = Pr\left(\beta\max_{i=1,...,p}|w_i|\geq |\sum_{i=1}^p w_iu_i|\right)
=Pr\left(\beta\geq |\sum_{i=1}^p w_iu_i|\right)\approx 2\Phi\left(\frac{\beta}{\|u\|}\right)-1.
\end{align}
Based on the central limit theorem,  the normal approximation~\eqref{eqn:prob_flip_b} is  accurate unless $p$ is very small. Recall that, when $w_i$'s are sampled from the Gaussian distribution, we have already calculated an upper bound in~\eqref{eqn:E[I]}, which is re-written as below:
\begin{align}\label{eqn:prob_flip_g}
P_{+,g} = Pr\left(\beta\max_{i=1,...,p}|w_i|\geq |\sum_{i=1}^p w_iu_i|\right)\leq \int_0^\infty 2p[2\Phi(t)-1]^{p-1}[2\Phi(\beta t/\|u\|)-1]\phi(t) dt.
\end{align}

Next, we provide a simulation study to justify the approximation and compare different distributions in terms of their impact on the probability~\eqref{eqn:prob_flip}, for $\beta=1$ as well as $\beta=0.1$. For simplicity, we simulate the data as a $p$-dimensional vector of uniform random numbers sampled from $unif[-1,1]$. We experiment with five different choices of $w$: the standard Gaussian, the uniform, the ``very sparse'' distribution~\eqref{eqn:vsrp} with $s=1$, $s=3$, and $s=10$. We vary $p$ from 10 to 1000. For each case, we repeat the simulations $10^7$ times to ensure sufficient accuracy. Figure~\ref{fig:prob_flip_theory} verifies that the two approximations \eqref{eqn:prob_flip_b} and \eqref{eqn:prob_flip_g} are  accurate. In Figure~\ref{fig:prob_flip}, we provide the curves for more types of projection matrices. From both figures, we see that using the Rademacher projection can considerably reduce (\ref{eqn:prob_flip}) compared with Gaussian (and other) projections, leading to a smaller $N_+$ value. This typically implies better utility.

\begin{figure}[t]
    \centering
    \includegraphics[width=3.2in]{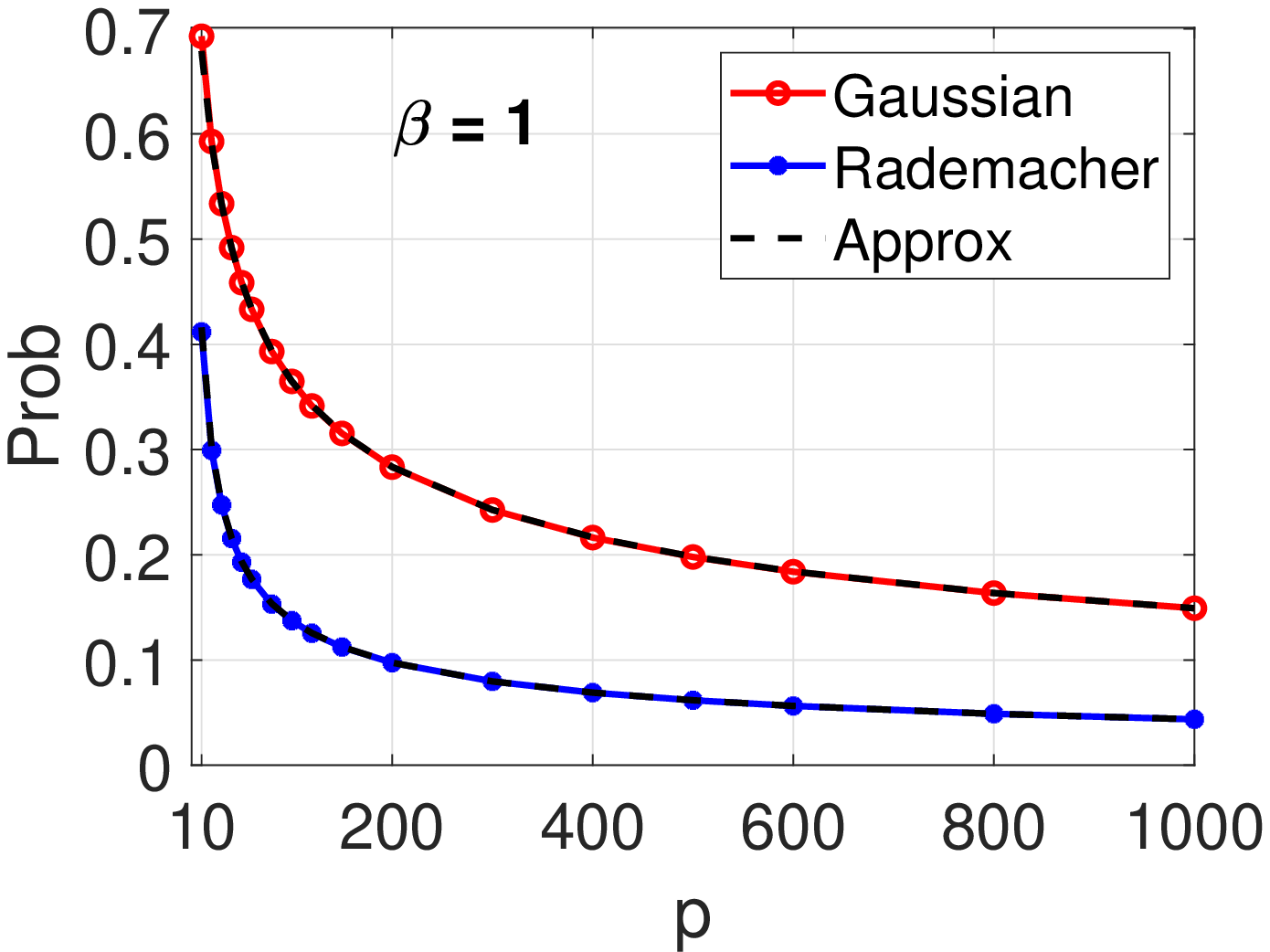}
    \includegraphics[width=3.2in]{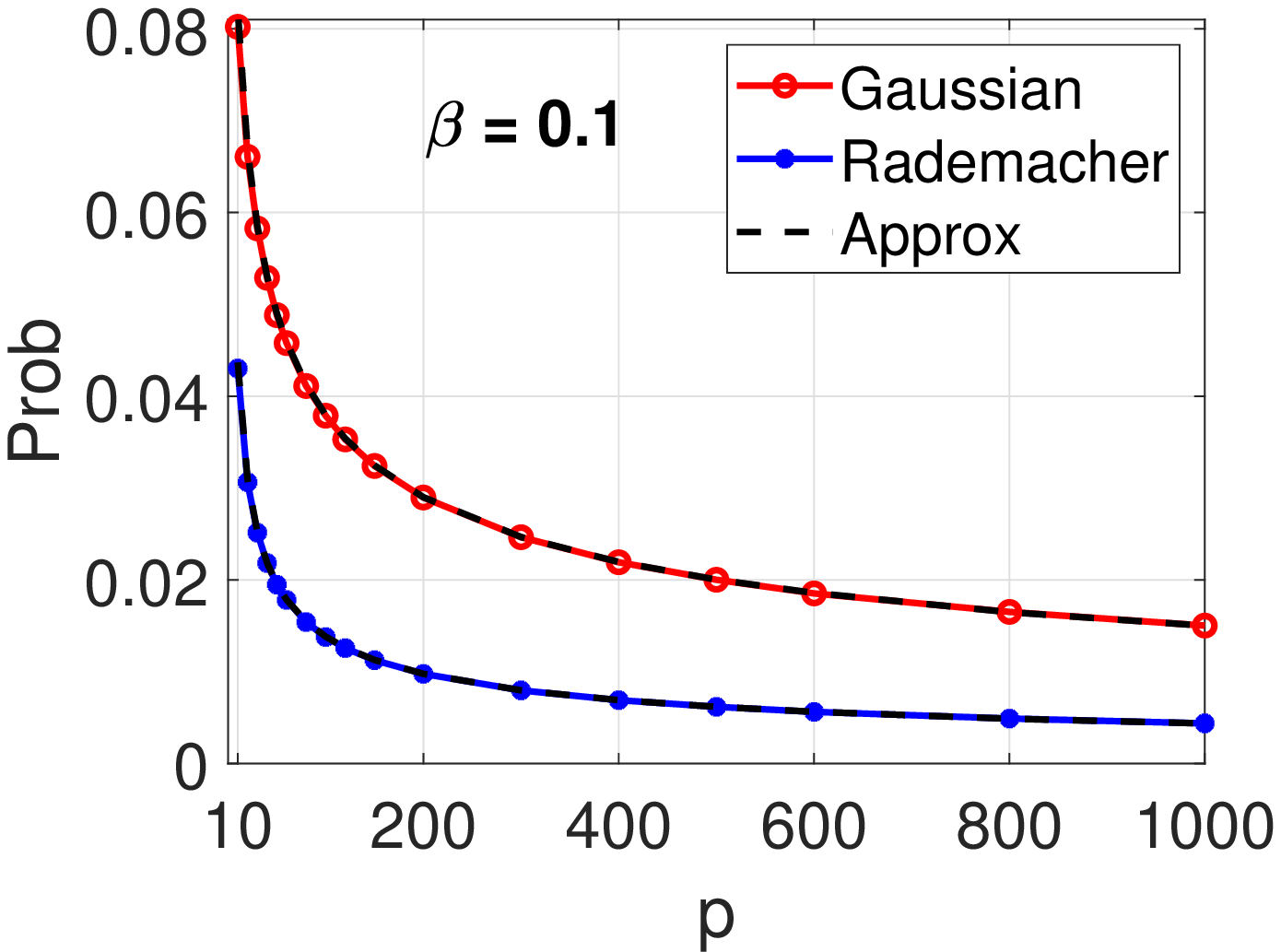}

\vspace{-0.1in}

    \caption{Simulations for evaluating~\eqref{eqn:prob_flip_b} and~\eqref{eqn:prob_flip_g}, using two  choices for $w$: the Gaussian distribution and the  Rademacher distribution (i.e.,~\eqref{eqn:vsrp} with $s=1$). We plot the two upper bounds~\eqref{eqn:prob_flip_b} and~\eqref{eqn:prob_flip_g} as black dashed curves, which essentially both overlap with their corresponding simulations.  }
    \label{fig:prob_flip_theory}
\end{figure}

\begin{figure}[h]
    \centering
    \includegraphics[width=3.2in]{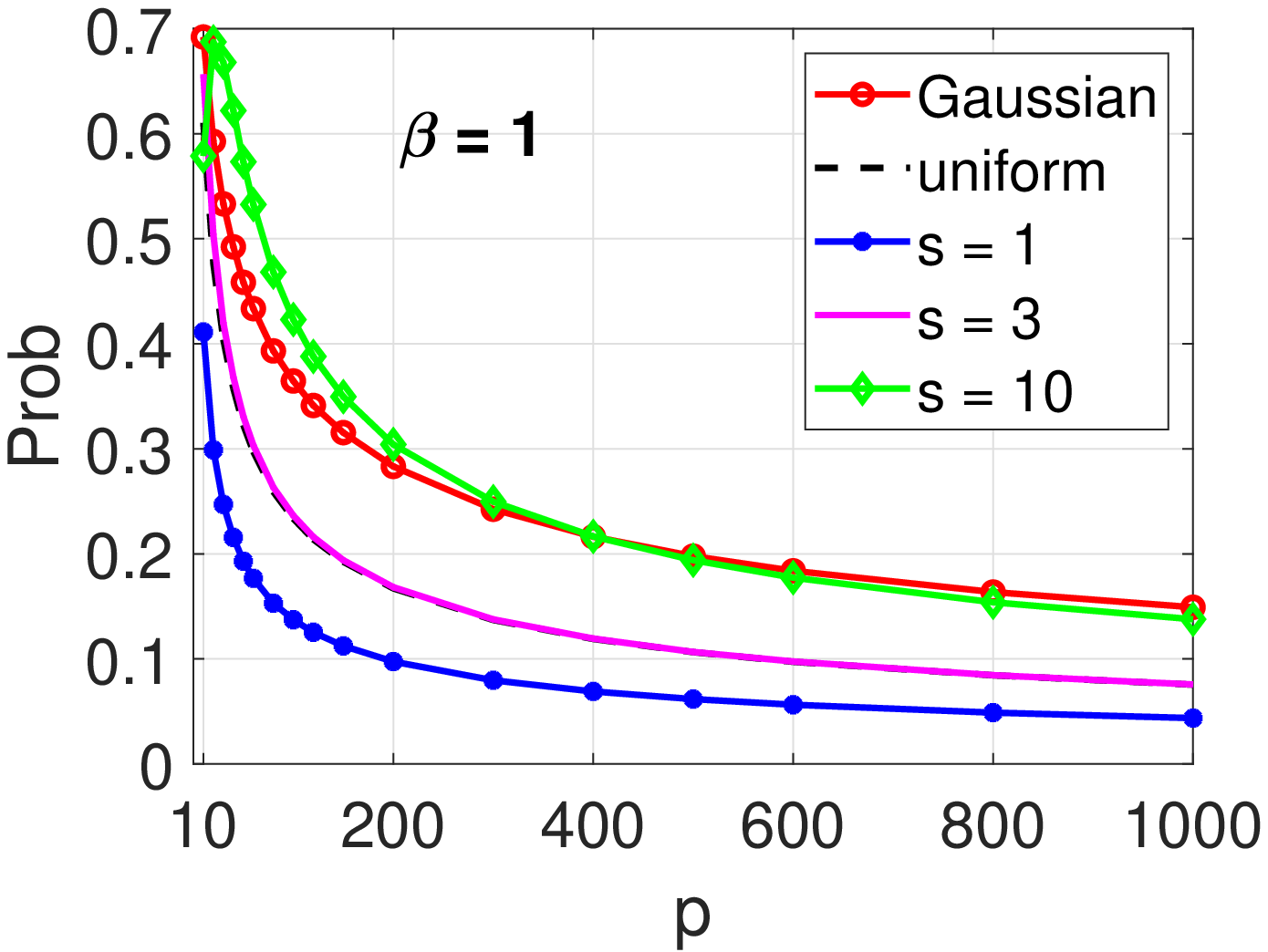}
    \includegraphics[width=3.2in]{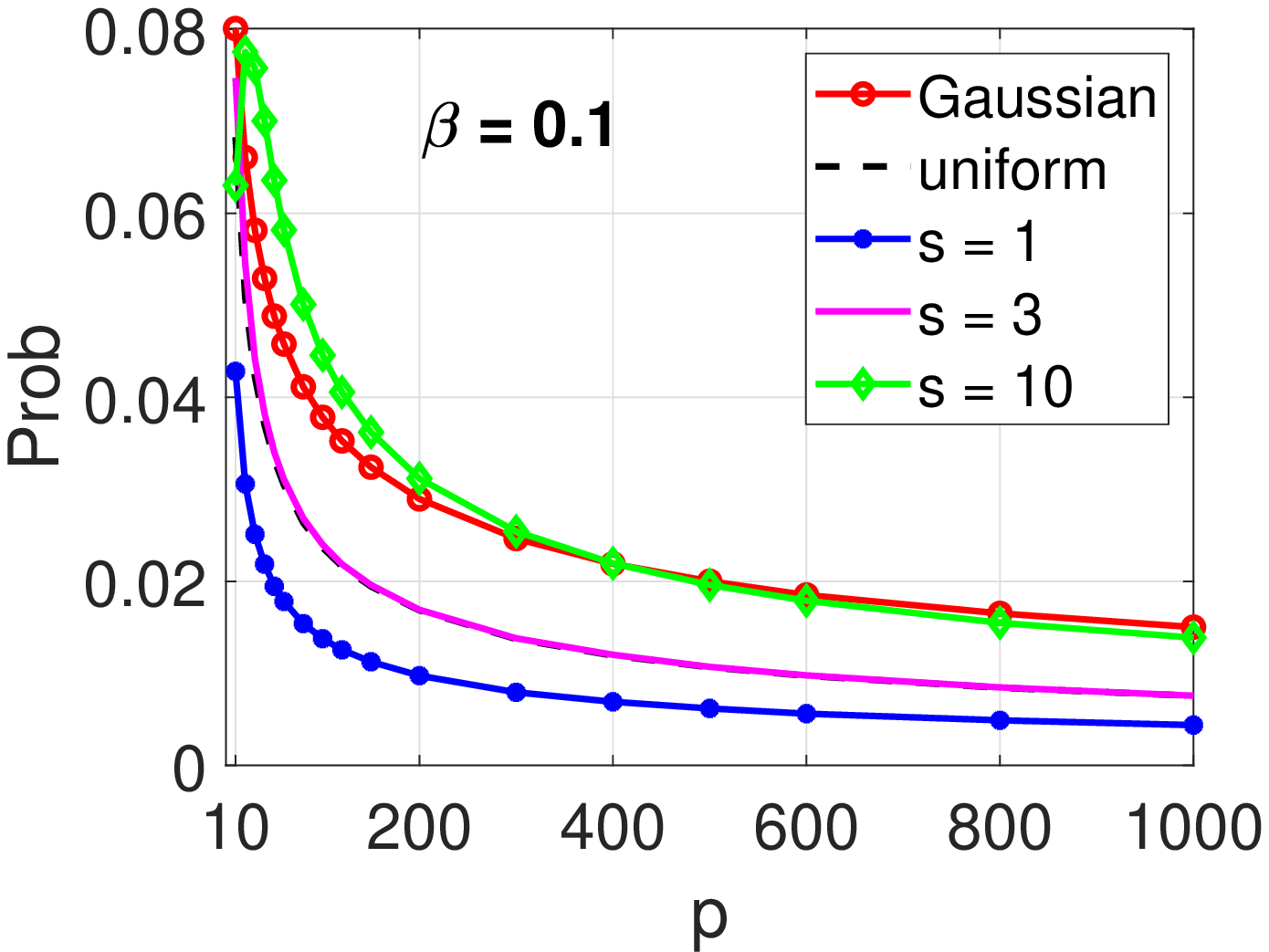}

\vspace{-0.1in}

    \caption{Simulations (same as in Figure~\ref{fig:prob_flip}) for evaluating~\eqref{eqn:prob_flip}, using five different choices for $w$: the Gaussian, the uniform, the ``very sparse'' distribution~\eqref{eqn:vsrp} with $s=1, 3$ and $10$. $s=1$ is the Rademacher distribution. The data vector is simulated by sampling each entry from $unif[-1,1]$. }
    \label{fig:prob_flip}
\end{figure}

\newpage

\noindent\textbf{$L_j$ in DP-SignRP-RR-smooth.} Similarly, we numerically evaluate the $L_j$ in Algorithm~\ref{alg:DP-signRP-RR-smooth}. We run Algorithm~\ref{alg:DP-signRP-RR-smooth} with $k=512$, which gives $512$ $L_j$ values. In Figure~\ref{fig:Lj_proportion}, we plot the proportion (or the approximated distribution) of the values of $L_j$ among $k$ projections. As we see, Rademacher projection produces the least number of small $L_j$ values and the largest number of higher $L_j$ values. As the smooth flipping probability equals $\frac{1}{\exp(\frac{L_j}{k}\epsilon)+1}$, larger $L_j$ leads to smaller probability of sign flipping. Hence, Rademacher is again the best choice for the projection matrix.

\begin{figure}[h]
    \centering
    \includegraphics[width=3.2in]{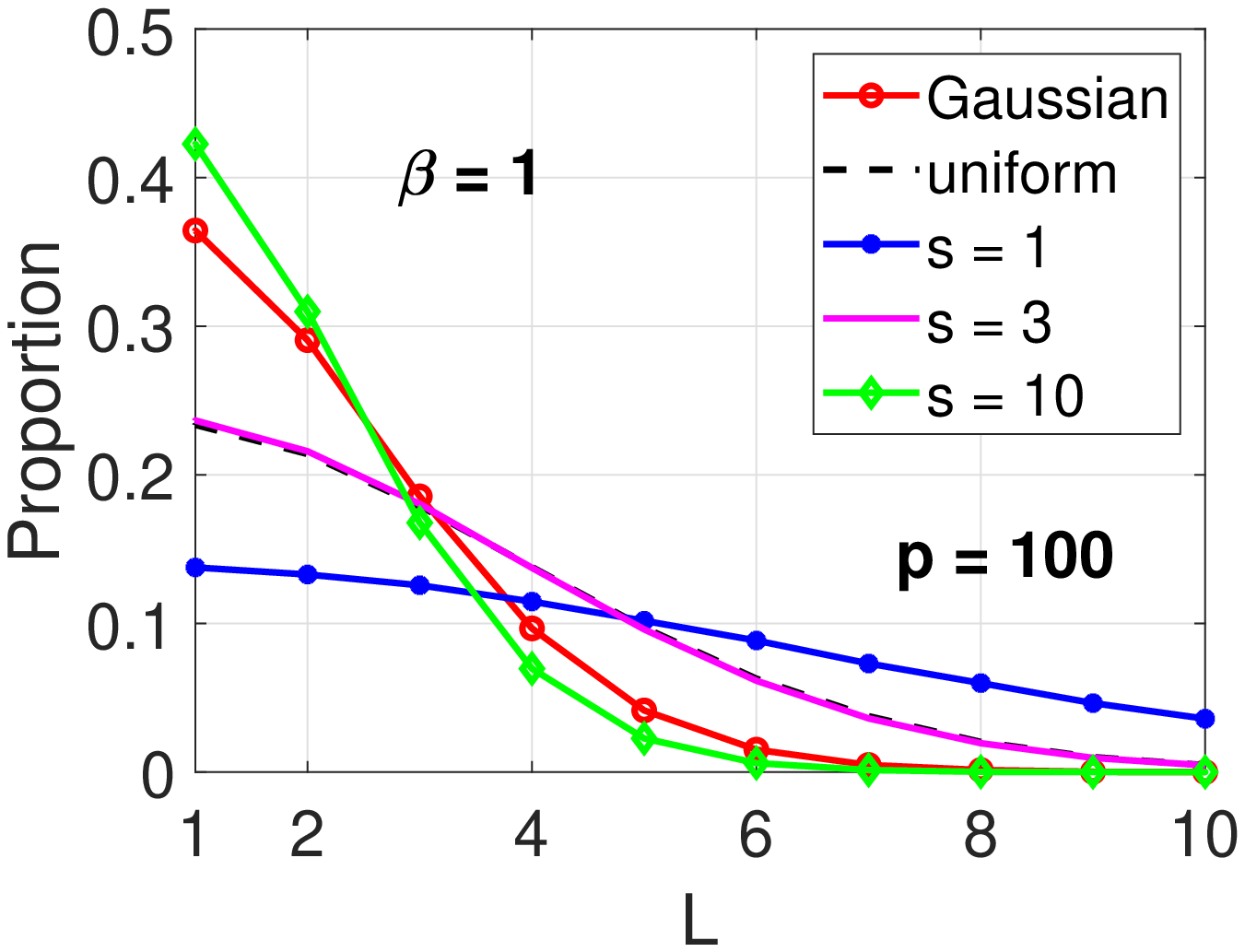}
    \includegraphics[width=3.2in]{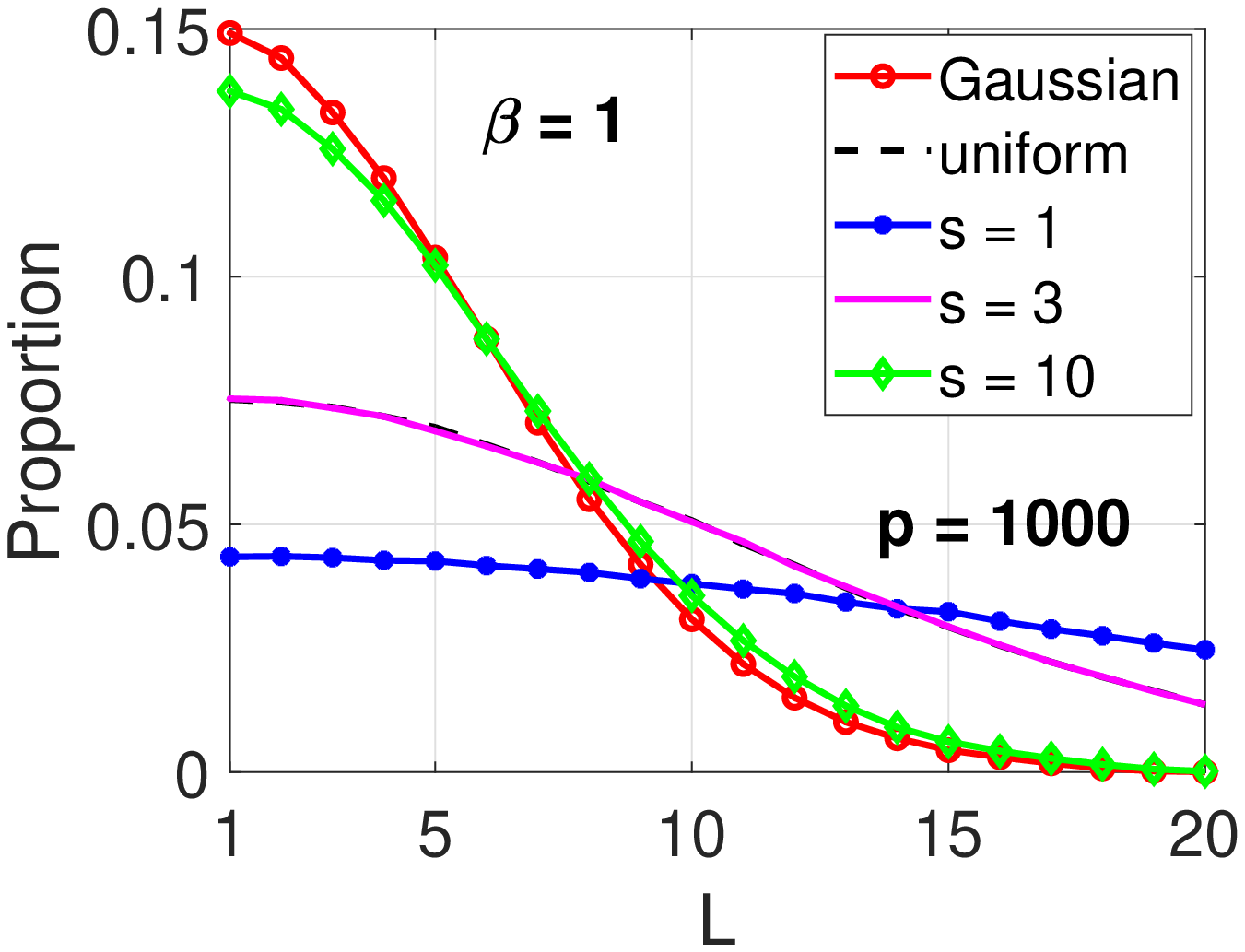}

\vspace{-0.1in}

    \caption{Simulations for evaluating $L_j$ in Algorithm~\ref{alg:DP-signRP-RR-smooth}, using different choices for $w$. $s=1$ is the Rademacher distribution. Left: $p=100$, right: $p=1000$. The $y$-axis is the proportion (normalized histogram) of the values of all the $L_j$, $j=1,...,k$ computed using $k=512$ projected samples.}
    \label{fig:Lj_proportion}
\end{figure}



\vspace{-0.1in}

\subsection{DP-SignOPORP with Smooth Flipping Probability}  \label{sec:DP-signOPORP}

Similarly, we may also take the sign of OPORP studied before, and make it differentially private. The smoothed flipping probability can also be adopted in this case. As before, we present the variant with Rademacher projection for conciseness. Algorithm~\ref{alg:DP-signOPORP-RR} provides the general framework, covering two variants. The first one, DP-SignOPORP-RR, is based on the standard randomized response (RR) technique. The second variant, DP-SignOPORP-RR-smooth, is an improved version with our proposed ``smooth flipping probability''.

\begin{algorithm}[h]
{
    \vspace{0.05in}
    \textbf{Input:} Data $u\in[-1,1]^p$; $\epsilon>0$; Number of projections~$k$

    \vspace{0.05in}

    \textbf{Output:}   Differentially private sign OPORP

    \vspace{0.05in}

    Apply Algorithm~\ref{alg:OPORP} with a random Rademacher projection vector to get the OPORP $x$

    \vspace{0.05in}

    \nonl\textbf{DP-SignOPORP-RR:}

    Compute $\tilde s_j=\begin{cases}
    sign(x_j), & \text{with prob.}\ \frac{e^{\epsilon}}{e^{\epsilon}+1}\\
    -sign(x_j), & \text{with prob.}\ \frac{1}{e^{\epsilon}+1}
    \end{cases}$ for $j=1,...,k$

    \nonl\textbf{DP-SignOPORP-RR-smooth:}

    Compute $L_j=\lceil \frac{|x_j|}{\beta} \rceil$ for $j=1,...,k$

    Compute $\tilde s_j=\begin{cases}
    sign(x_j), & \text{with prob.}\ \frac{e^{\epsilon_j'}}{e^{\epsilon_j'}+1}\\
    -sign(x_j), & \text{with prob.}\ \frac{1}{e^{\epsilon_j'}+1}
    \end{cases}$ for $j=1,...,k$, with $\epsilon_j'=L_j\epsilon$

    For $\tilde s_j=0$, assign a random coin in $\{-1,1\}$

    Return $\tilde s$ as the DP-SignOPORP of $u$
    }
    \caption{DP-SignOPORP-RR and DP-SignOPORP-RR-smooth}
    \label{alg:DP-signOPORP-RR}
\end{algorithm}

\vspace{0.1in}
\noindent\textbf{Benefit of binning.} In Algorithm~\ref{alg:DP-signOPORP-RR}, we see that the key difference compared with DP-SignRP methods is that, the $\frac{1}{N_+}$ (in Algorithm~\ref{alg:DP-signRP-RR}) or $\frac{1}{k}$ (in Algorithm~\ref{alg:DP-signRP-RR-smooth}) factor is removed from the flipping probability formulas, which is a significant advantage in terms of privacy. This improvement is a result of the binning step in OPORP. By Definition~\ref{def:neighbor}, two neighboring data $u'$ and $u$ differ in one coordinate. For SignRP, since each projected data is an aggregation of the whole data vector, when we switch from $u$ to $u'$, (in principle) all $k$ projections are possible to change. In contrast, in OPORP, since each data entry appears in only one bin, $u'$ will only cause exactly one projected output to change, leaving all other bins untouched. This provides the intuition on the source of privacy gain of DP-SignOPORP.

\begin{theorem}  \label{theo:DP-signOPORP privacy}
Both variants in Algorithm~\ref{alg:DP-signOPORP-RR} are $\epsilon$-DP.
\end{theorem}
\begin{proof}
The proof follows the idea of the proofs of DP-SignRP methods, but we need to additionally consider the empty bins. Denote $x$ and $x'$ as the OPORP from Algorithm~\ref{alg:OPORP} using a same random vector $w$, and let $s=sign(x)$, $s'=sign(x')$. Suppose $u$ and $u'$ differ in dimension $i$, and dimension $i$ is assigned to the $j^*$-th bin with $j^*=\lceil \pi(i)/(p/k)\rceil$, where $\pi: [p]\mapsto [p]$ is the permutation in OPORP. For any $y\in\{-1,1\}^k$, it is easy to see that to compute $\log\frac{Pr(\tilde s=y)}{Pr(\tilde s'=y)}$, it suffices to look at the $j^*$-th output sample because other probabilities cancel out. For the $j^*$-th sample, when $s_{j^*}\neq 0$ and $s_{j^*}'\neq 0$, by the same arguments as in the proof of Theorem~\ref{theo:privacy-DPSignRP-RR-smooth}, we know that $e^{-\epsilon}\leq \frac{Pr(s_{j^*})=a}{Pr(s_{j^*}')=a}\leq e^\epsilon$, $a\in\{-1,1\}$, for both variants (DP-SignOPORP-RR and DP-SignOPORP-RR-smooth). When one of $s_{j^*}$ and $s_{j^*}'$ equals $0$, it is also easy to see that (assume $s_{j^*}=1$ and $s_{j^*}'$ is a random coin)
\begin{align*}
    1\leq \frac{Pr(s_{j^*})=1}{Pr(s_{j^*}')=1}=\frac{2e^\epsilon}{e^\epsilon+1}\leq e^\epsilon,\quad e^{-\epsilon}\leq \frac{Pr(s_{j^*})=-1}{Pr(s_{j^*}')=-1}=\frac{2}{e^\epsilon+1}\leq 1.
\end{align*}
This holds for both variants. In particular, this case corresponds to $L_{j^*}=1$ for the smooth flipping probability. This proves the theorem.
\end{proof}

\vspace{0.1in}
\noindent\textbf{Multiple repetitions.} In Algorithm~\ref{alg:DP-signOPORP-RR} Line 7, we see that if OPORP generates an empty bin (i.e., $x_j=0$), we must assign a random sign to maintain DP. This would undermine the utility since a random coin does not provide any useful information. Therefore, it is desirable to avoid empty bins for better utility. One simple way is to repeat the DP-SignOPORP (with smaller $k$) for $t>1$ times, and concatenate the output vectors. Repetition is a standard strategy for count-sketch~\citep{charikar2004finding}. For example, if the target $k=256$, we may run Algorithm~\ref{alg:DP-signOPORP-RR} for $t=4$ times, each time with $256/4=64$ projected values and privacy budget $\epsilon/4$. Since we are using fewer bins per run, the number of empty bins will be reduced. After concatenating the 4 output vectors, we still get a 256-dimensional bit vector with $\epsilon$-DP by the composition theorem. In the experiments, we will see that this strategy may improve the overall performance of DP-SignOPORP.

\section{Experiments}  \label{sec:experiment}

We present a set of experiments on retrieval and classification tasks to demonstrate the performance of different DP-RP and DP-SignRP methods. Specifically, we will compare the following algorithms:

\begin{itemize}
    \item Raw-data-G-OPT: the method of directly adding optimal Gaussian noise with sensitivity $\beta$ to the original data vectors.

    \item DP-RP-G (Algorithm~\ref{alg:DP-RP-gaussian-noise} + Theorem~\ref{theo:kenthapadi}): Gaussian DP-RP with the Gaussian mechanism.

    \item DP-RP-G-OPT (Algorithm~\ref{alg:DP-RP-gaussian-noise} + Theorem~\ref{theo:gauss-optimal}): DP-RP with Gaussian random projection matrix and the optimal Gaussian noise.

    \item DP-RP-G-OPT-B (Algorithm~\ref{alg:DP-RP-G-OPT-B}): DP-RP with Rademacher random projection matrix and the optimal Gaussian noise mechanism.

    \item DP-OPORP (Algorithm~\ref{alg:DP-OPORP}): OPORP algorithm with optimal Gaussian noise.

    \item DP-SignOPORP-RR and DP-SignOPORP-RR-smooth (Algorithm~\ref{alg:DP-signOPORP-RR}): Signed OPORP with standard randomized response and  smooth flipping probability, respectively.
\end{itemize}

As noted in Section~\ref{sec:DP-signOPORP}, due to the feature binning procedure, DP-SignOPORP is substantially better than DP-SignRP in terms of privacy protection. Therefore, we present the better-performing DP-SignOPORP method in our experiments. For conciseness, we will only present results for $\beta=1$ and $\epsilon\in [0.1,20]$. The chosen $\epsilon$ values cover the common range of $\epsilon$ in many DP applications (e.g., \citet{haeberlen2014differential,kenny2021use}). Also, we fix $\delta=10^{-6}$ for approximate $(\epsilon,\delta)$-DP algorithms.

\subsection{Similarity Search}

We first test the methods in similarity search problems, which is an important application of RP and SignRP in industrial applications. In this experiment, we use two standard image retrieval datasets, MNIST~\citep{lecun1998gradient} and CIFAR~\citep{krizhevsky2009learning}. The MNIST dataset contains 60000 28$\times$28 handwritten digits as the training set, and 10000 digits for testing. The CIFAR dataset includes 60000 natural images with size 32$\times$32 (gray-scale) in total, with a 50000/10000 train-test split. On these two datasets, the pixel values are between 0 and 1 by nature. For both datasets, we treat the training set as the database, and use the samples from the test set as the queries. We set the true neighbors for each query as the top-50 sample vectors with the highest cosine similarity to the query. To search with DP-RP (and DP-OPORP), we estimate the cosine between a query and a sample point by the cosine between their corresponding output noisy projected values. For DP-SignRP (and DP-SignOPORP), we compute the Hamming distances between the private bit vectors between the query and the database points. For both approaches, we return the samples with highest cosine or smallest Hamming distance to the query. The evaluation metrics are precision and recall, defined as
\begin{align*}
    &\text{precision}@R=\frac{\text{\# of true positives in}\ R \ \text{retrieved points}}{R}, \\
    &\text{recall}@R=\frac{\text{\# of true positives in}\ R \ \text{retrieved points}}{\text{\# of gold-standard neighbors}},
\end{align*}
and in our setting, the number of gold-standard neighbors per query is 50.  Our presented results are also averaged over 10 independent repetitions.

\begin{figure}[t]
\centering

    \mbox{
    \includegraphics[width=2.2in]{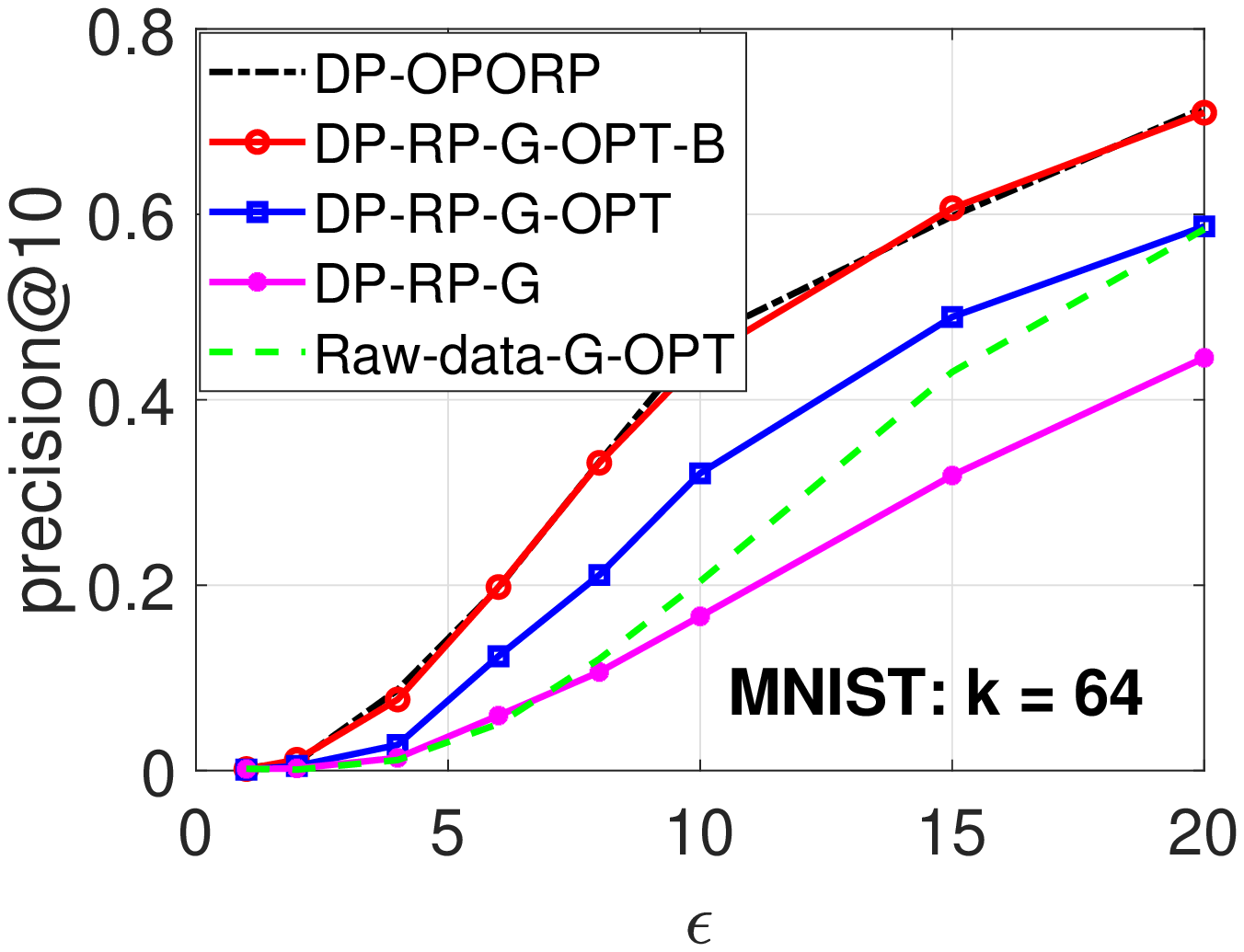} \hspace{-0.15in}
    \includegraphics[width=2.2in]{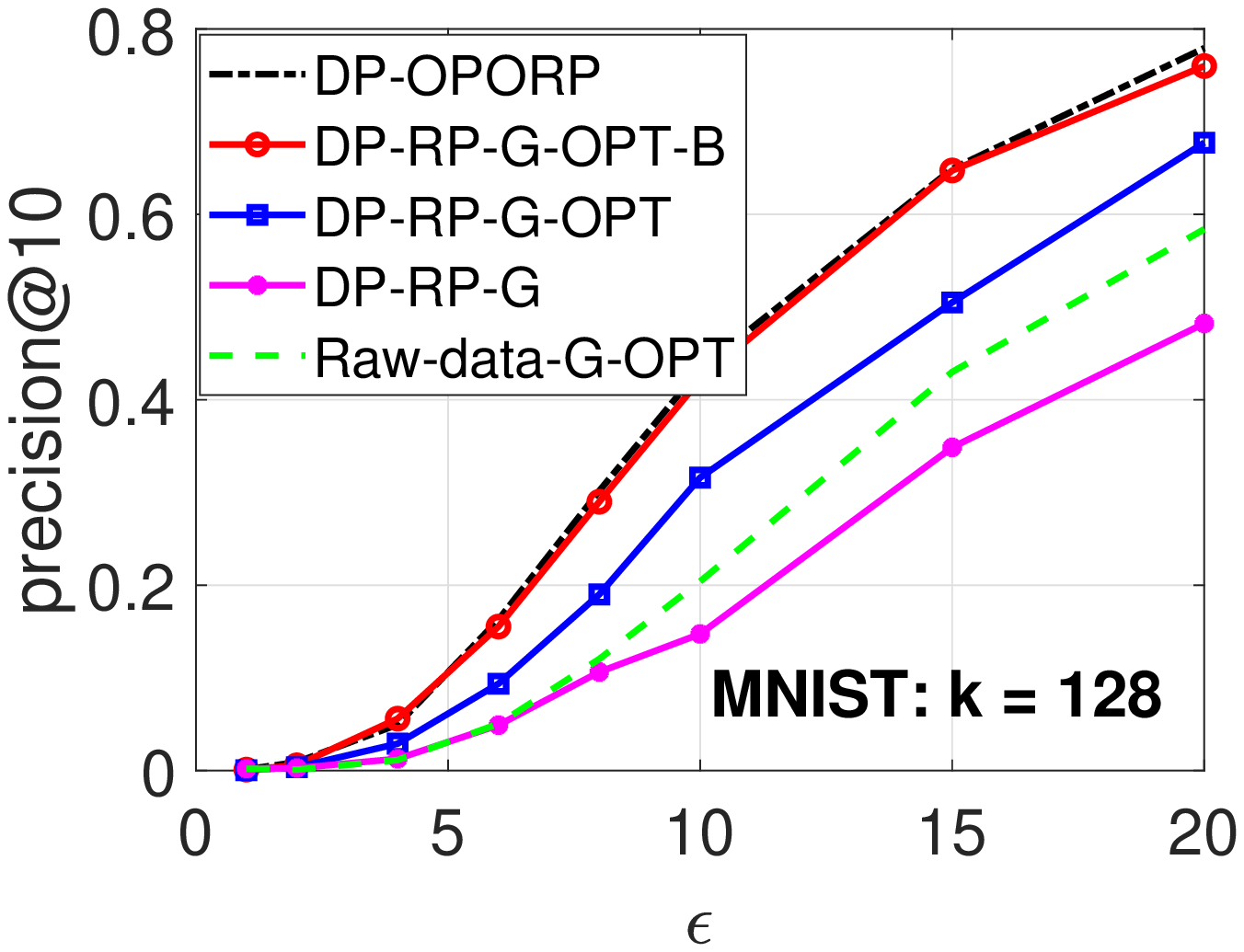}
    \includegraphics[width=2.2in]{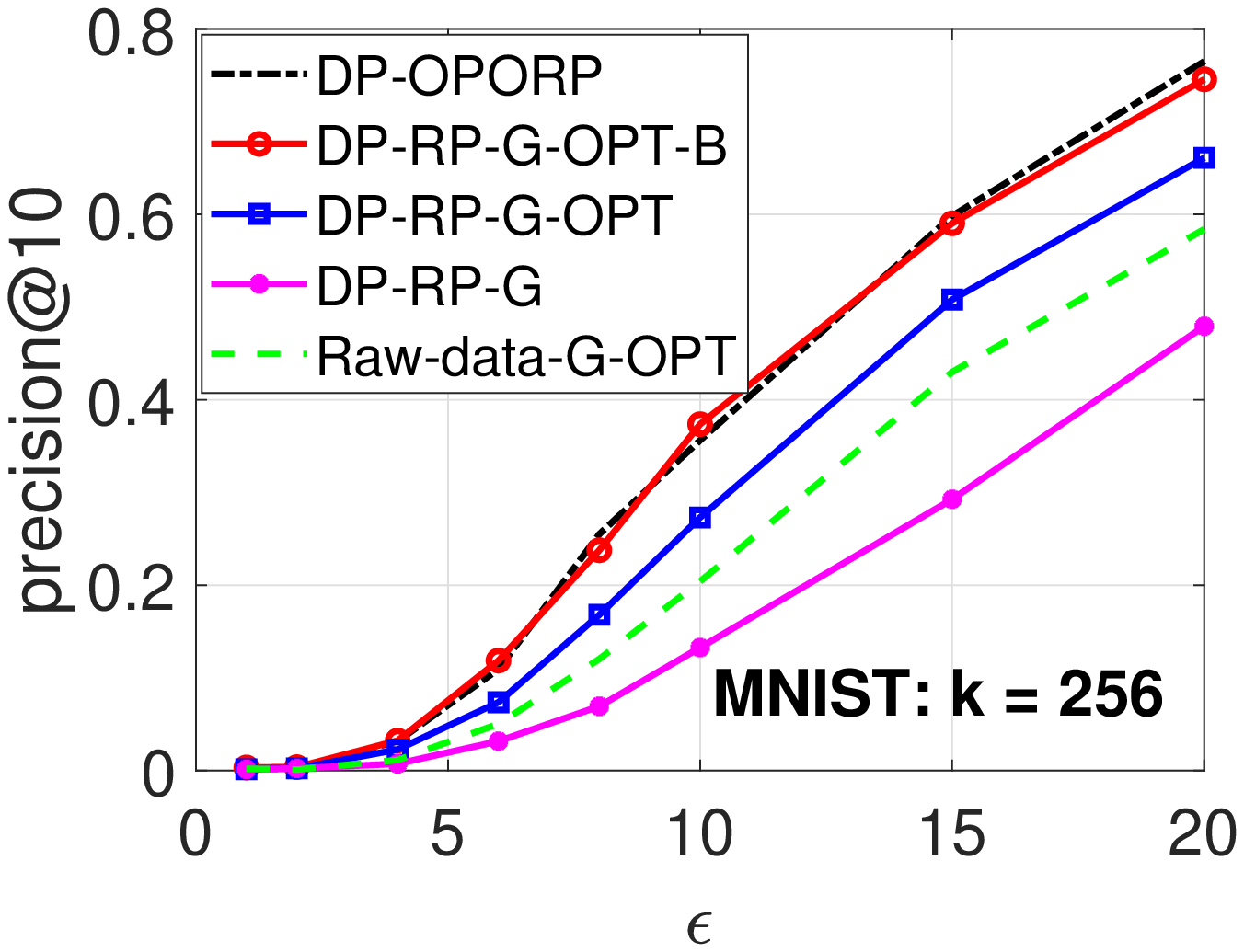}
    }

    \mbox{
    \includegraphics[width=2.2in]{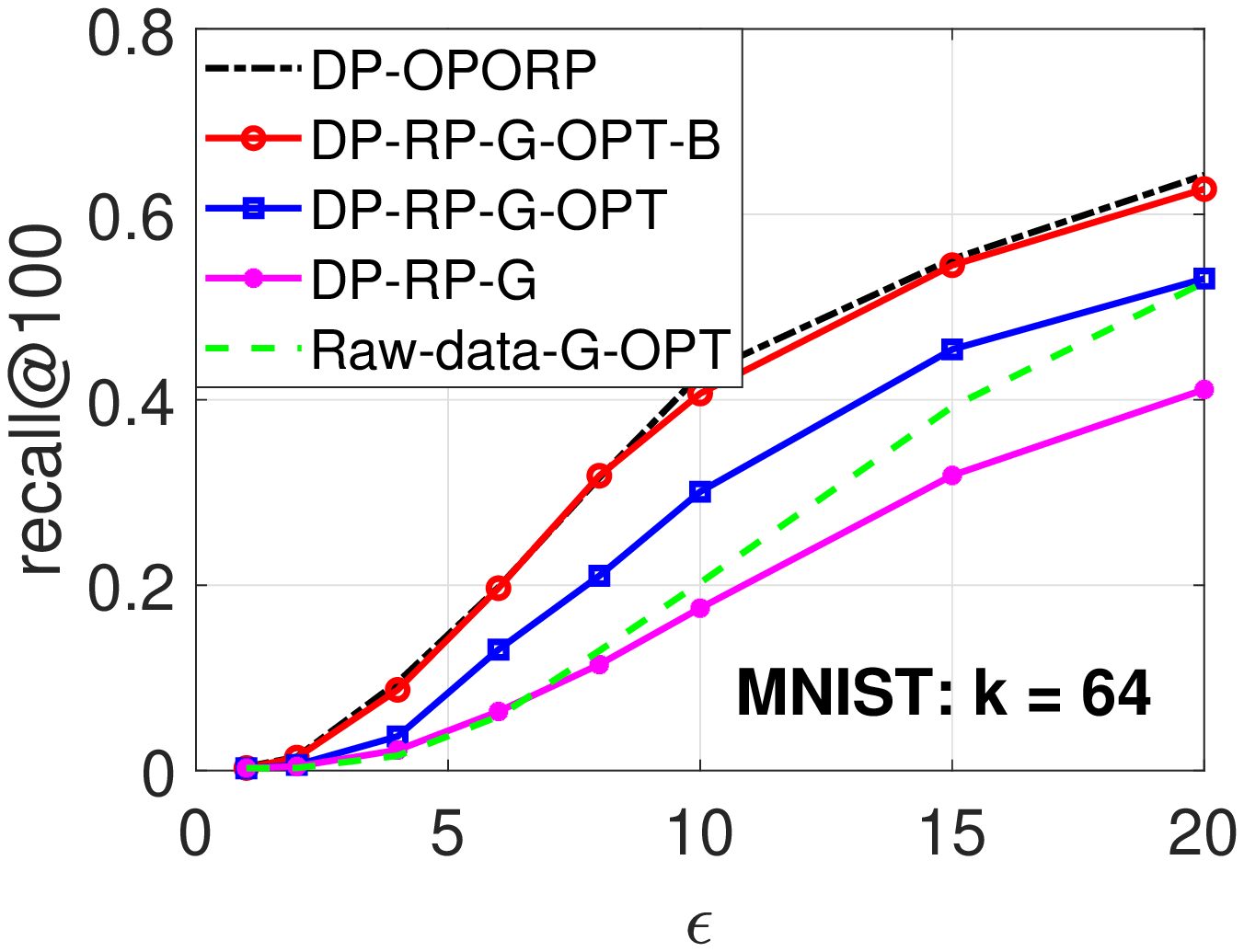} \hspace{-0.15in}
    \includegraphics[width=2.2in]{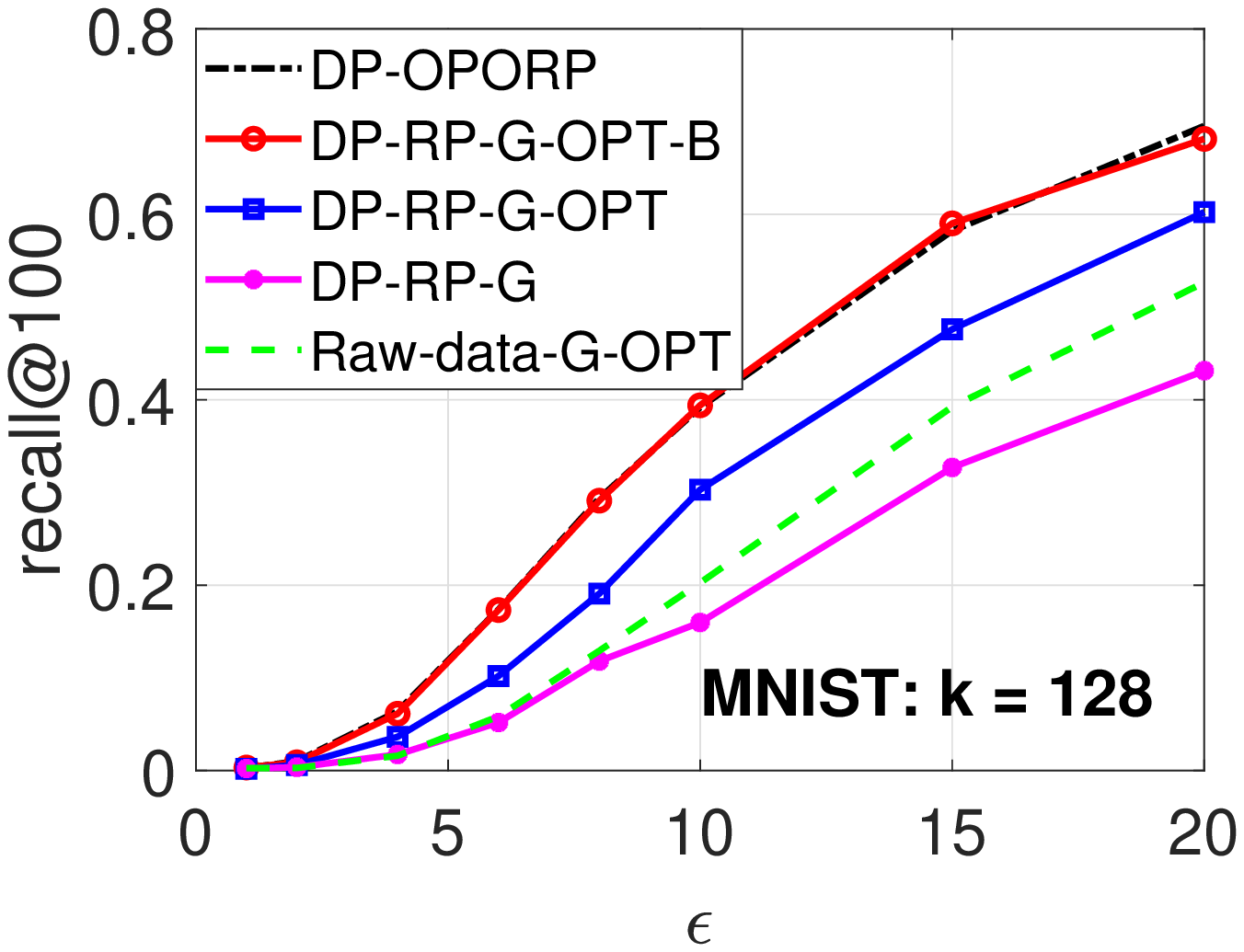}
    \includegraphics[width=2.2in]{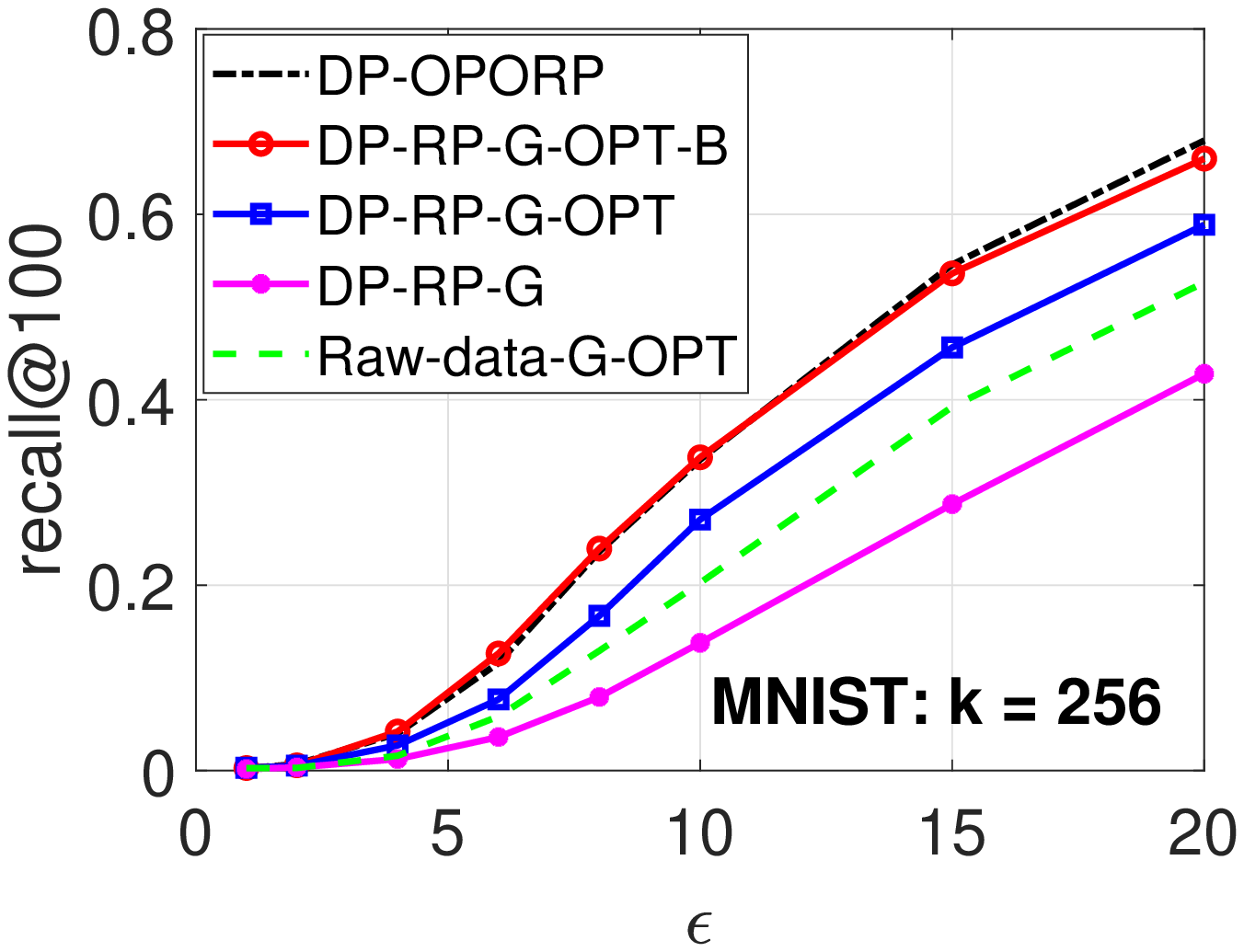}
    }

\vspace{-0.15in}

\caption{Retrieval recall and precision on MNIST, $\beta=1$, $\delta=10^{-6}$.}
\label{fig:MNIST_vs_eps_DP-RP}
\end{figure}

\begin{figure}[h!]
\centering

    \mbox{
    \includegraphics[width=2.2in]{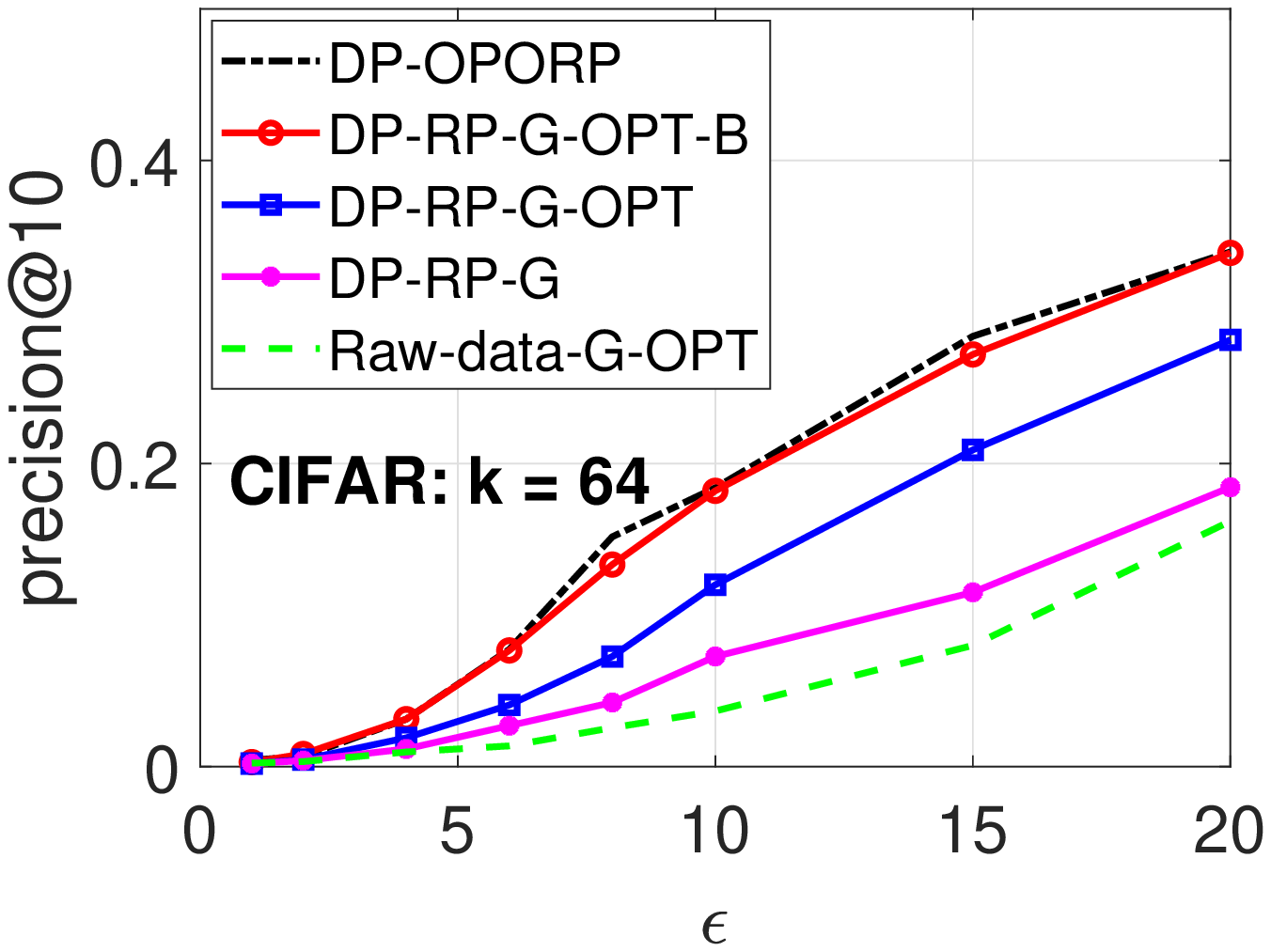} \hspace{-0.15in}
    \includegraphics[width=2.2in]{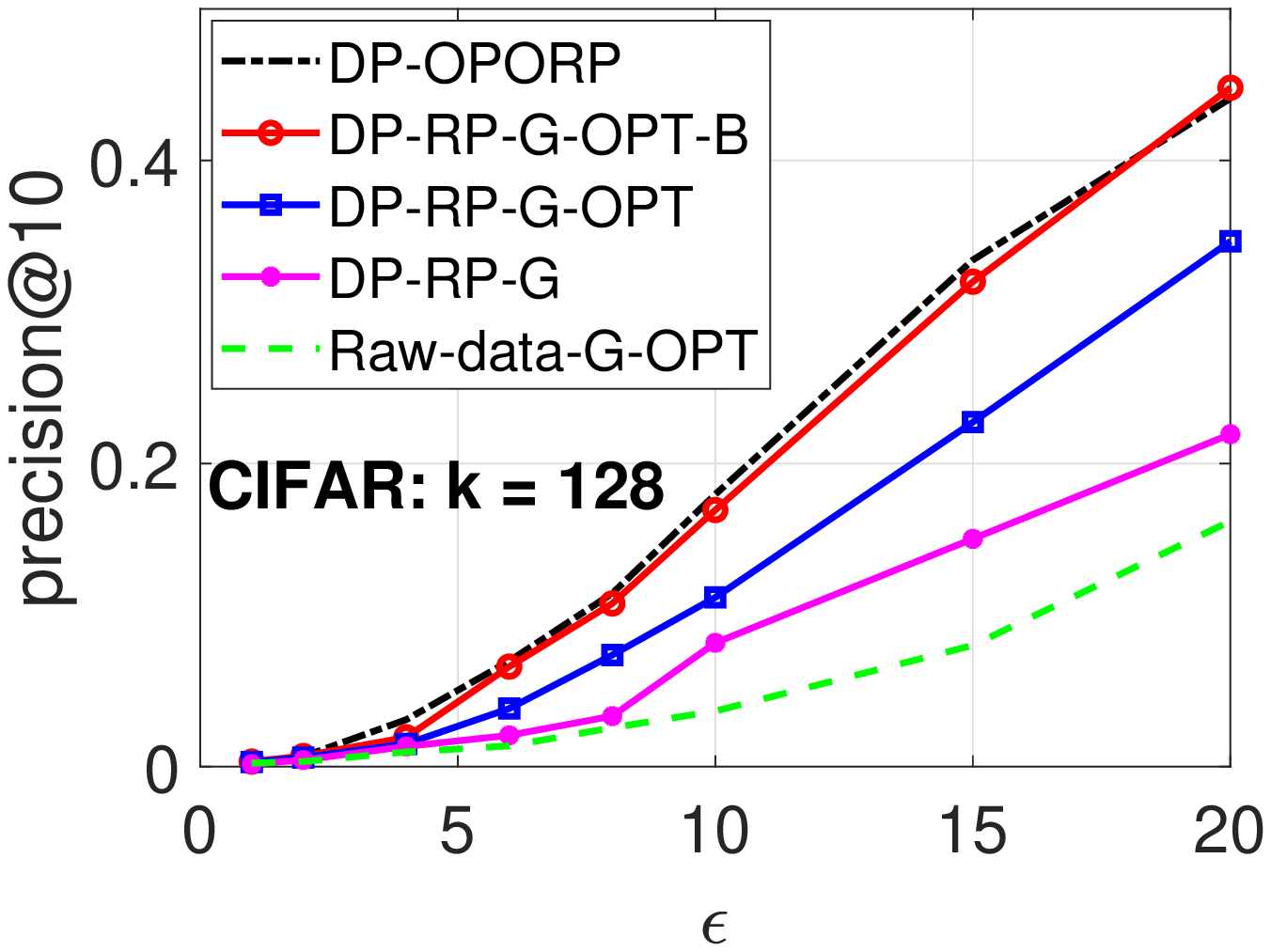}
    \includegraphics[width=2.2in]{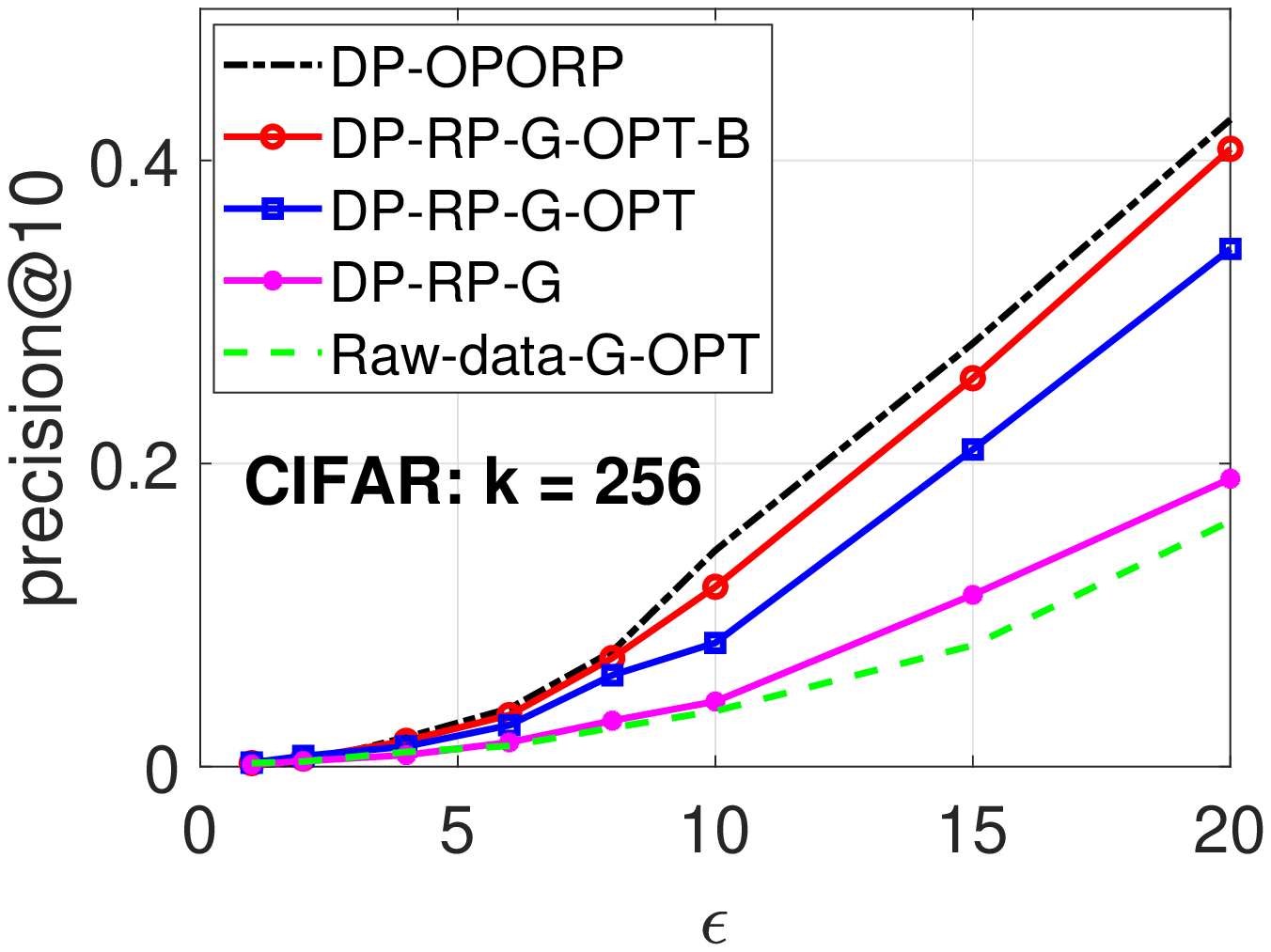}
    }

    \mbox{
    \includegraphics[width=2.2in]{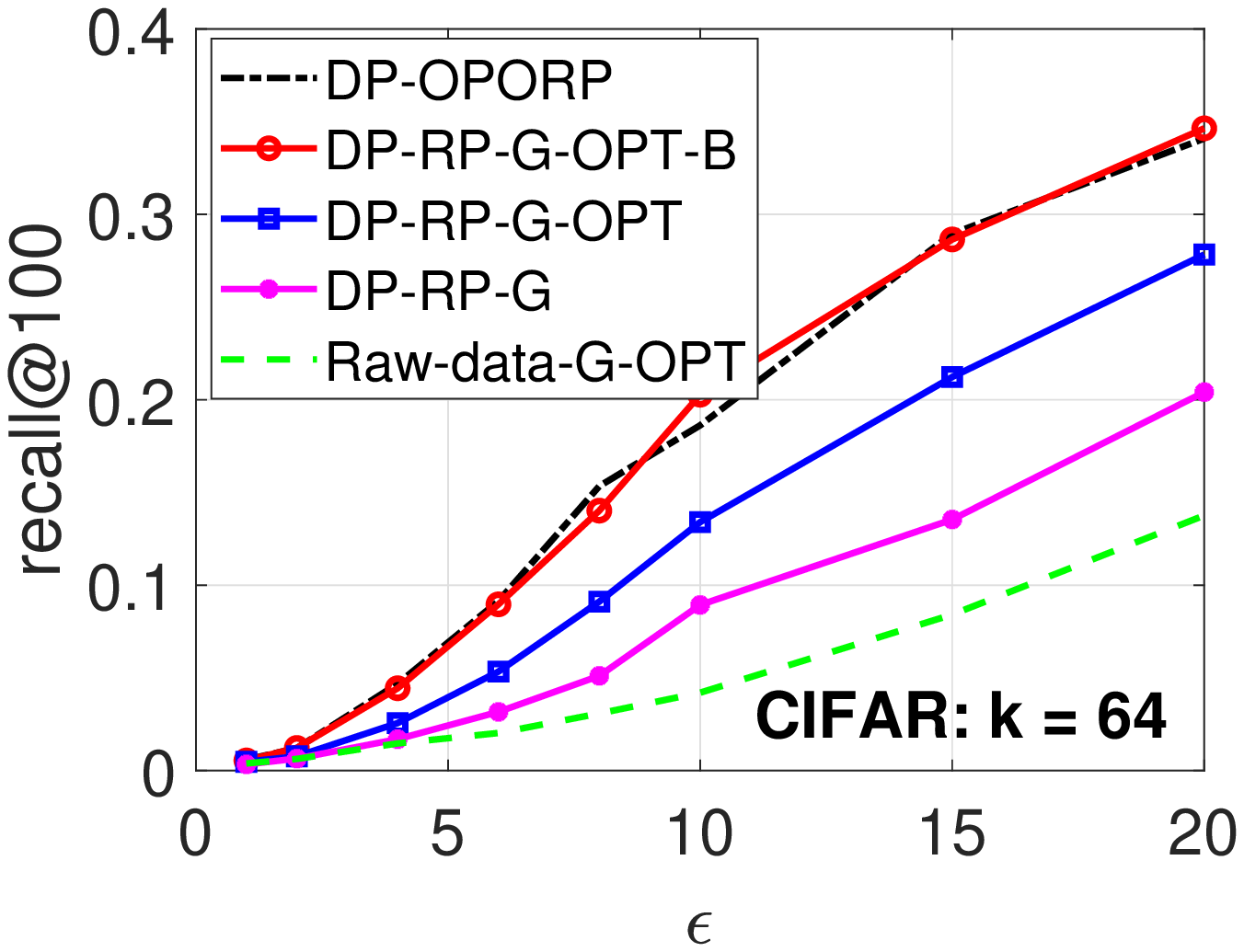} \hspace{-0.15in}
    \includegraphics[width=2.2in]{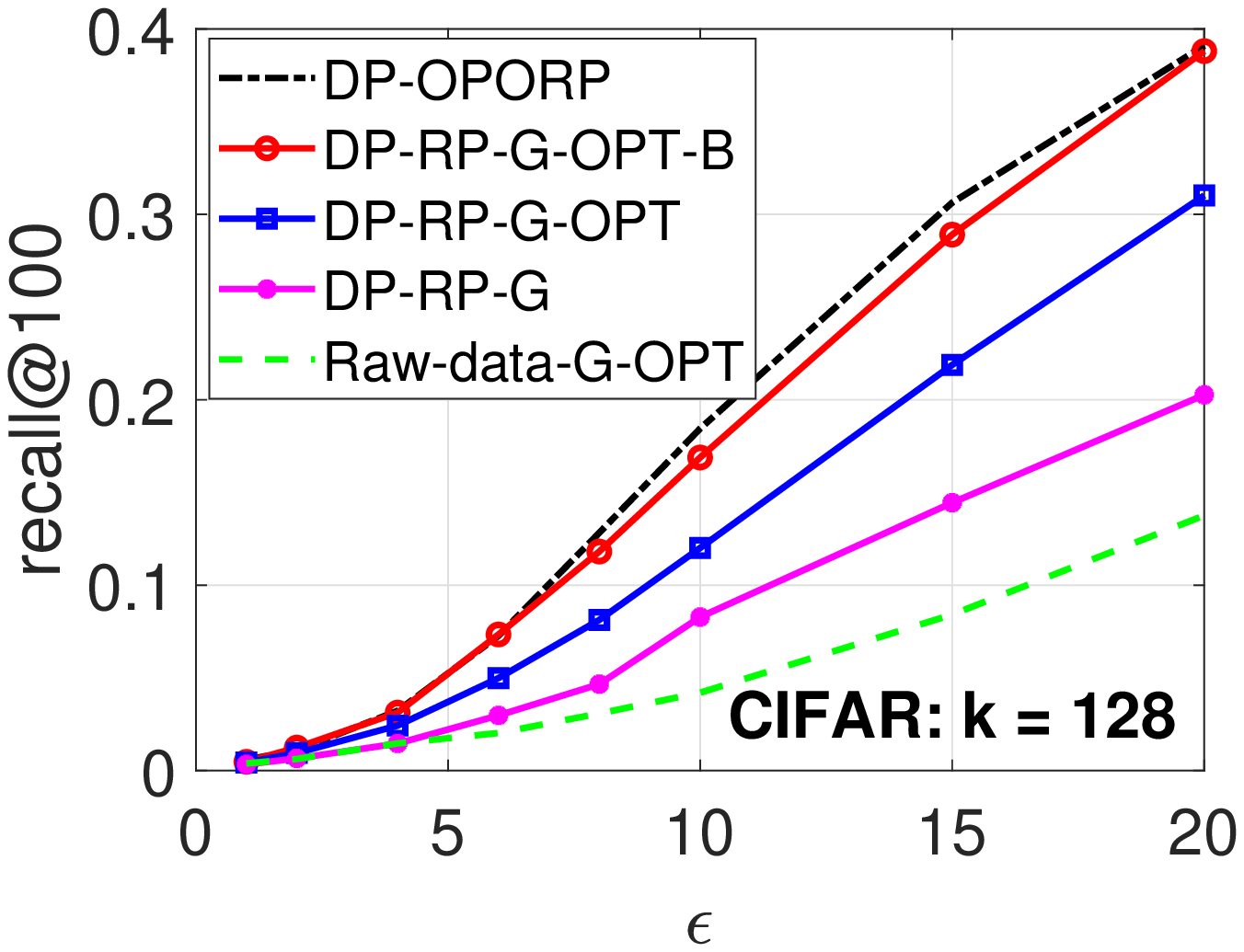}
    \includegraphics[width=2.2in]{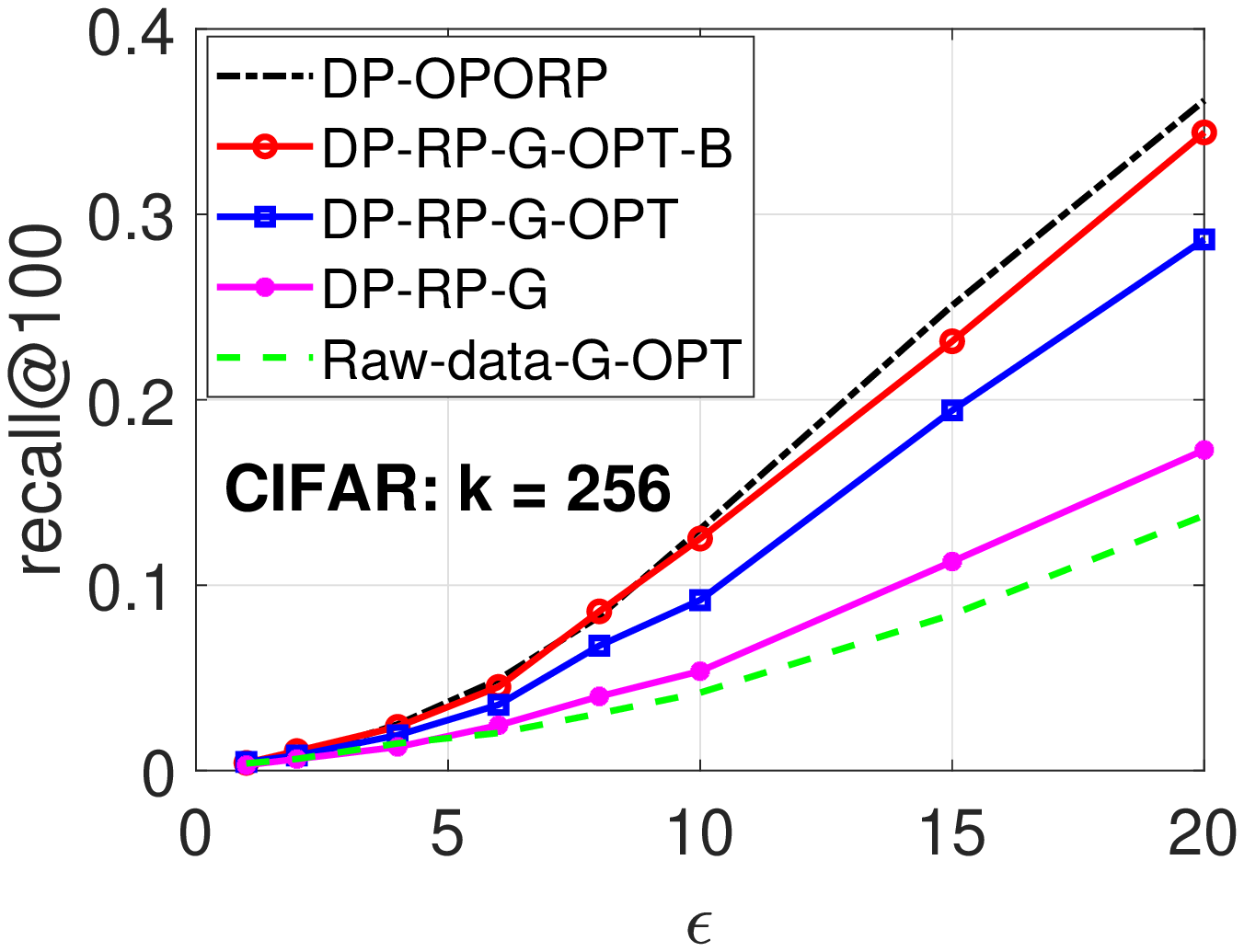}
    }

\vspace{-0.15in}

\caption{Retrieval recall and precision on CIFAR, $\beta=1$, $\delta=10^{-6}$.}
\label{fig:CIFAR_vs_eps_DP-RP}\vspace{-0.05in}
\end{figure}

\vspace{0.1in}

\noindent\textbf{Results.} In Figure~\ref{fig:MNIST_vs_eps_DP-RP} and Figure~\ref{fig:CIFAR_vs_eps_DP-RP}, we plot the precision and recall of DP-RP DP-OPORP algorithms on MNIST and CIFAR, respectively. In each plot, we fix a number of projections $k$ and plot the metrics against $\epsilon$. We observe the following:
\begin{itemize}
    \item As a result of the optimal Gaussian mechanism, DP-RP-G-OPT substantially outperforms the prior method, DP-RP-G, in terms of search precision and recall at all $\epsilon$. DP-RP-G-OPT-B further improves DP-RP-G-OPT considerably, demonstrating that the Rademacher projection matrix should be suggested for DP-RP.

    \item DP-OPORP achieves very similar utility (almost overlapping curves) with the best DP-RP variant, DP-RP-G-OPT-B, in this privacy regime. In practice, DP-OPORP might be more favorable because of its computational efficiency. These two methods provide much higher precision and recall than the Gaussian mechanism applied to the original data. On CIFAR, DP-OPORP performs slightly better than DP-RP-G-OPT-B, which may partially be explained by the extra variance reduction factor $\frac{p-k}{p-1}$ of DP-OPORP as in Theorem~\ref{theo:DP-RP-inner} and Theorem~\ref{theo:DP-OPORP-inner}.
\end{itemize}

\subsection{SVM Classification}

The RP and SignRP methods can be used to train machine learning models. We evaluate the performance of DP-SignRP-RR in classification problems trained by Support Vector Machine (SVM)~\citep{cortes1995support}. We the algorithms on the WEBSPAM dataset from the LIBSVM website~\citep{chang2011libsvm}. We apply the ``max normalization'', i.e., divide each data column by its largest magnitude such that the data entries are bounded in $[0,1]$. As the WEBSPAM dataset was generated from 3-grams (character $n$-grams), the data vectors are very high-dimensional and extremely sparse. Thus, this dataset  reveals a serious disadvantage of adding noise directly to the original data, because adding noise this way generates fully dense data vectors. In order to conduct the experiments, we only used 10K data samples for training and another 10K samples for testing and remove all features that are empty. By this preprocessing, the data dimensionality of the WEBSPAM dataset is reduced from the original 16 million to ``merely'' 250K. On the other hand, RP-type of methods (including OPORP) does not suffer from this problem because with $k$ RP samples we obtain  a dataset of only $k$ dimensions.

\begin{figure}[b!]

\vspace{-0.2in}

\centering
    \mbox{
    \includegraphics[width=2.2in]{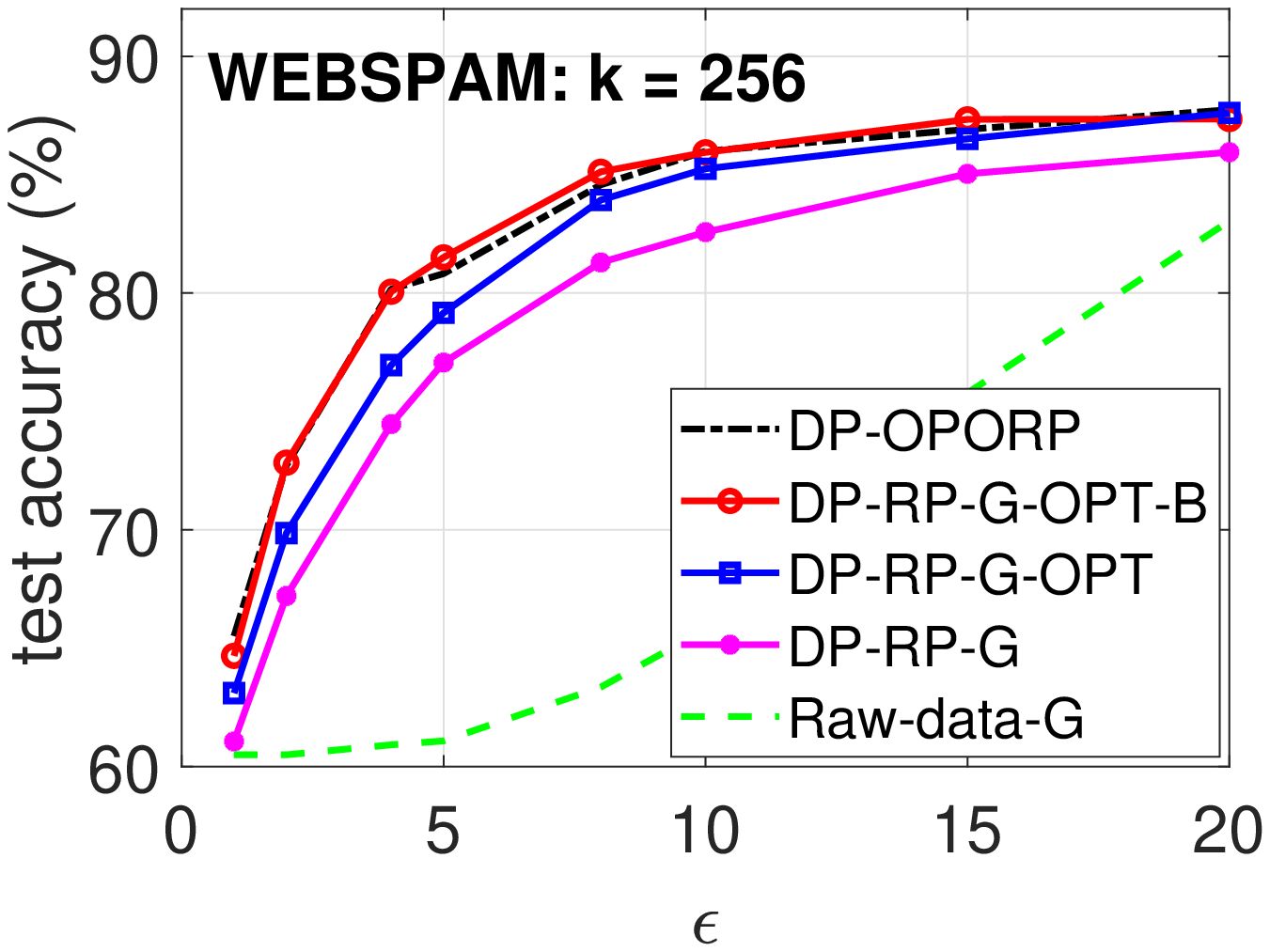} \hspace{-0.15in}
    \includegraphics[width=2.2in]{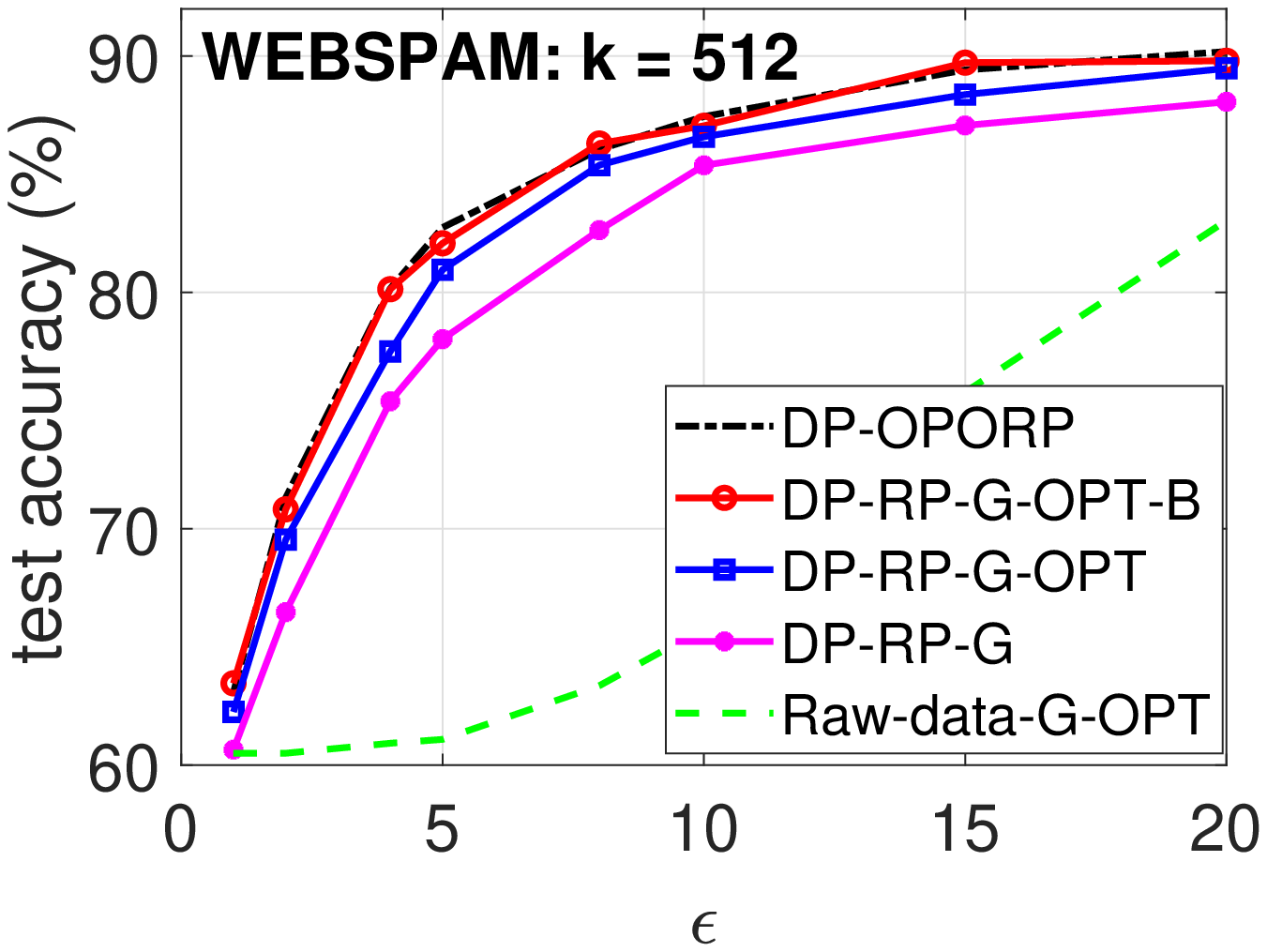} \hspace{-0.15in}
    \includegraphics[width=2.2in]{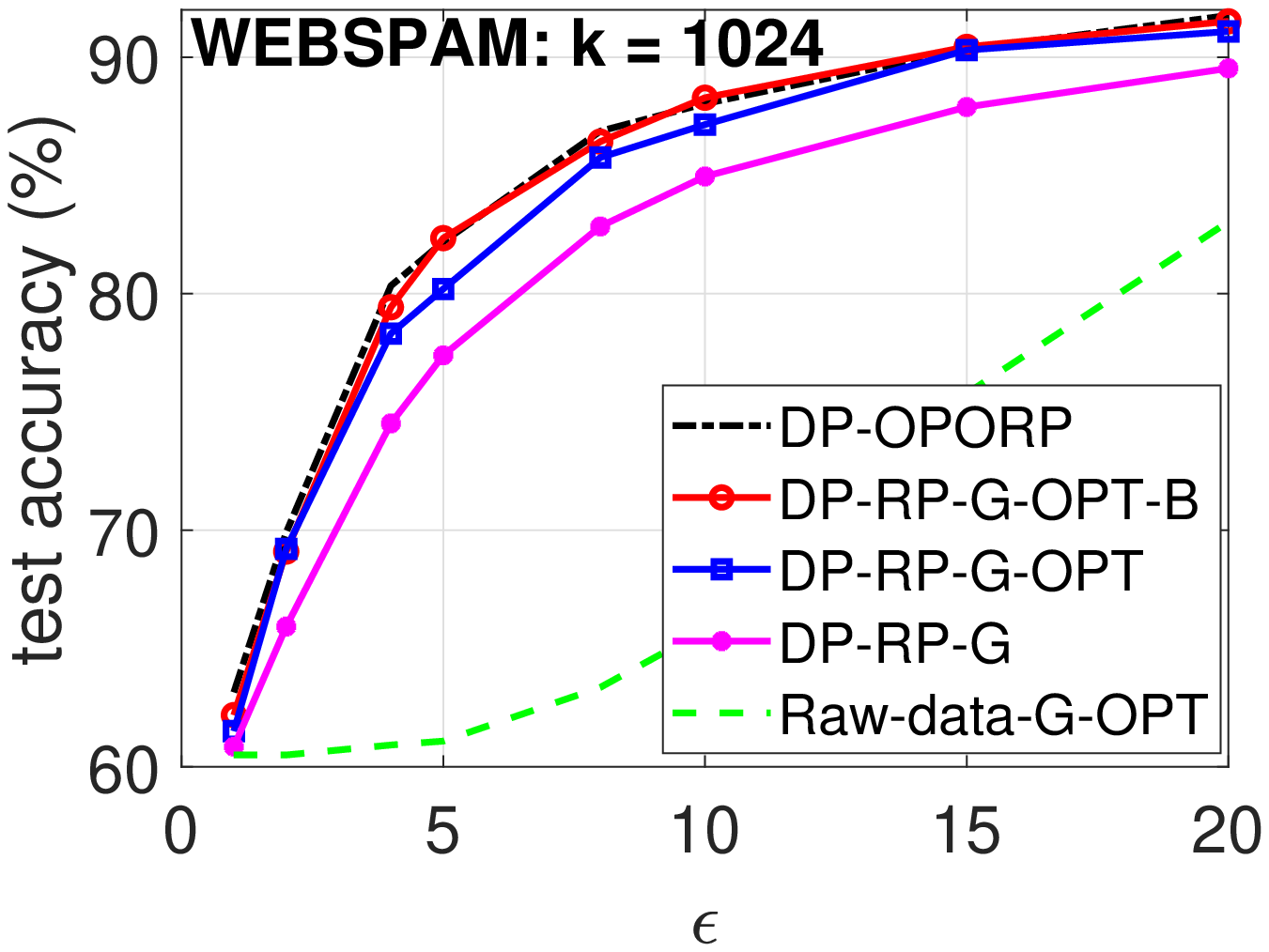}
    }

\vspace{-0.15in}

\caption{SVM classification on WEBSPAM, $\beta=1$, $\delta=10^{-6}$.}
\label{fig:SVM_vs_eps}\vspace{-0.25in}
\end{figure}

We report the test accuracy for various DP methods in Figure~\ref{fig:SVM_vs_eps}. The conclusions are similar to those in the retrieval tasks: (i) DP-RP-G-OPT-B and DP-OPORP attain the highest accuracy among all full-precision RP-type algorithms, and are significantly better than Raw-data-G-OPT. This experiment again confirms the effectiveness of DP-RP and DP-OPORP for data privatization.

\subsection{Results for DP-SignOPORP}

\begin{figure}[h]
\centering

    \mbox{
    \includegraphics[width=2.2in]{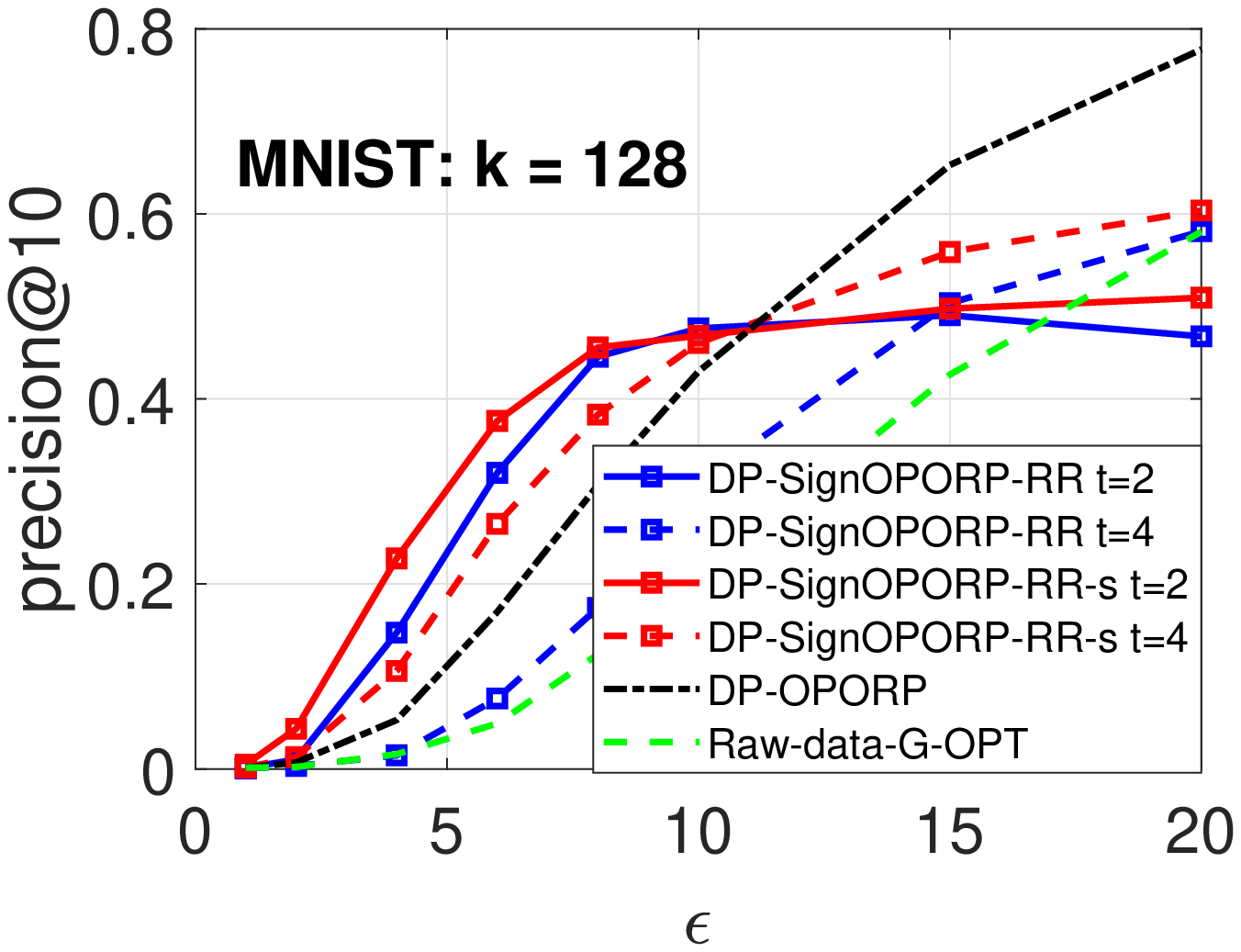} \hspace{-0.15in}
    \includegraphics[width=2.2in]{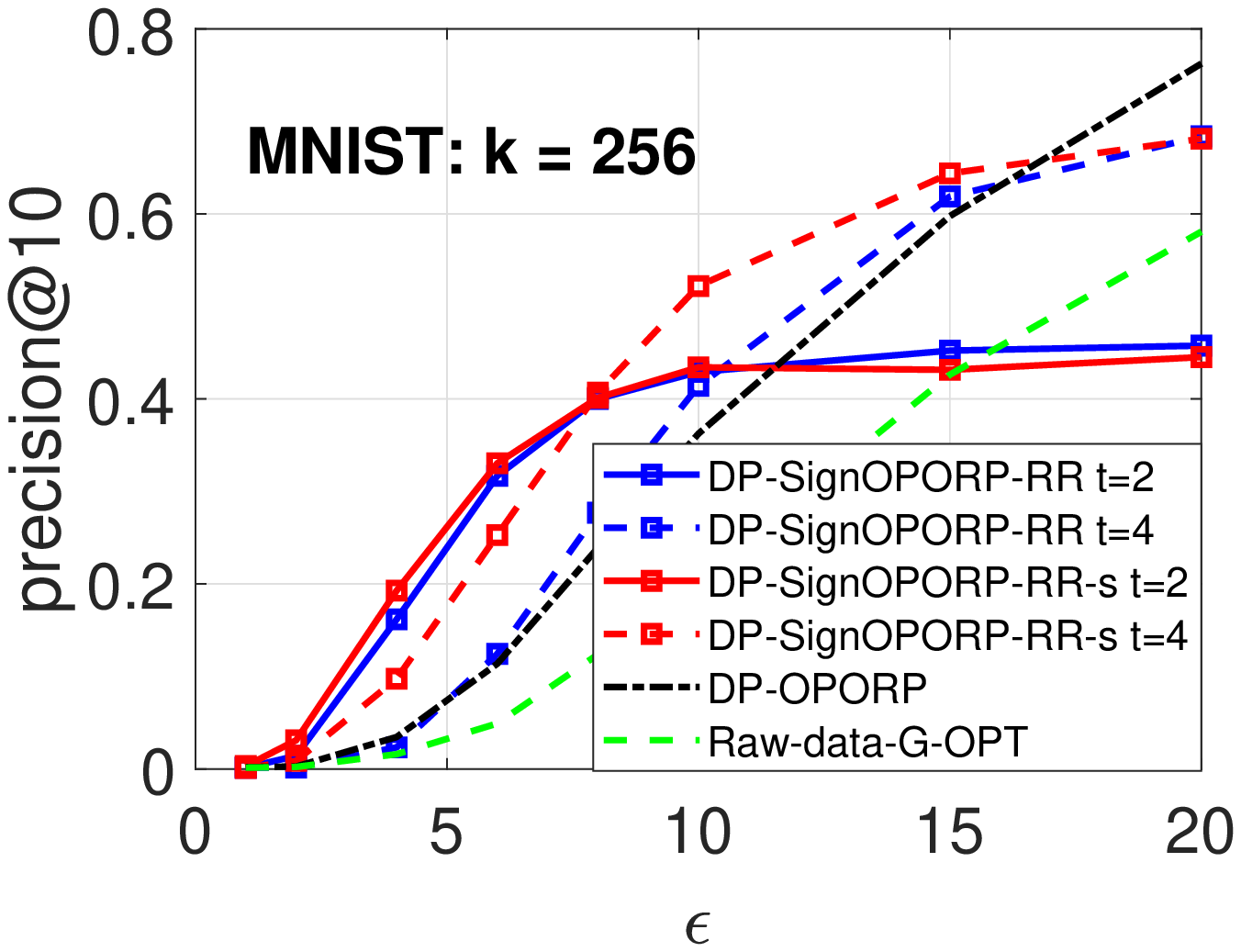}
    \hspace{-0.15in}
    \includegraphics[width=2.2in]{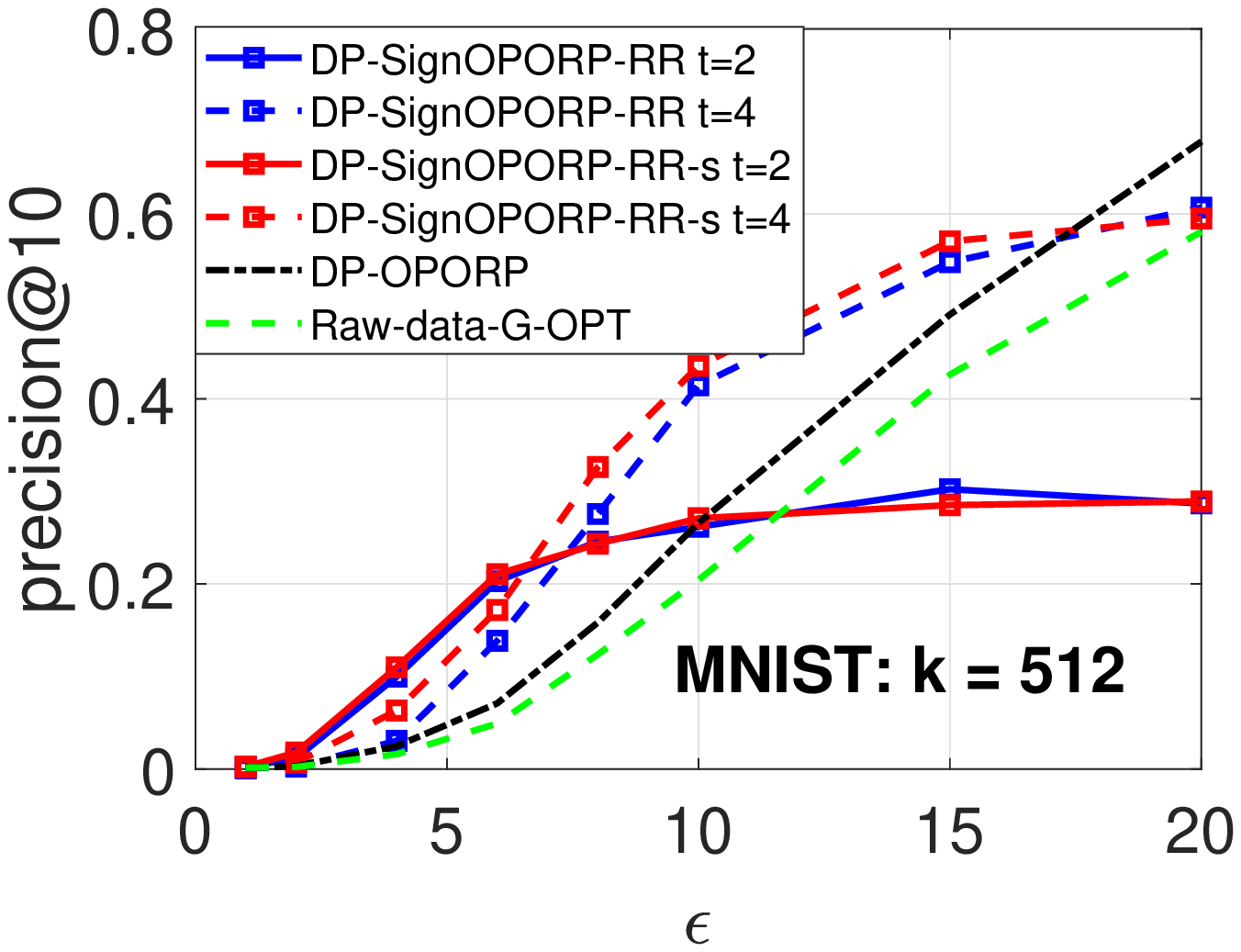}
    }

    \mbox{
    \includegraphics[width=2.2in]{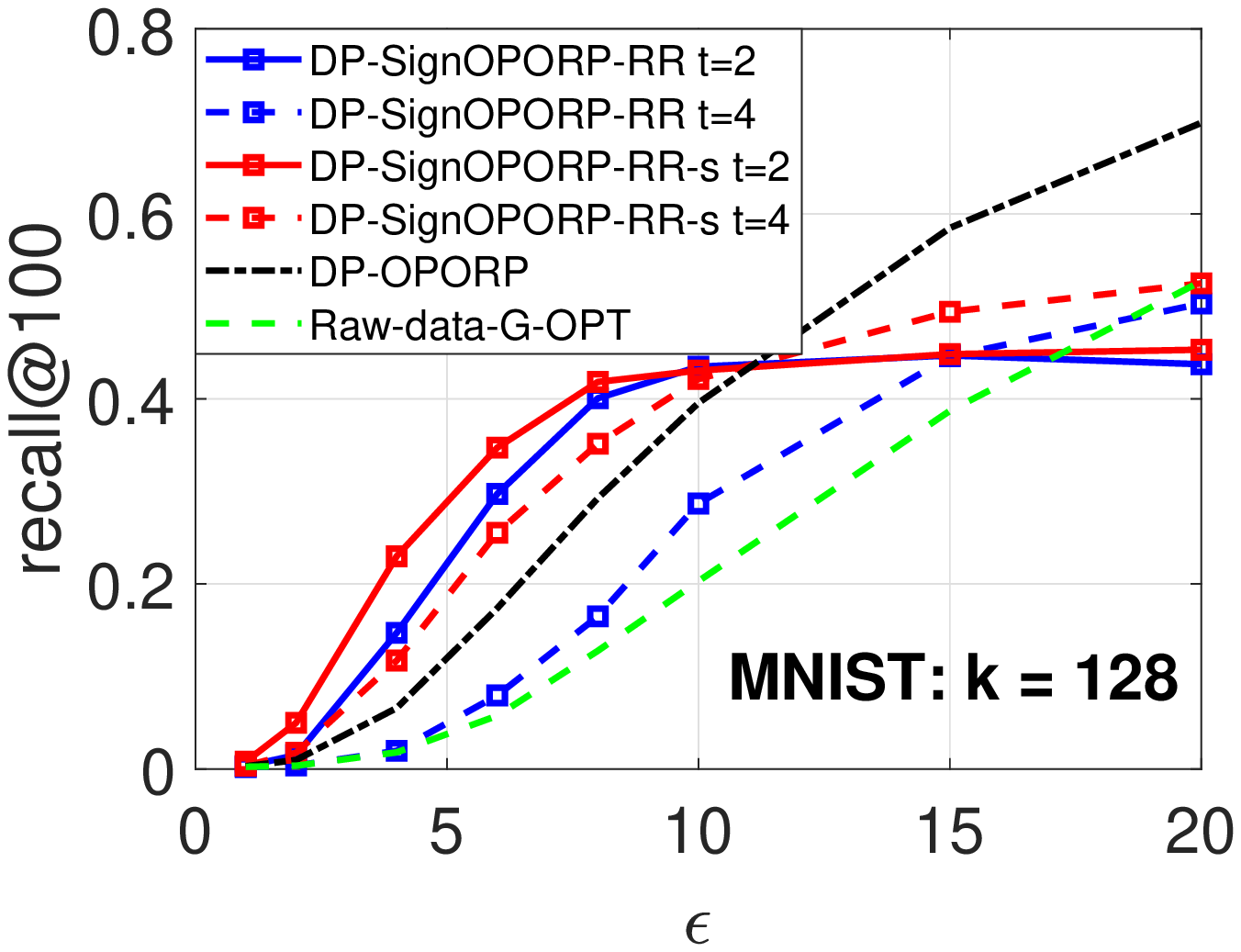} \hspace{-0.15in}
    \includegraphics[width=2.2in]{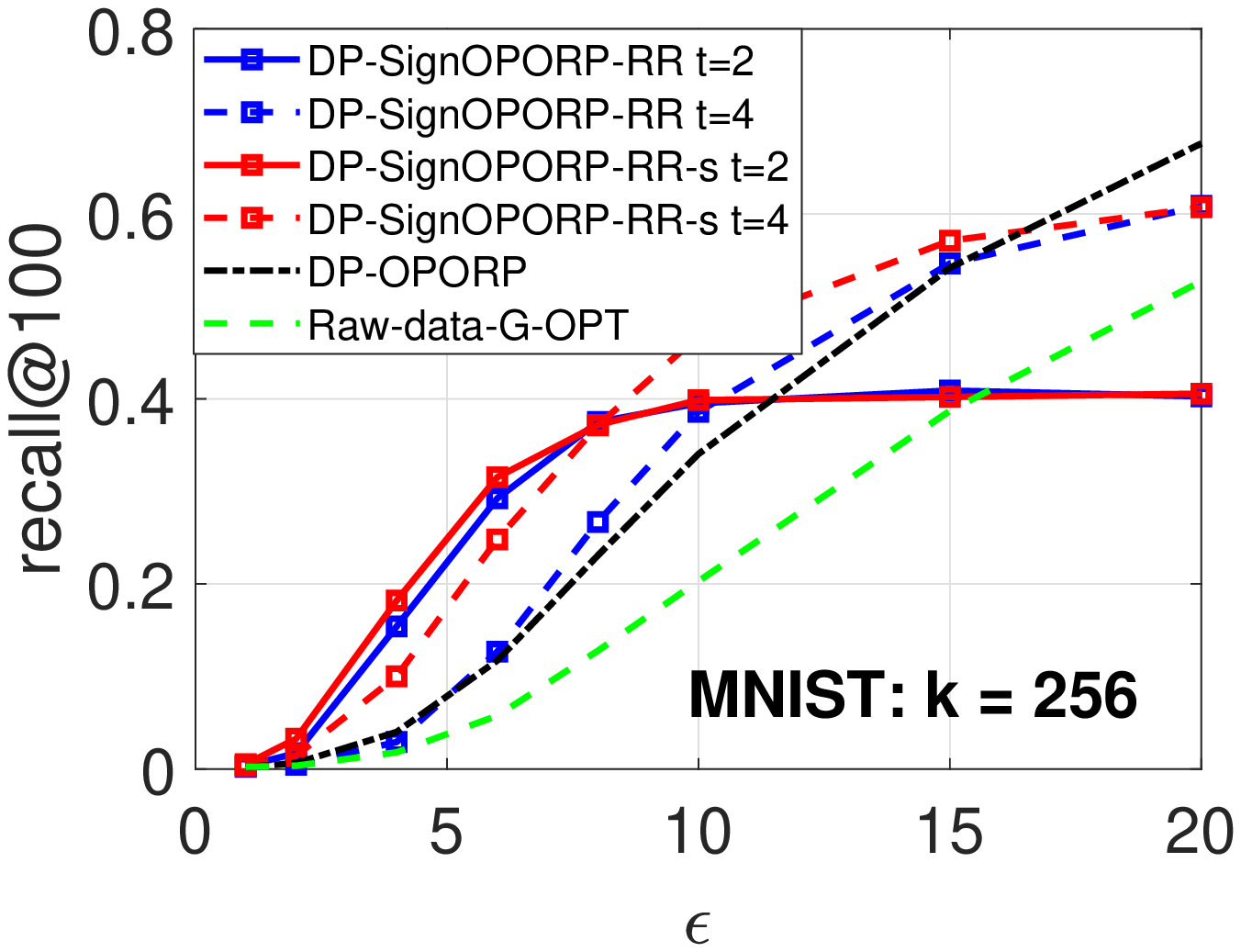}
    \hspace{-0.15in}
    \includegraphics[width=2.2in]{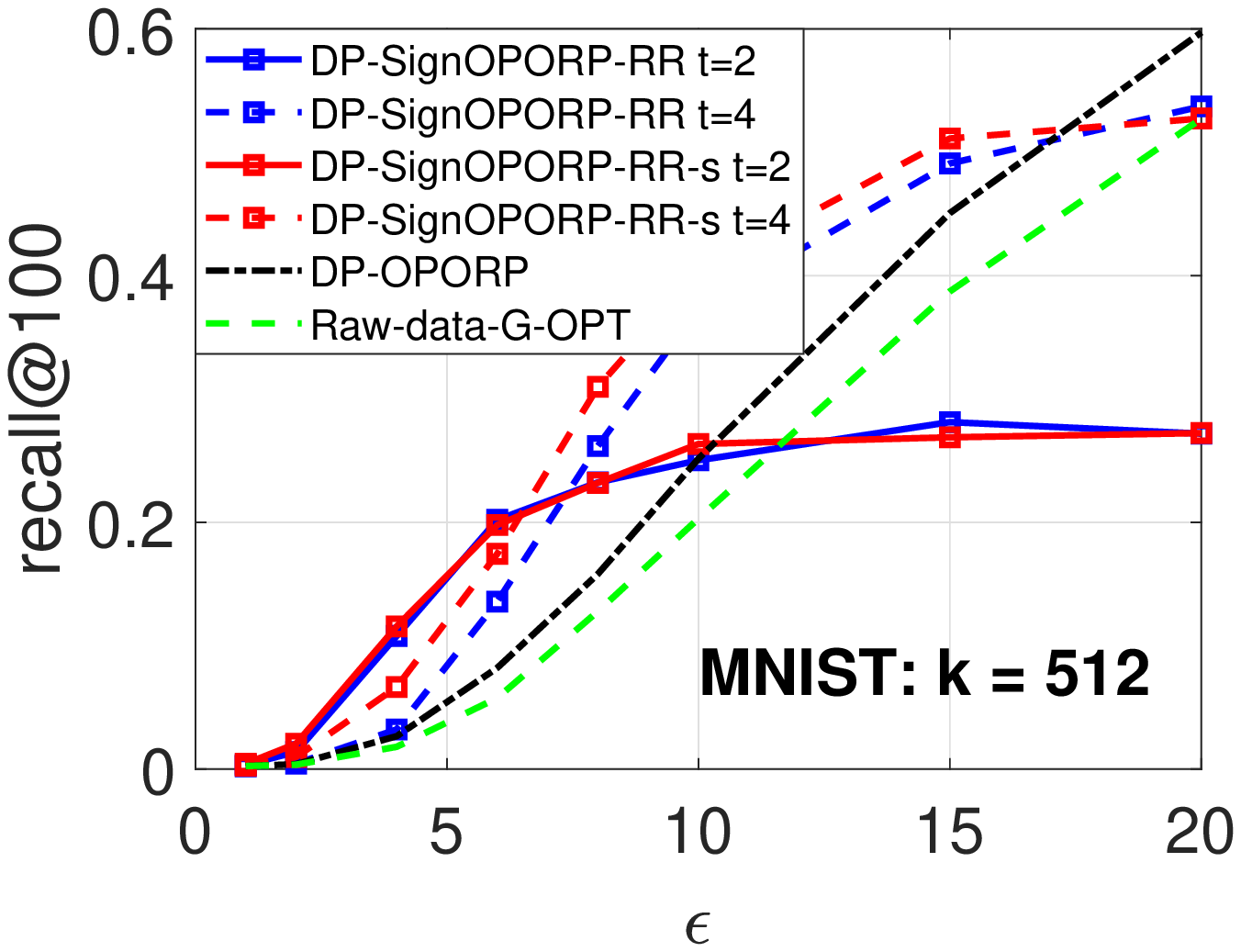}
    }

\vspace{-0.15in}

\caption{Retrieval on MNIST with DP-SignOPORP-RR and DP-SignOPORP-RR-smooth (in the caption, ``-s'' stands for ``-smooth''. For DP-OPORP and Raw-data-G-OPT, we let $\delta=10^{-6}$.}
\label{fig:MNIST_vs_eps_sign_oporp}\vspace{-0.1in}
\end{figure}

\begin{figure}[h]
\centering
    \mbox{
    \includegraphics[width=2.2in]{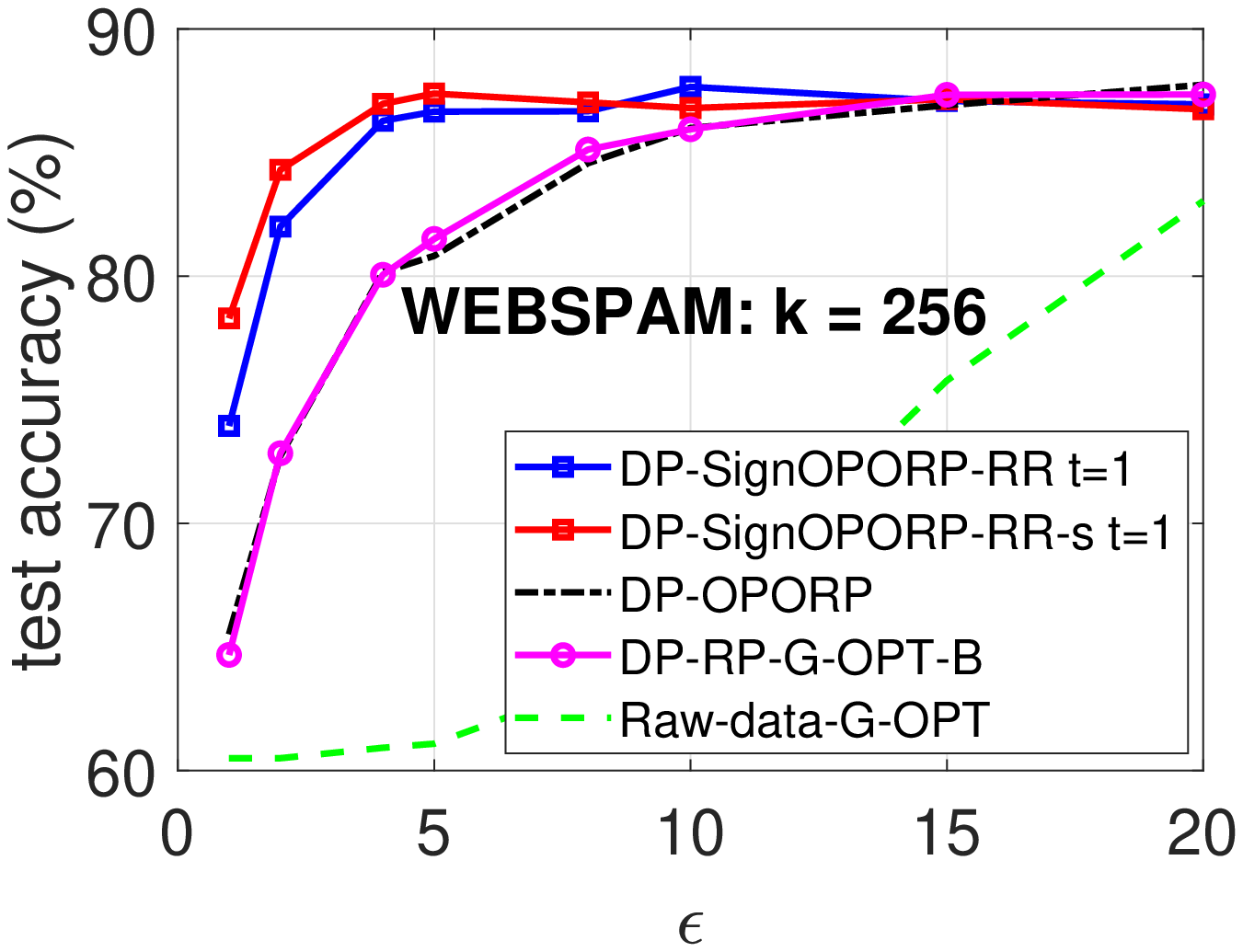} \hspace{-0.15in}
    \includegraphics[width=2.2in]{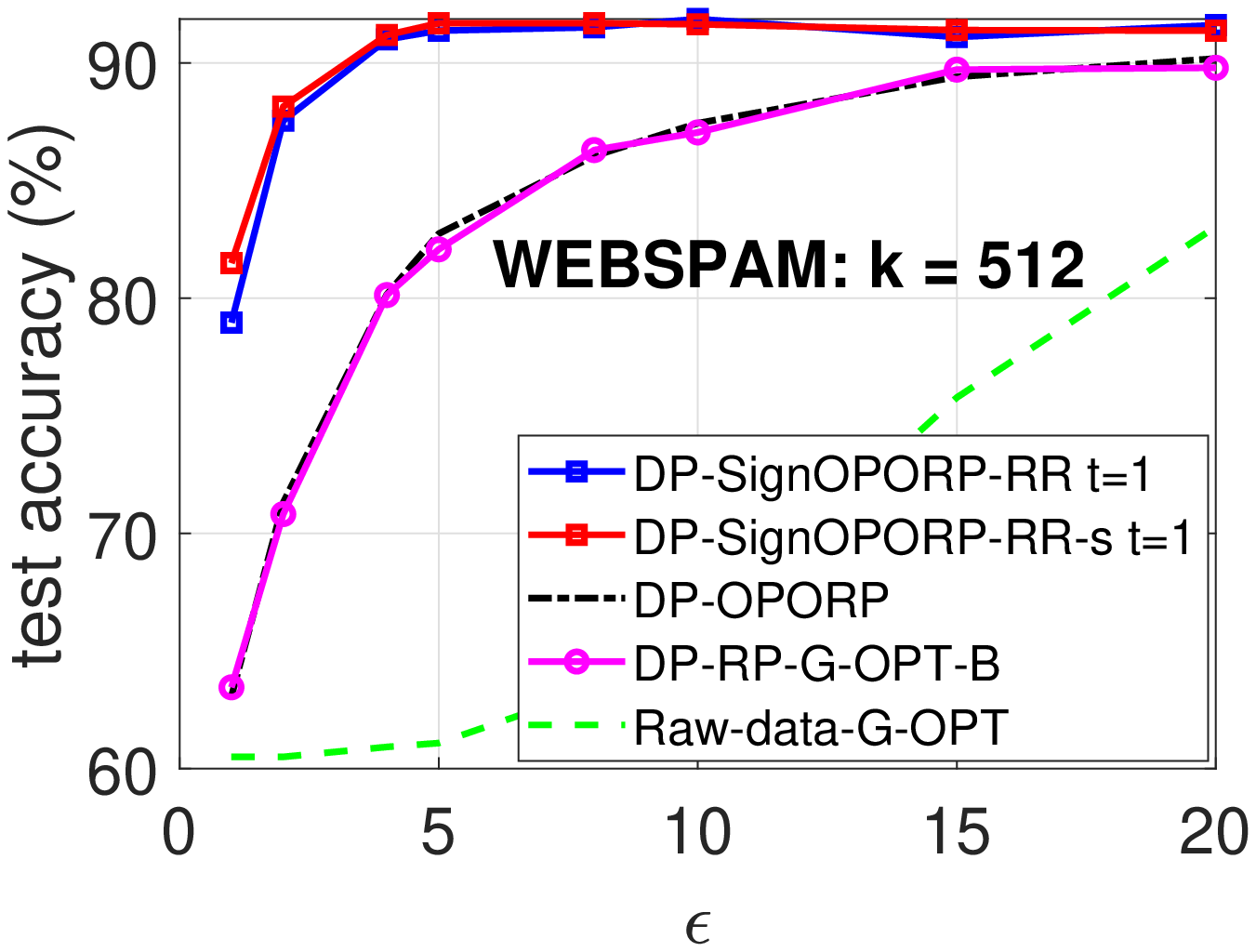} \hspace{-0.15in}
    \includegraphics[width=2.2in]{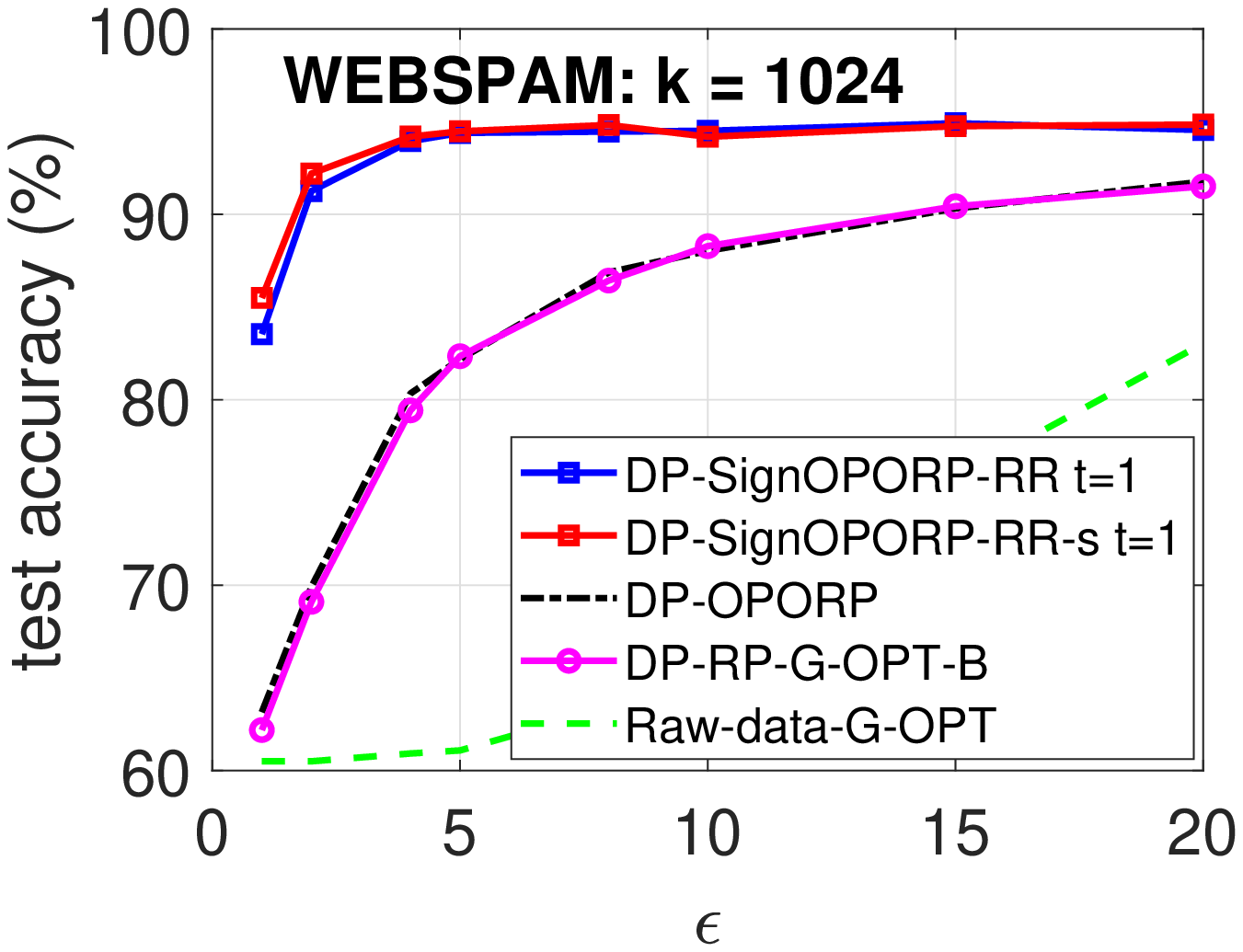}
    }

\vspace{-0.15in}

\caption{Classification accuracy on Webspam with DP-SignOPORP-RR and DP-SignOPORP-RR-s. For DP-OPORP and Raw-data-G-OPT,  we let $\delta=10^{-6}$.}
\label{fig:SVM_vs_eps_sign_oporp}
\end{figure}

In Figure~\ref{fig:MNIST_vs_eps_sign_oporp}, we plot the precision and recall of DP-SignOPORP-RR (with standard randomized response), and DP-SignOPORP-RR-smooth (with smooth flipping probability). In this figure, we present the results with repetition $t=2,4$ which yield good overall performance. As we can see, (i) the proposed smooth flipping probability considerably improves the standard randomized response technique, and (ii) DP-SignOPORP, in general, provides better search accuracy than DP-OPORP when $\epsilon< 10\sim 15$. This range of $\epsilon$ is common in the use cases of DP for providing reasonable privacy protection. Also, we note that using a smaller $t$ (repetitions) typically performs better when $\epsilon$  is relatively small, but worse when $\epsilon$ is large.

In Figure~\ref{fig:SVM_vs_eps_sign_oporp}, we report the results on SVM classification. For this task, we plot DP-SignOPORP with $t=1$. We again observe the advantage of DP-SignOPORP over DP-OPORP, and the advantage of the smooth flipping probability over the standard randomized response. Specifically, when $\epsilon\approx 5$, the test accuracy of Raw-data-G-OPT is around 60\% (for WEBSPAM this is almost the same as random guessing), but DP-SignOPORP can achieve 95\% accuracy with $k=1024$.

\section{iDP-SignRP Under Individual Differential Privacy (iDP)}   \label{sec:individual DP}

We provide additional algorithms and evaluation on the so-called ``individual differential privacy'' (iDP, Definition~\ref{def:idp})~\citep{comas2017individual}, a relaxation of standard DP which aims at improving the utility of private algorithms. iDP treats the dataset $u$ as the ``ground truth'' to be protected. The ``indistinguishability'' requirement is only cast on $u$ and its neighbors specifically, instead of on any possible dataset. Therefore, operationally, in the noise addition mechanisms for example, iDP essentially follows the local sensitivity (Definition~\ref{def:local-sensitivity}) when computing the noise level, which can be much smaller than that required by the standard DP. iDP has achieved excellent utility for computing robust statistics at small $\epsilon$~\citep{comas2017individual}. While iDP does not provide the same level of privacy protection as the ``worst-case'' standard DP, it might be sufficient in certain application scenarios, e.g., data publishing/release, when the procedure is non-interactive and the released dataset is indeed the target that one is interested in privatizing.

\vspace{0.1in}

For the DP algorithms that have been discussed previously in this paper, we first note that, for DP-RP and DP-OPORP, the local sensitivity at any $u\in\mathcal U$ equals the global sensitivity. In other words, iDP does not help improve DP-RP and DP-OPORP. Also, we will soon discuss  the reason why SignRP can be much better than SignOPORP under iDP. Therefore, we will mainly investigate the SignRP algorithms under iDP.

\vspace{0.1in}

We propose two iDP-SignRP methods, based on noise addition and sign flipping, respectively. Both approaches share the same key idea of iDP, that is, many signs of the projected values do not need perturbations. This can be seen from Figure~\ref{fig:smooth_flip_prob}, where the ``local flipping probability'' is non-zero only in the regime when the projected data is near $0$ (i.e., $L=1$ in Algorithm~\ref{alg:DP-signRP-RR-smooth}). Since in other cases the local flip probability is zero, perturbation is not needed. As a result, out of $k$ projections, only a fraction of the projected values needs to be perturbed. This significantly reduces the noise injected to SignRP and boosts the utility by a very large margin.

\subsection{iDP-SignRP-G by Gaussian Noise Addition}

In Algorithm~\ref{alg:DP-SignRP-gaussian-noise-individual}, we present the iDP-SignRP-G method for one data vector $u$. We use the ``local flipping probability'' (e.g., in Figure~\ref{fig:smooth_flip_prob}) to choose which projections are perturbed before taking signs. After applying random projection to get $k$ projected values, we do the following steps:

\begin{enumerate}
    \item We compute noise-indicators $(I_1,...,I_k)$ for each projected value in $x=\frac{1}{\sqrt k}W^Tu$ using Algorithm~\ref{alg:noise-indicator}. Denote $\mathcal A=\{I_j:I_j=1, j=1,...,k\}$ and $N_+=|\mathcal A|$. This is the maximal number of different signs of $x$ and $x'=W^Tu'$, $\forall u'\in Nb(u)$.

    \item We compute the sensitivity $\Delta_2=\beta\max_{i=1,...,p} \| W_{[i,\mathcal A]}\|$, where $W_{[i,\mathcal A]}$ denotes the $i$-th row of $W$ indexed at $\mathcal A$, which is an $N_+$-dimensional vector.

    \item We use the optimal Gaussian mechanism (Theorem~\ref{theo:gauss-optimal}) to compute $\sigma$, with $\Delta_2$ computed above and privacy parameters $(\epsilon,\delta)$.

    \item For $j=1,..,k$, if $j\notin\mathcal A$, we take $\tilde s_j=sign(x_j)$; if $j\in\mathcal A$, we take $\tilde s_j=sign(x_j+G)$ where $G\sim N(0,\sigma^2)$ is a Gaussian noise. Finally we output $\tilde s=[\tilde s_1,...,\tilde s_k]$.
\end{enumerate}

\begin{algorithm}[t]
	{
		\vspace{0.05in}
		\textbf{Input:} Data $u\in[-1,1]^p$; Privacy parameters $\epsilon>0$, $\delta\in (0,1)$; Number of projections~$k$
		
		\vspace{0.05in}
		
		\textbf{Output:}   Differentially private sign random projections
		\vspace{0.05in}		
		
		Apply RP by $x=\frac{1}{\sqrt k}W^Tu$, where $W\in\mathbb R^{p\times k}$ is a random Rademacher matrix
		
		For every projected value in $x$, compute $(I_1,...,I_k)$ by Algorithm~\ref{alg:noise-indicator}
		
		Let $\mathcal A=\{I_j:I_j=1, j=1,...,k\}$ and $\tilde N_+=|\mathcal A|$
		
		Compute sensitivity $\Delta_2=\beta\sqrt{\frac{\tilde N_+}{k} }$

		 Compute $\sigma$ by Theorem~\ref{theo:gauss-optimal} with $\Delta_2$ and privacy budget $\epsilon$ and $\delta$
		
		 Compute $\tilde s_j=\begin{cases}
		 sign(x_j), & j\notin \mathcal A\\
		 sign(x_j+G), & j\in\mathcal A
		 \end{cases}$, where $G\sim N(0,\sigma^2)$ is iid Gaussian noise

		Return $\tilde s = [\tilde s_1,...,\tilde s_k]$
	}
	\caption{iDP-SignRP-G (DP-SignRP with Gaussian noise)}
	\label{alg:DP-SignRP-gaussian-noise-individual}
\end{algorithm}

\begin{algorithm}[t]
	{
		\vspace{0.05in}
		\textbf{Input:} Data $u\in[-1,1]^p$; one projected value $z$; adjacency parameter $\beta$
		
		
		\textbf{Output:}  Indicator $I$ w.r.t. projection $w$ for data vector $u$
		
		\vspace{0.05in}		

		$I=0$
		
		\textbf{If} $\beta/\sqrt{k}\geq |z|$
		
		\hspace{0.2in}$I=1$

		\textbf{End If}

	}
	\caption{Compute noise-indicator of iDP-SignRP-G for one projection}
	\label{alg:noise-indicator}
\end{algorithm}

Let's explain the intuition behind DP-SignRP-G. Since a neighboring data vector $u'$ only differs from $u$ in one dimension by at most $\beta$, for each single projection $w$, when $\beta\max_{i=1,...,p}|w_i|\leq |w^Tu|$, there is no neighbor $u'$ of $u$ that may change the sign of the projected value of $u$, i.e., $sign(w^Tu')\neq sign(w^Tu)$. In other words, when $\beta\max_{i=1,...,p}|w_i|\leq |w^Tu|$,  no noise is needed for this projected value to attain iDP. This is the reason why we call the output of Algorithm~\ref{alg:noise-indicator} a ``noise-indicator''. Consequently, in step 4 of iDP-SignRP-G it suffices to add Gaussian noise only to those projected values $x_j$ with $j\in\mathcal A$, instead of to all $k$ projections as in DP-RP-G-OPT.

\vspace{0.1in}

\begin{theorem}[iDP-SignRP-G] \label{theo:privacy-DP-SignRP-Gaussian}
Algorithm~\ref{alg:DP-SignRP-gaussian-noise-individual} is $(\epsilon,\delta)$-iDP for data $u$.
\end{theorem}

\begin{proof}
For a data vector $u$, let $Nb(u)$ be its neighbor set with vectors that differ from $u$ by at most $\beta$ in one dimension. Denote $x=\frac{1}{\sqrt k}W^Tu$ and $x'=\frac{1}{\sqrt k}W^Tu'$. Let $(I_1,...,I_k)$ be the noise-indicators from Algorithm~\ref{alg:noise-indicator} and $\mathcal A=\{i: I_j=1\}$, $\tilde N_+=|\mathcal A|$. Consider the two sets separately:

\begin{itemize}
    \item For $j\in [k]\setminus \mathcal A$, by the condition $\beta/\sqrt{k}\leq |z|$, we know that $\forall u'\in Nb(u)$, it holds that $sign(x_i)=sign(x_i')$.

    \item For $j\in\mathcal A$, consider the sub-vector $x_{\mathcal A}$. Adding iid Gaussian noise to $x_{\mathcal A}$ according to Theorem~\ref{theo:gauss-optimal} with $\Delta_2=\beta\sqrt{\frac{\tilde N_+}{k} }$ ensures the $(\epsilon,\delta)$-DP of $x_{\mathcal A}$. By the post-processing property of DP, we know that $sign(x_{\mathcal A})$ is also $(\epsilon,\delta)$-DP. Thus, for any $Q\in \{-1,1\}^{N_+}$, we have $Pr(sign(x_{\mathcal A})=Q)-e^\epsilon Pr(sign(x'_{\mathcal A})=Q)\leq \delta$, $\forall u'\in Nb(u)$.
\end{itemize}
Combining two parts, we have for any $Q\in \{-1,1\}^k$,
\begin{align*}
    Pr(sign(x)=Q)-e^\epsilon Pr(sign(x')=Q)=Pr(sign(x_{\mathcal A})=Q)-e^\epsilon Pr(sign(x'_{\mathcal A})=Q)\leq \delta,
\end{align*}
for all $u'\in Nb(u)$. By the symmetry of DP (on the sub-vector $x_{\mathcal A}$), we also know that $Pr(sign(x')=Q)-e^\epsilon Pr(sign(x)=Q)\leq \delta$. This proves the $(\epsilon,\delta)$-iDP by Definition~\ref{def:idp}.
\end{proof}

\subsection{iDP-SignRP-RR by Randomized Response}

\begin{algorithm}[h]
{
    \vspace{0.05in}
    \textbf{Input:} Data $u\in[-1,1]^p$, privacy parameters $\epsilon>0$, $0<\delta<1$, number of projections~$k$

    \vspace{0.05in}

    \textbf{Output:}   Differentially private sign random projections

    \vspace{0.05in}

    Apply RP by $x=\frac{1}{\sqrt k}W^Tu$, where $W\in\mathbb R^{p\times k}$ is a random Rademacher matrix

	For every column in $W$, compute $(I_1,...,I_k)$ by Algorithm~\ref{alg:noise-indicator}
	
	Let $\mathcal A=\{I_j:I_j=1, j=1,...,k\}$ and $\tilde N_+=|\mathcal A|$

    Compute $\tilde s_j=\begin{cases}
    sign(x_j), & j\notin\mathcal A\\
    sign(x_j), &j\in\mathcal A\ \text{with prob.}\ \frac{e^{\epsilon'}}{e^{\epsilon'}+1}\\
    -sign(x_j), &j\in\mathcal A\ \text{with prob.}\ \frac{1}{e^{\epsilon'}+1}
    \end{cases}$ for $j=1,...,k$, with $\epsilon'=\epsilon/\tilde N_+$

    Return $\tilde s$ as the DP-SignRP of $u$
    }
    \caption{iDP-SignRP-RR}
    \label{alg:DP-signRP-RR-individual}
\end{algorithm}

Similar to Section~\ref{sec:DP-signRP}, we also have an iDP-SignRP-RR method with pure $\epsilon$-DP by random sign flipping, as summarized in Algorithm~\ref{alg:DP-signRP-RR-individual}. After we apply random projection $x=\frac{1}{\sqrt k}W^Tu$, we call the same procedure as in iDP-SignRP-G to determine set $\mathcal A$ representing the projected values that need perturbation for iDP. For $j\notin\mathcal A$, we use the original $\tilde s_j=sign(x_j)$. For $j\in\mathcal A$, we keep $sign(x_j)$ with probability $\frac{e^{\epsilon'}}{e^{\epsilon'}+1}$ and flip the sign otherwise, where $e^{\epsilon'}=\epsilon/\tilde N_+$ with $\tilde N_+=|\mathcal A|$.

\begin{theorem} \label{theo:DP-SignRP}
Algorithm~\ref{alg:DP-signRP-RR-individual} achieves $\epsilon$-iDP for data $u$.
\end{theorem}
\begin{proof}
The high-level proof idea is similar to that of Theorem~\ref{theo:privacy-DP-SignRP-Gaussian}. For $u\in [-1,1]^p$ let $u'$ be an $\beta$-neighboring data. Let $s=sign(W^Tu)\in \{-1,+1\}^k$, $s'=sign(W^Tu')\in \{-1,+1\}^k$, and denote $\tilde s$ and $\tilde s'$ as the randomized output of $s$ and $s'$ by Algorithm~\ref{alg:DP-signRP-RR-individual}, respectively. Consider $\mathcal A$ in Algorithm~\ref{alg:DP-signRP-RR-individual}. By Algorithm~\ref{alg:noise-indicator}, we know that for $j\notin\mathcal A$, $Pr(\tilde s_j=\tilde s_j')=Pr(s_j=s_j')=1$, $\forall u'\in Nb(u)$. For projections in $\mathcal A$, denote $S=\{j\in\mathcal A:s_j\neq s_j'\}$ and $S^c=\mathcal A\setminus S$. For any vector $y\in \{-1,+1\}^k$, we further define $S_0=\{j\in S:s_j=y_j\}$, $S_1=\{j\in S:s_j\neq y_j\}$, $S^c_0=\{j\in S^c:s_j=y_j\}$ and $S^c_1=\{j\in S^c:s_j\neq y_j\}$. Since the $k$ projections are independent, by composition we have
\begin{align*}
    \log\frac{Pr(\tilde s=y)}{Pr(\tilde s'=y)}&=\log\frac{\prod_{j\notin\mathcal A}Pr(\tilde s_j=y_j)\prod_{j\in S^c_0}\frac{e^{\epsilon'}}{e^{\epsilon'}+1}\prod_{j\in S^c_1}\frac{1}{e^{\epsilon'}+1}\prod_{j\in S_0}\frac{e^{\epsilon'}}{e^{\epsilon'}+1}\prod_{j\in S_1}\frac{1}{e^{\epsilon'}+1}}{\prod_{j\notin\mathcal A}Pr(\tilde s_j'=y_j)\prod_{j\in S^c_0}\frac{e^{\epsilon'}}{e^{\epsilon'}+1}\prod_{j\in S^c_1}\frac{1}{e^{\epsilon'}+1}\prod_{j\in S_0}\frac{1}{e^{\epsilon'}+1}\prod_{j\in S_1}\frac{e^{\epsilon'}}{e^{\epsilon'}+1}} \\
    &\leq \log\frac{\prod_{j\in S}\frac{e^{\epsilon'}}{e^{\epsilon'}+1}}{\prod_{j\in S}\frac{1}{e^{\epsilon'}+1}}=|S|\epsilon'\leq \tilde N_+ \epsilon'=\epsilon,
\end{align*}
which proves the $\epsilon$-iDP according to Definition~\ref{def:idp}.
\end{proof}

The number of projections that requires noise addition, $\tilde N_+$, is also tightly related to the $P_+(\|u\|,p)$ (Proposition~\ref{prop:num-of-changes} and (\ref{eqn:p-vars-prob-bound})). Particularly, $\tilde N_+$ would be small when the data has a relatively large norm compared with the change in neighboring data $\beta$. Therefore, both iDP-SignRP methods would have better utility when the data norm is large. The reduction from $k$ to $\tilde N_+$ also leads to smaller Gaussian noise or smaller flipping probability for the values that need to be perturbed. Specifically, note that in Algorithm~\ref{alg:DP-SignRP-gaussian-noise-individual}, the optimal Gaussian mechanism is deployed with sensitivity $\Delta_2=\beta\sqrt{\frac{\tilde N_+}{k}}$, instead of $\Delta_2=\beta$ as in (\ref{eq:l2-sensitivity-known}) for DP-RP-G-OPT.

\vspace{0.15in}
\noindent\textbf{iDP-SignOPORP}. Similarly, we can also apply iDP to the SignOPORP method. Basically, we only need to replace $x$ in Line 3 in both Algorithm~\ref{alg:DP-SignRP-gaussian-noise-individual} and Algorithm~\ref{alg:DP-signRP-RR-individual} by the OPORP of $u$. However, we note that this iDP-SignOPORP procedure is considerably worse than iDP-SignRP in performance. This is because, by the binning step in OPORP, the average scale of each projected value becomes much smaller. This implies that in Algorithm~\ref{alg:noise-indicator}, the magnitude of $z$ would be much smaller, so a lot more projected values will require perturbation, which leads to a utility loss. This illustrates the superiority of SignRP under iDP: since each RP aggregates the whole data vector, SignRP is more robust to a small change in the data. Hence, less noise is needed.

\subsection{Empirical Results on iDP}

\begin{figure}[h]

\vspace{-0.2in}

\centering

    \mbox{
    \includegraphics[width=2.2in]{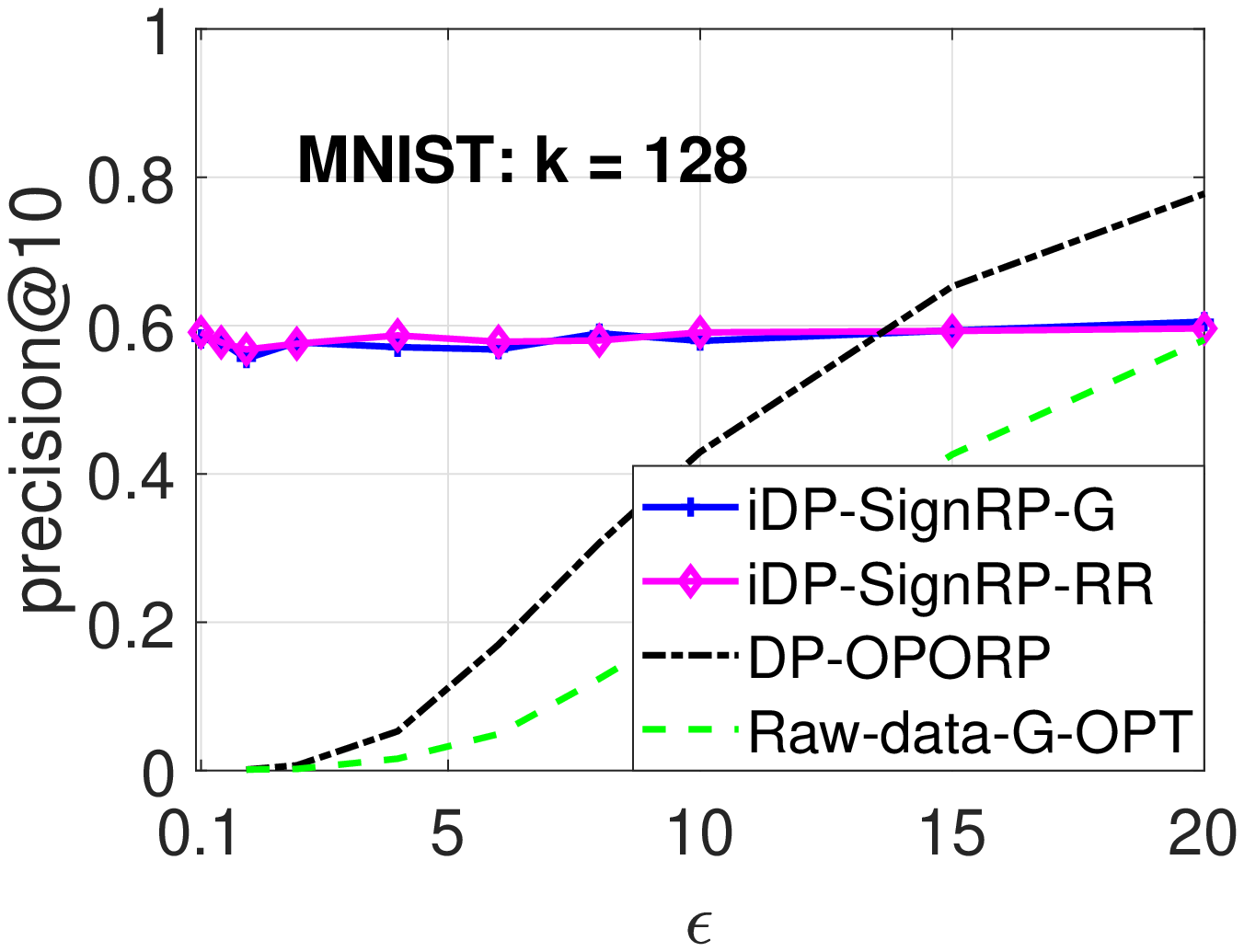} \hspace{-0.15in}
    \includegraphics[width=2.2in]{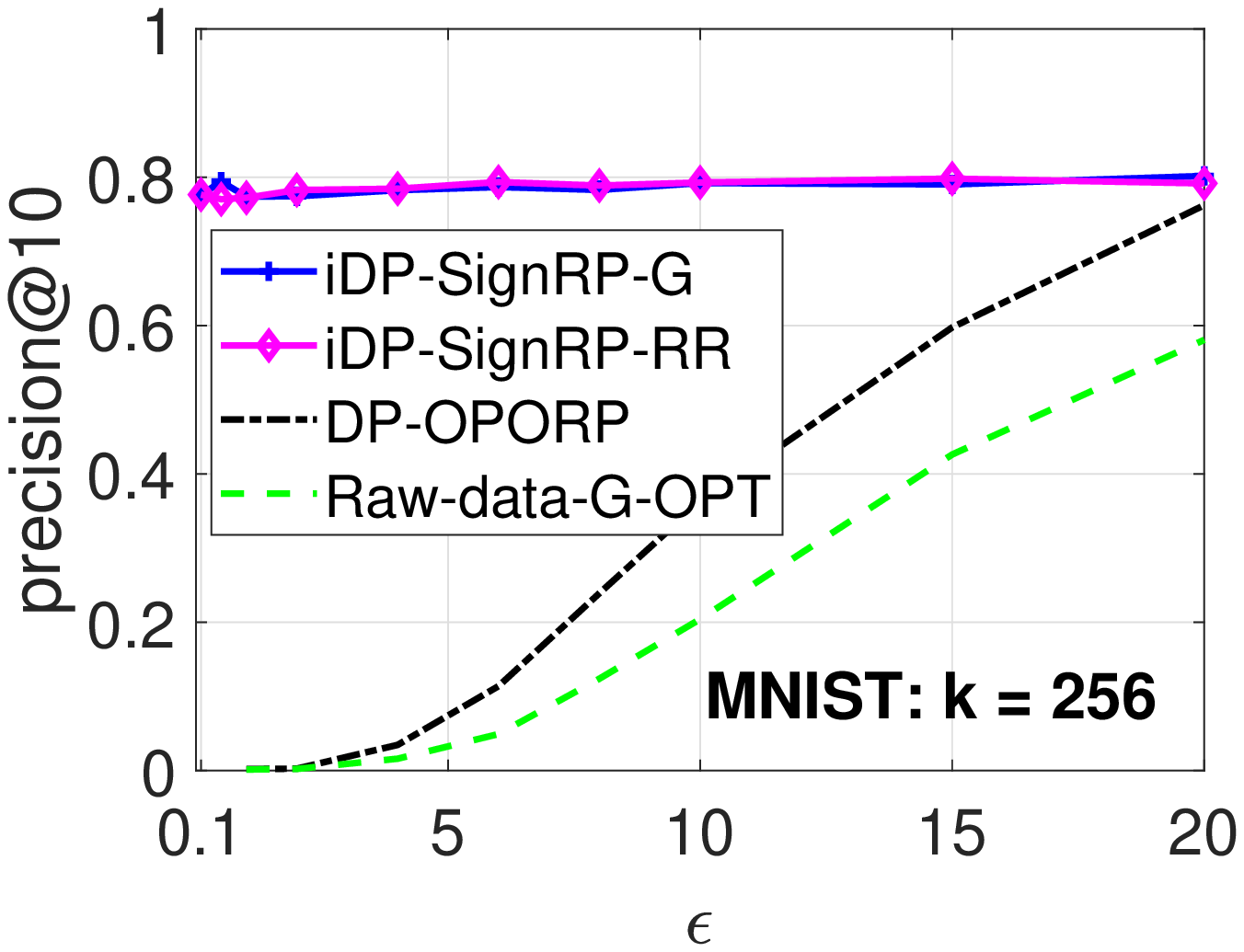}
    \hspace{-0.15in}
    \includegraphics[width=2.2in]{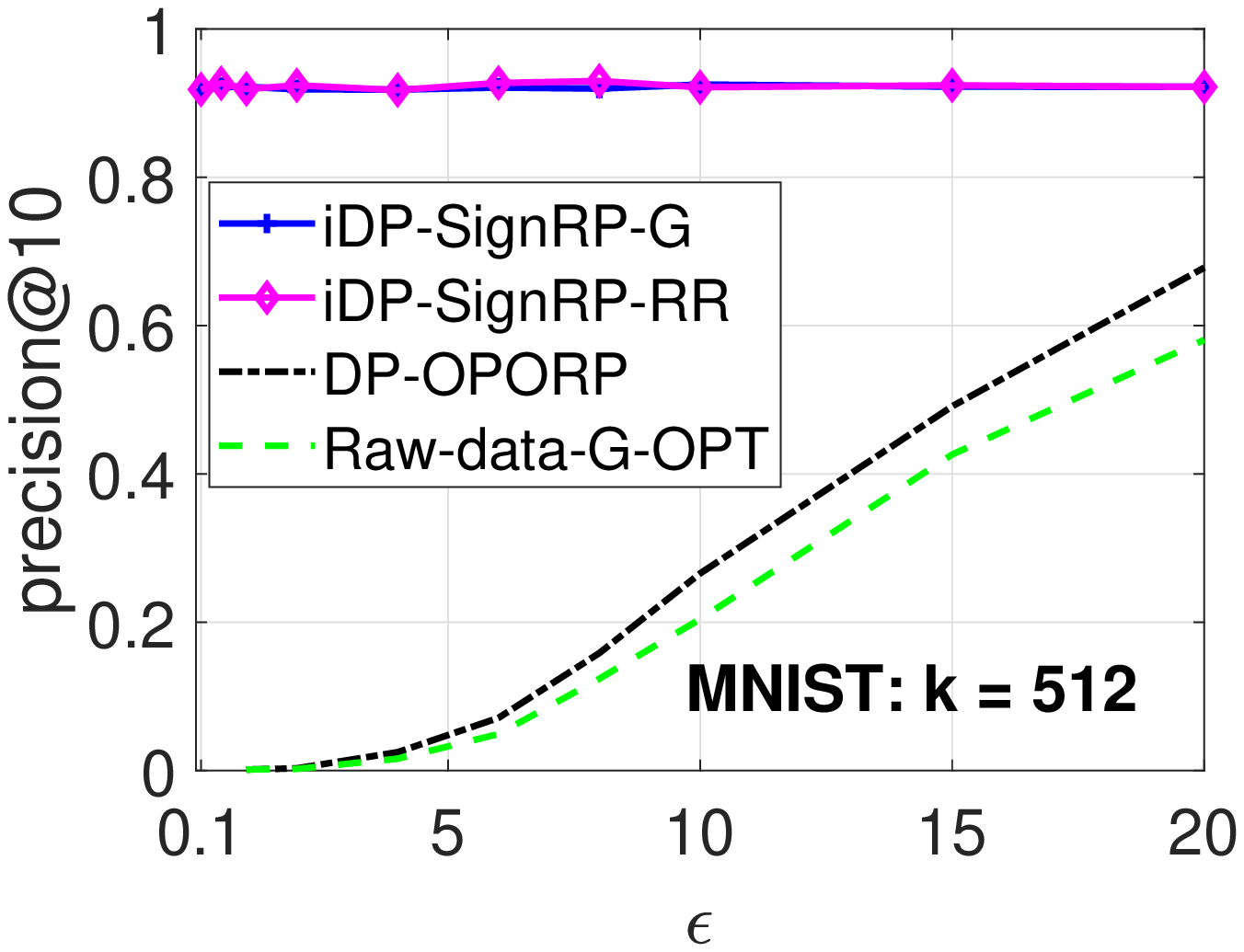}
    }

    \mbox{
    \includegraphics[width=2.2in]{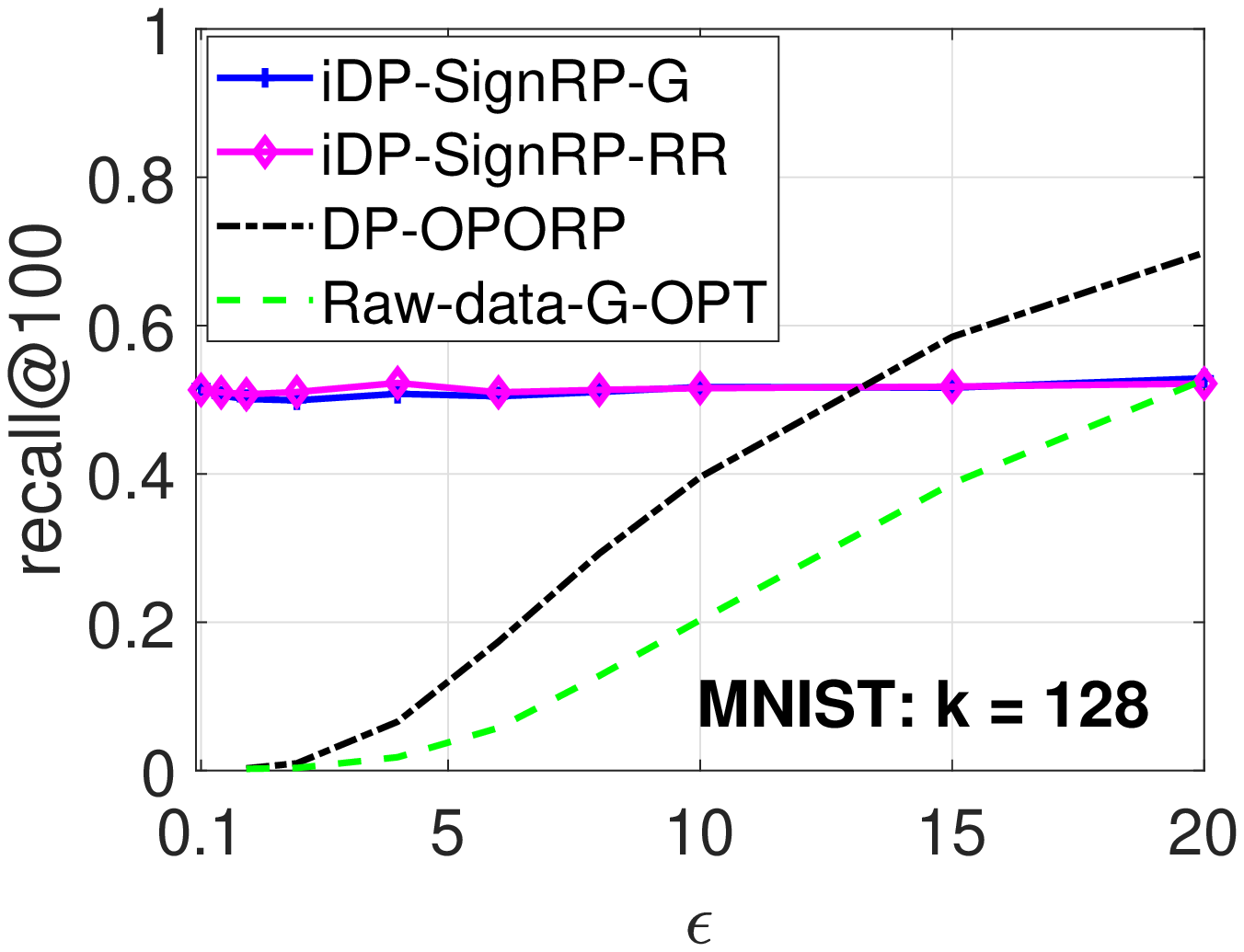} \hspace{-0.15in}
    \includegraphics[width=2.2in]{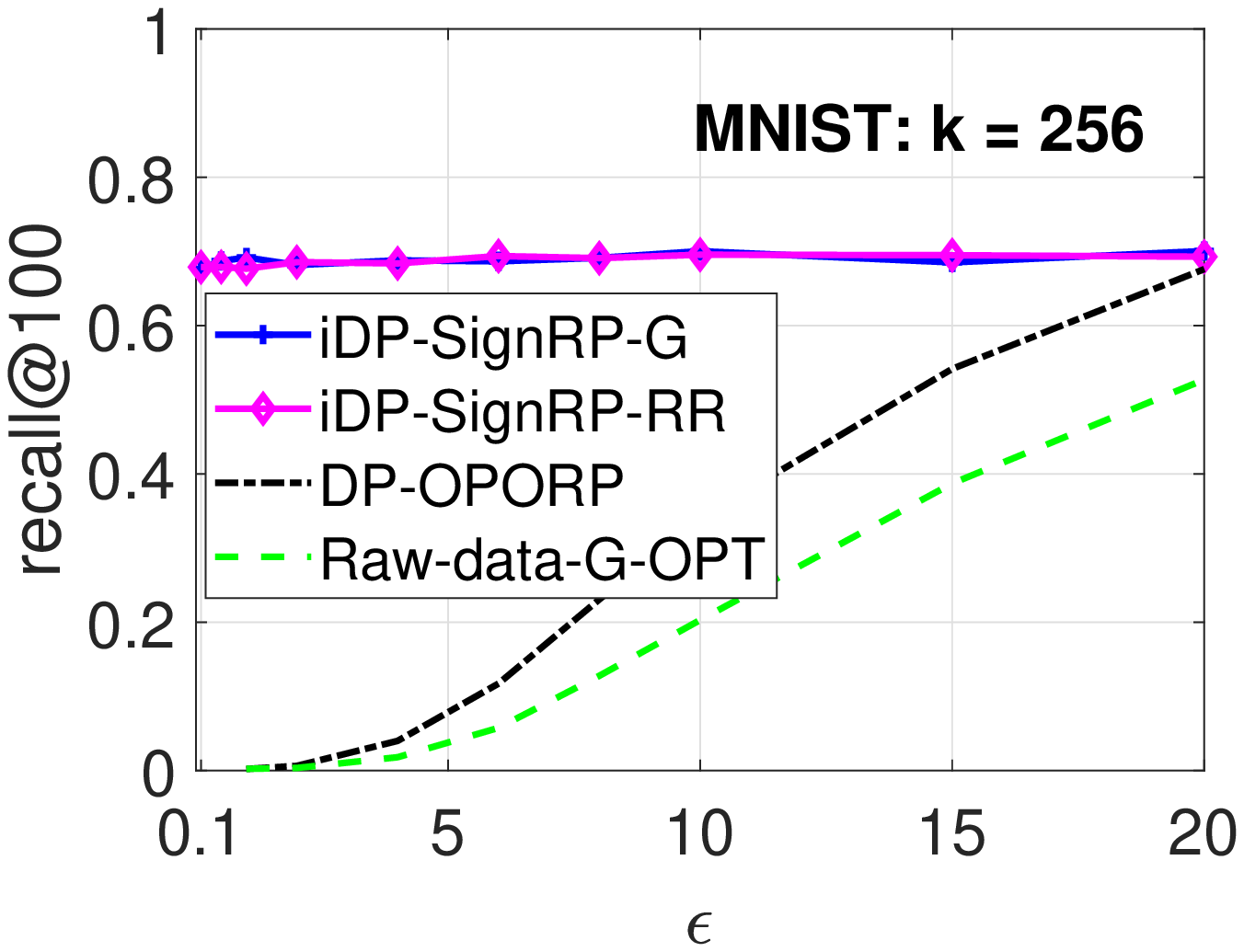}
    \hspace{-0.15in}
    \includegraphics[width=2.2in]{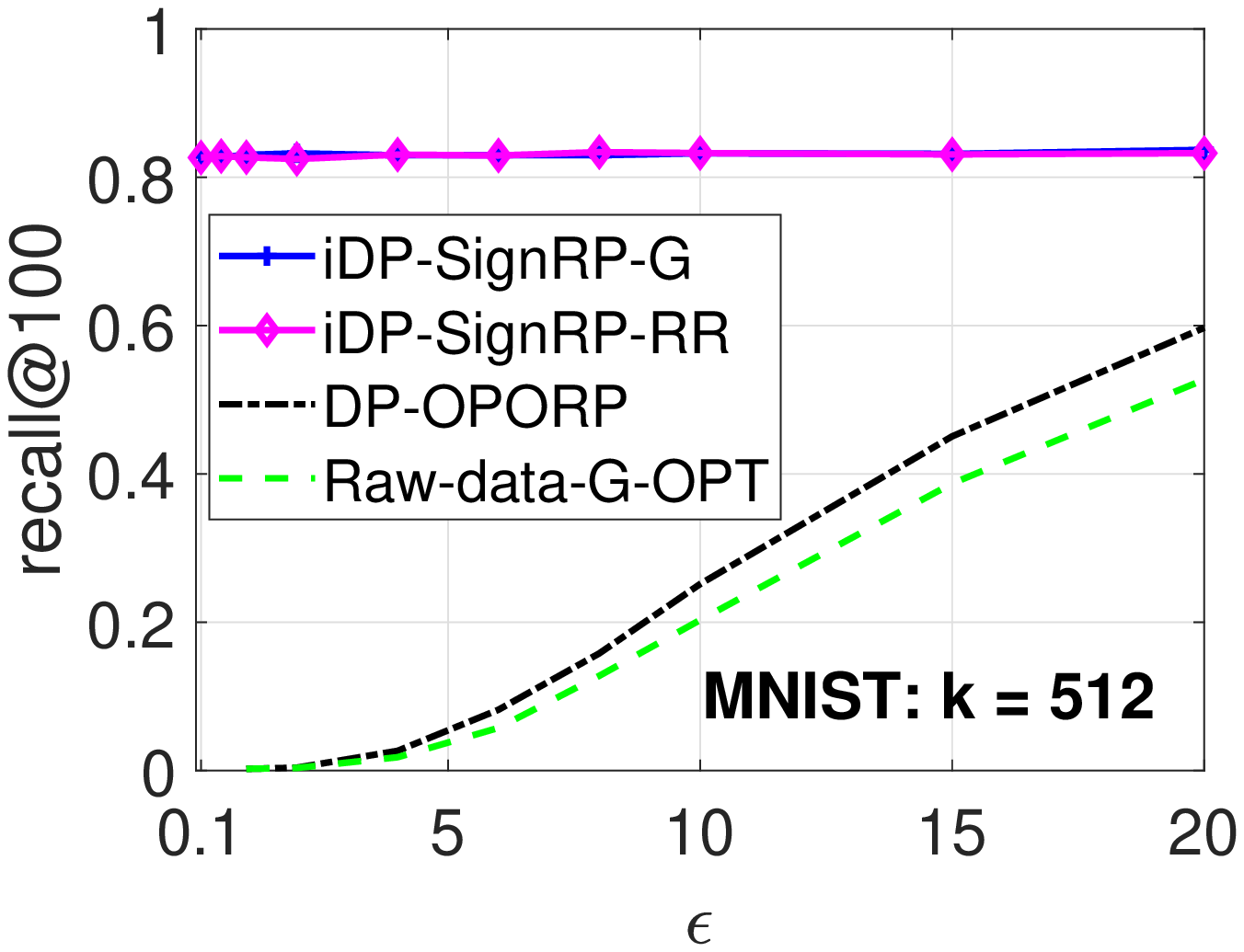}
    }

\vspace{-0.15in}

\caption{Retrieval on MNIST with iDP-SignRP, $\beta=1$, $\delta=10^{-6}$.}
\label{fig:MNIST_vs_eps_iDP}
\end{figure}

\begin{figure}[h]

\vspace{-0.15in}

\centering

    \mbox{
    \includegraphics[width=2.2in]{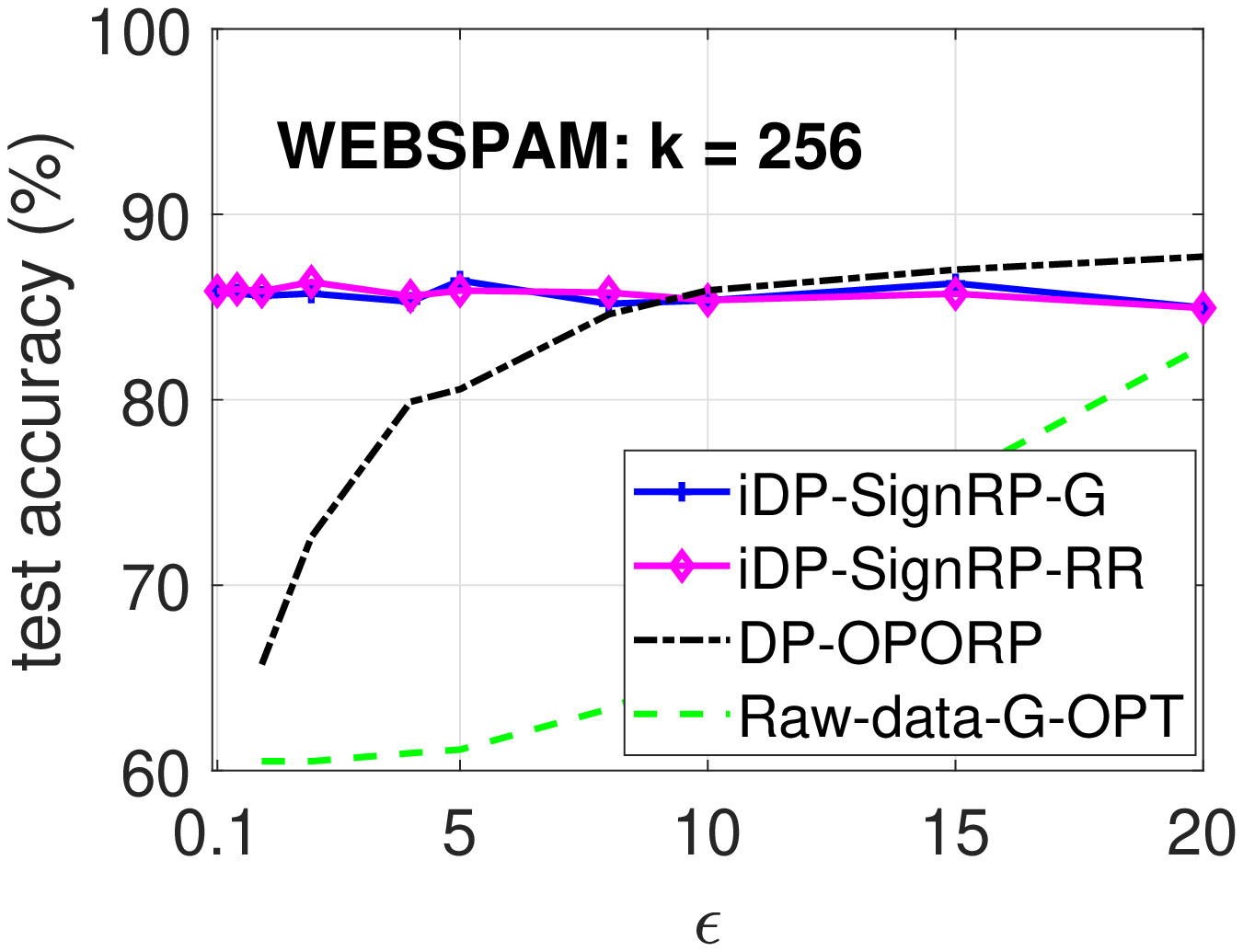} \hspace{-0.15in}
    \includegraphics[width=2.2in]{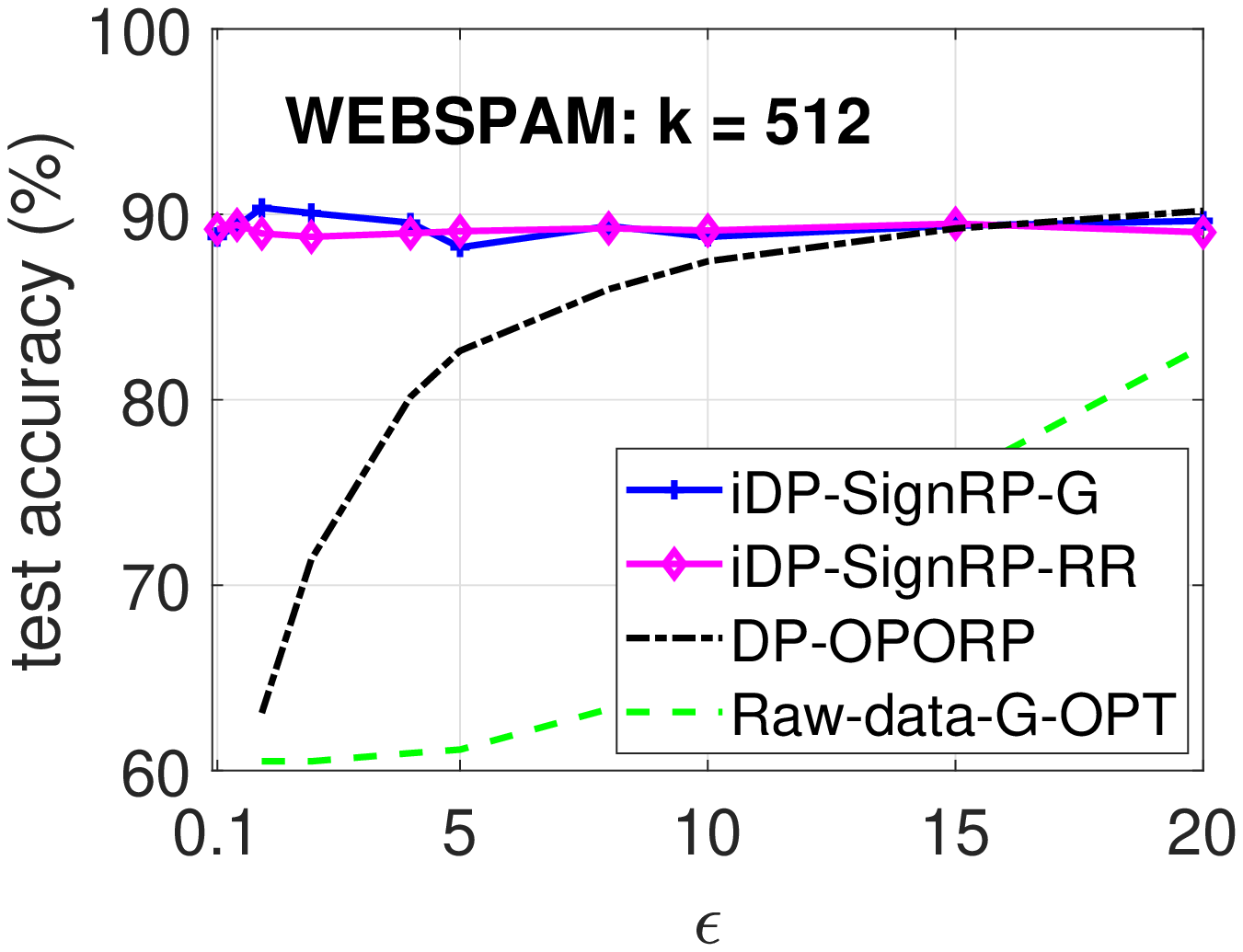} \hspace{-0.15in}
    \includegraphics[width=2.2in]{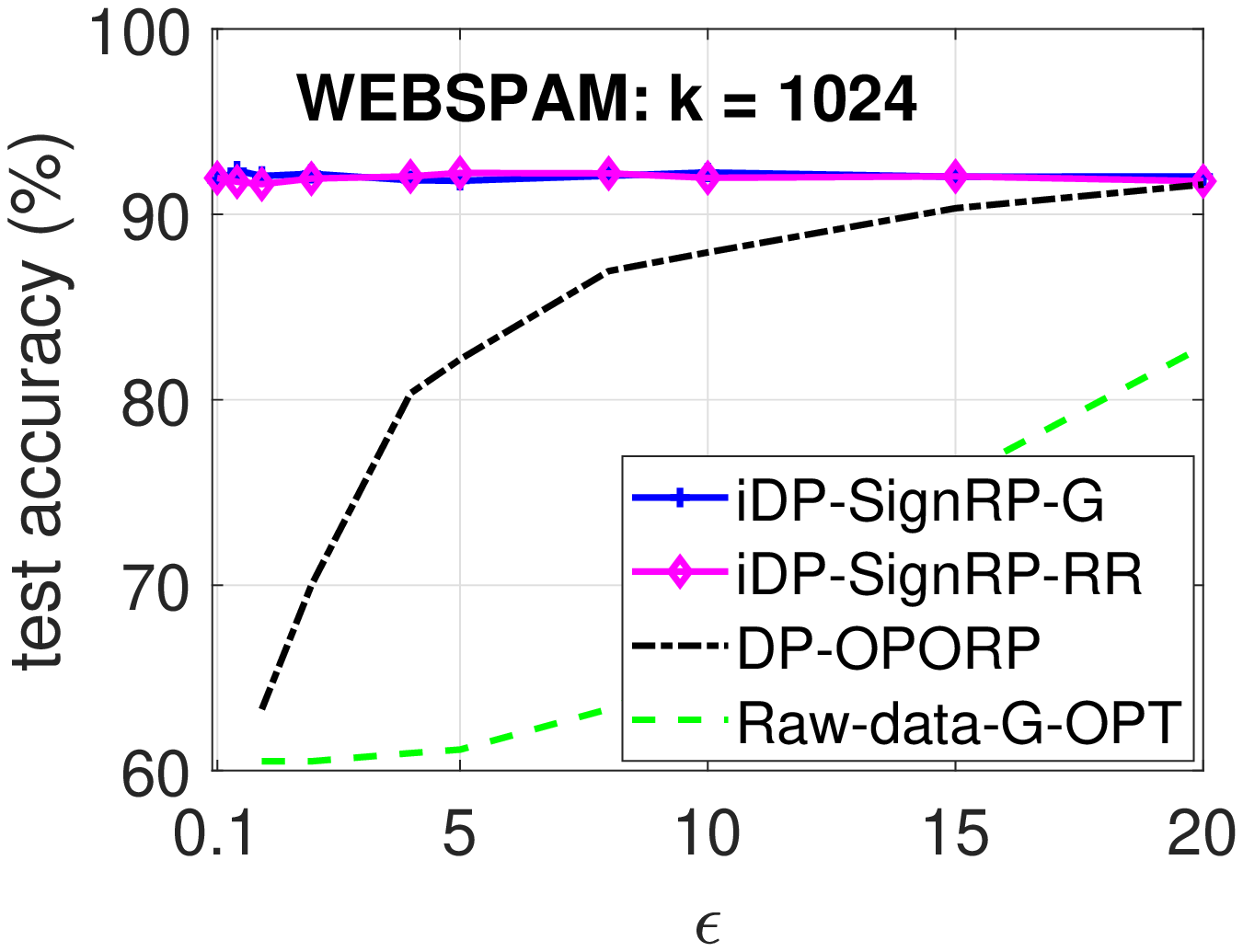}
    }

\vspace{-0.15in}

\caption{SVM on WEBSPAM with iDP-SignRP, $\beta=1$, $\delta=10^{-6}$.}
\label{fig:SVM_vs_eps_iDP}
\end{figure}


To demonstrate the empirical gain in utility of iDP-SignRP, we conduct the same set of experiments as in Section~\ref{sec:experiment}. Figure~\ref{fig:MNIST_vs_eps_iDP} reports the precision and recall on MNIST, and Figure~\ref{fig:SVM_vs_eps_iDP} presents the SVM test accuracy on WEBSPAM. As we can see, iDP-SignRP achieves very high utility even when $\epsilon<0.1$. We see that the curves of iDP-SignRP are almost flat. This is because only a small fraction of projected values are perturbed, so the untouched projected values already provide rich information for search and classification. In other words, the experimental results illustrate that the SignRP itself is already very strong in protecting the individual differential privacy. In other words, SignRP itself is already a strong method to protect the privacy of each specific dataset with respect to the individual DP, e.g., in non-interactive data publishing tasks.


\section{Conclusion} \label{sec:conclusion}
\vspace{-0.05in}

The  concept of differential privacy (DP) has been widely accepted as an elegant mathematical framework for protecting data privacy. In practice, however, deploying DP algorithms often leads to considerably worse performance (utility). For example, if we directly add noise to high-dimensional data vectors, the required noise level (i.e., the variance of the noise random variables) is often quite high which can severely degrade the performance of the subsequent tasks. Furthermore, inserting noise to the original data completely destroys data sparsity. For example, in text data or data for commercial advertising CTR models, for each data vector, typically only a very small fraction of the coordinates are non-zero. Adding noise to the original dataset can easily generate dense vectors in millions or even billions dimensions. Therefore, DP algorithms based on random projections (RP), which is an effective tool for dimension reduction, can be desirable. This is our motivation.

\vspace{0.1in}

\noindent Our key contribution is the observation that the signs of the randomly projected data vectors can be very stable when the original data vectors are perturbed. That is, we let $x = \sum_{i=1}^p u_i w_i$, where $u$ is the original $p$-dimensional data vector and $w$ is a projection vector. The sign $sign(x) = sign\left(\sum_{i=1}^p u_i w_i\right)$ can be changed if $u$ is modified to be $u^\prime$, with a probability which depends on the choice of the distribution of $w$. Our study reveals that such (sign-flipping) probability can be quite low (e.g., $<0.1$), under the common definition of ``neighbors'' in DP. This is the basis for our work.

\vspace{0.1in}

\noindent In this study, we have developed two lines of strategies for taking advantage of this sign-flipping probability. By leveraging the idea of smooth sensitivity, we propose a ``smooth flipping'' strategy which satisfies the rigorous definition of DP. The algorithm we recommend is named \textbf{DP-SignOPORP} by using a variant of the count-sketch for generating the projection vectors. In a recent work, \citet{li2023oporp} proposed ``one-permutation + one random projection'' (OPORP) which improves the original count-sketch in two aspects: (i) fixed-length binning, and (ii) normalization. In this study, we have observed that DP-SignOPORP noticeably outperforms other DP-RP algorithms.

\vspace{0.1in}

\noindent We have also exploited the other direction for utilizing the sign-flipping probability of random projections, by leveraging the concept of ``individual differential privacy'' (iDP) and ``local sensitivity''. The proposed  ``iDP-SignRP'' algorithms perform remarkably well, in that they are able to achieve excellent utility even at very small $\epsilon$ values (e.g., $\epsilon<0.5$). Although iDP does not strictly satisfy  DP, it is anticipated that iDP might be still useful for certain applications such as releasing datasets.

\vspace{0.1in}
\noindent The family of DP-RP and DP-SignRP algorithms can be conveniently used in training AI models for protecting the privacy of input features. A notable use case would be the embedding-based retrieval (EBR), which is nowadays widely adopted in commercial applications. Embeddings are trained/stored and then might be shared across different countries/companies/business units. For example, it is common that embeddings trained from one model might be used as the input feature data for other models. Understandably, privacy alarms might be  triggered in those applications.



\section*{Acknowledgement}

The authors would like to thank colleagues from Microsoft Research: Janardhan Kulkarni, Yin Tat Lee, Sergey Yekhanin, for helpful discussions.

\bibliography{refs_scholar}

\begin{thebibliography}{99}
\providecommand{\natexlab}[1]{#1}
\providecommand{\url}[1]{\texttt{#1}}
\expandafter\ifx\csname urlstyle\endcsname\relax
  \providecommand{\doi}[1]{doi: #1}\else
  \providecommand{\doi}{doi: \begingroup \urlstyle{rm}\Url}\fi

\bibitem[Abadi et~al.(2016)Abadi, Chu, Goodfellow, McMahan, Mironov, Talwar,
  and Zhang]{abadi2016deep}
Mart{\'{\i}}n Abadi, Andy Chu, Ian~J. Goodfellow, H.~Brendan McMahan, Ilya
  Mironov, Kunal Talwar, and Li~Zhang.
\newblock Deep learning with differential privacy.
\newblock In \emph{Proceedings of the 2016 {ACM} {SIGSAC} Conference on
  Computer and Communications Security (CCS)}, pages 308--318, Vienna, Austria,
  2016.

\bibitem[Achlioptas(2003)]{achlioptas2003database}
Dimitris Achlioptas.
\newblock Database-friendly random projections: Johnson-lindenstrauss with
  binary coins.
\newblock \emph{J. Comput. Syst. Sci.}, 66\penalty0 (4):\penalty0 671--687,
  2003.

\bibitem[Agarwal et~al.(2018)Agarwal, Suresh, Yu, Kumar, and
  McMahan]{agarwal2018cpsgd}
Naman Agarwal, Ananda~Theertha Suresh, Felix~X. Yu, Sanjiv Kumar, and Brendan
  McMahan.
\newblock {cpSGD}: Communication-efficient and differentially-private
  distributed {SGD}.
\newblock In \emph{Advances in Neural Information Processing Systems
  (NeurIPS)}, pages 7575--7586, Montr{\'{e}}al, Canada, 2018.

\bibitem[AlOmar et~al.(2021)AlOmar, Aljedaani, Tamjeed, Mkaouer, and
  El{-}Glaly]{abdullah2021finding}
Eman~Abdullah AlOmar, Wajdi Aljedaani, Murtaza Tamjeed, Mohamed~Wiem Mkaouer,
  and Yasmine~N. El{-}Glaly.
\newblock Finding the needle in a haystack: On the automatic identification of
  accessibility user reviews.
\newblock In \emph{Proceedings of the Conference on Human Factors in Computing
  Systems (CHI)}, pages 387:1--387:15, Virtual Event / Yokohama, Japan, 2021.

\bibitem[Balle and Wang(2018)]{balle2018improving}
Borja Balle and Yu{-}Xiang Wang.
\newblock Improving the gaussian mechanism for differential privacy: Analytical
  calibration and optimal denoising.
\newblock In \emph{Proceedings of the 35th International Conference on Machine
  Learning (ICML)}, pages 403--412, Stockholmsm{\"{a}}ssan, Stockholm, Sweden,
  2018.

\bibitem[Beigi et~al.(2020)Beigi, Mosallanezhad, Guo, Alvari, Nou, and
  Liu]{beigi2020privacy}
Ghazaleh Beigi, Ahmadreza Mosallanezhad, Ruocheng Guo, Hamidreza Alvari,
  Alexander Nou, and Huan Liu.
\newblock Privacy-aware recommendation with private-attribute protection using
  adversarial learning.
\newblock In \emph{Proceedings of the Thirteenth {ACM} International Conference
  on Web Search and Data Mining (WSDM)}, pages 34--42, Houston, TX, USA, 2020.

\bibitem[Berlioz et~al.(2015)Berlioz, Friedman, K{\^{a}}afar, Boreli, and
  Berkovsky]{berlioz2015applying}
Arnaud Berlioz, Arik Friedman, Mohamed~Ali K{\^{a}}afar, Roksana Boreli, and
  Shlomo Berkovsky.
\newblock Applying differential privacy to matrix factorization.
\newblock In \emph{Proceedings of the 9th {ACM} Conference on Recommender
  Systems (RecSys)}, pages 107--114, Vienna, Austria, 2015.

\bibitem[Bingham and Mannila(2001)]{bingham2001random}
Ella Bingham and Heikki Mannila.
\newblock Random projection in dimensionality reduction: Applications to image
  and text data.
\newblock In \emph{Proceedings of the Seventh {ACM} {SIGKDD} International
  Conference on Knowledge Discovery and Data Mining (KDD)}, pages 245--250, San
  Francisco, CA, 2001.

\bibitem[Blocki et~al.(2012)Blocki, Blum, Datta, and
  Sheffet]{blocki2012johnson}
Jeremiah Blocki, Avrim Blum, Anupam Datta, and Or~Sheffet.
\newblock The johnson-lindenstrauss transform itself preserves differential
  privacy.
\newblock In \emph{Proceedings of the 53rd Annual {IEEE} Symposium on
  Foundations of Computer Science (FOCS)}, pages 410--419, New Brunswick, NJ,
  2012.

\bibitem[Blum et~al.(2005)Blum, Dwork, McSherry, and Nissim]{blum2005practical}
Avrim Blum, Cynthia Dwork, Frank McSherry, and Kobbi Nissim.
\newblock Practical privacy: the {SuLQ} framework.
\newblock In \emph{Proceedings of the Twenty-fourth {ACM}
  {SIGACT-SIGMOD-SIGART} Symposium on Principles of Database Systems (PODS)},
  pages 128--138, Baltimore, MD, 2005.

\bibitem[Boufounos and Baraniuk(2008)]{boufounos20081bit}
Petros Boufounos and Richard~G. Baraniuk.
\newblock 1-bit compressive sensing.
\newblock In \emph{Proceedings of the 42nd Annual Conference on Information
  Sciences and Systems (CISS)}, pages 16--21, Princeton, NJ, 2008.

\bibitem[Cand{\`{e}}s et~al.(2006)Cand{\`{e}}s, Romberg, and
  Tao]{candes2006robust}
Emmanuel~J. Cand{\`{e}}s, Justin~K. Romberg, and Terence Tao.
\newblock Robust uncertainty principles: exact signal reconstruction from
  highly incomplete frequency information.
\newblock \emph{{IEEE} Trans. Inf. Theory}, 52\penalty0 (2):\penalty0 489--509,
  2006.

\bibitem[Chang and Lin(2011)]{chang2011libsvm}
Chih-Chung Chang and Chih-Jen Lin.
\newblock Libsvm: a library for support vector machines.
\newblock \emph{ACM Transactions on Intelligent Systems and Technology (TIST)},
  2\penalty0 (3):\penalty0 27, 2011.

\bibitem[Charikar et~al.(2004)Charikar, Chen, and
  Farach-Colton]{charikar2004finding}
Moses Charikar, Kevin Chen, and Martin Farach-Colton.
\newblock Finding frequent items in data streams.
\newblock \emph{Theor. Comput. Sci.}, 312\penalty0 (1):\penalty0 3--15, 2004.

\bibitem[Charikar(2002)]{charikar2002similarity}
Moses~S Charikar.
\newblock Similarity estimation techniques from rounding algorithms.
\newblock In \emph{Proceedings of the Thiry-Fourth Annual ACM Symposium on
  Theory of Computing (STOC)}, pages 380--388, Montreal, Canada, 2002.

\bibitem[Chaudhuri and Monteleoni(2008)]{chaudhuri2008privacy}
Kamalika Chaudhuri and Claire Monteleoni.
\newblock Privacy-preserving logistic regression.
\newblock In \emph{Advances in Neural Information Processing Systems (NIPS)},
  pages 289--296, Vancouver, Canada, 2008.

\bibitem[Chaudhuri et~al.(2011)Chaudhuri, Monteleoni, and
  Sarwate]{chaudhuri2011differentially}
Kamalika Chaudhuri, Claire Monteleoni, and Anand~D. Sarwate.
\newblock Differentially private empirical risk minimization.
\newblock \emph{J. Mach. Learn. Res.}, 12:\penalty0 1069--1109, 2011.

\bibitem[Chen et~al.(2017)Chen, Fisch, Weston, and Bordes]{chen2017reading}
Danqi Chen, Adam Fisch, Jason Weston, and Antoine Bordes.
\newblock Reading wikipedia to answer open-domain questions.
\newblock In \emph{Proceedings of the 55th Annual Meeting of the Association
  for Computational Linguistics (ACL)}, pages 1870--1879, Vancouver, Canada,
  2017.

\bibitem[Chen et~al.(2015)Chen, Wilson, Tyree, Weinberger, and
  Chen]{chen2015compressing}
Wenlin Chen, James Wilson, Stephen Tyree, Kilian Weinberger, and Yixin Chen.
\newblock {Compressing Neural Networks with the Hashing Trick}.
\newblock In \emph{Proceedings of the 32nd International Conference on Machine
  Learning (ICML)}, pages 2285--2294, Lille, France, 2015.

\bibitem[Cormode et~al.(2018{\natexlab{a}})Cormode, Jha, Kulkarni, Li,
  Srivastava, and Wang]{cormode2018privacy}
Graham Cormode, Somesh Jha, Tejas Kulkarni, Ninghui Li, Divesh Srivastava, and
  Tianhao Wang.
\newblock Privacy at scale: Local differential privacy in practice.
\newblock In \emph{Proceedings of the 2018 International Conference on
  Management of Data (SIGMOD)}, pages 1655--1658, Houston, TX,
  2018{\natexlab{a}}.

\bibitem[Cormode et~al.(2018{\natexlab{b}})Cormode, Kulkarni, and
  Srivastava]{cormode2018marginal}
Graham Cormode, Tejas Kulkarni, and Divesh Srivastava.
\newblock Marginal release under local differential privacy.
\newblock In \emph{Proceedings of the 2018 International Conference on
  Management of Data (SIGMOD)}, pages 131--146, Houston, TX,
  2018{\natexlab{b}}.

\bibitem[Cortes and Vapnik(1995)]{cortes1995support}
Corinna Cortes and Vladimir Vapnik.
\newblock Support-vector networks.
\newblock \emph{Mach. Learn.}, 20\penalty0 (3):\penalty0 273--297, 1995.

\bibitem[Dahl et~al.(2013)Dahl, Stokes, Deng, and Yu]{dahl2013large}
George~E. Dahl, Jack~W. Stokes, Li~Deng, and Dong Yu.
\newblock Large-scale malware classification using random projections and
  neural networks.
\newblock In \emph{Proceedings of the {IEEE} International Conference on
  Acoustics, Speech and Signal Processing (ICASSP)}, pages 3422--3426,
  Vancouver, Canada, 2013.

\bibitem[Dasgupta(2000)]{dasgupta2000experiments}
Sanjoy Dasgupta.
\newblock Experiments with random projection.
\newblock In \emph{Proceedings of the 16th Conference in Uncertainty in
  Artificial Intelligence (UAI)}, pages 143--151, Stanford, CA, 2000.

\bibitem[Dasgupta and Freund(2008)]{dasgupta2008random}
Sanjoy Dasgupta and Yoav Freund.
\newblock Random projection trees and low dimensional manifolds.
\newblock In \emph{Proceedings of the 40th Annual {ACM} Symposium on Theory of
  Computing (STOC)}, pages 537--546, Victoria, Canada, 2008.

\bibitem[Datar et~al.(2004)Datar, Immorlica, Indyk, and
  Mirrokni]{datar2004locality}
Mayur Datar, Nicole Immorlica, Piotr Indyk, and Vahab~S Mirrokni.
\newblock Locality-sensitive hashing scheme based on p-stable distributions.
\newblock In \emph{Proceedings of the Twentieth Annual Symposium on
  Computational Geometry (SCG)}, pages 253--262, Brooklyn, NY, 2004.

\bibitem[Desfontaines and Pej{\'{o}}(2020)]{desfontaines2020sok}
Damien Desfontaines and Bal{\'{a}}zs Pej{\'{o}}.
\newblock Sok: Differential privacies.
\newblock \emph{Proc. Priv. Enhancing Technol.}, 2020\penalty0 (2):\penalty0
  288--313, 2020.

\bibitem[Dickens et~al.(2022)Dickens, Thaler, and Ting]{dickens2022order}
Charlie Dickens, Justin Thaler, and Daniel Ting.
\newblock Order-invariant cardinality estimators are differentially private.
\newblock In \emph{Advances in Neural Information Processing Systems
  (NeurIPS)}, New Orleans, LA, 2022.

\bibitem[Ding et~al.(2017)Ding, Kulkarni, and Yekhanin]{ding2017collecting}
Bolin Ding, Janardhan Kulkarni, and Sergey Yekhanin.
\newblock Collecting telemetry data privately.
\newblock In \emph{Advances in Neural Information Processing Systems (NIPS)},
  pages 3571--3580, Long Beach, CA, 2017.

\bibitem[Dong et~al.(2022)Dong, Roth, and Su]{dong2022gaussian}
Jinshuo Dong, Aaron Roth, and Weijie~J Su.
\newblock Gaussian differential privacy.
\newblock \emph{Journal of the Royal Statistical Society Series B: Statistical
  Methodology}, 84\penalty0 (1):\penalty0 3--37, 2022.

\bibitem[Dong et~al.(2008)Dong, Charikar, and Li]{dong2008asymmetric}
Wei Dong, Moses Charikar, and Kai Li.
\newblock Asymmetric distance estimation with sketches for similarity search in
  high-dimensional spaces.
\newblock In \emph{Proceedings of the 31st Annual International {ACM} {SIGIR}
  Conference on Research and Development in Information Retrieval (SIGIR)},
  pages 123--130, 2008.

\bibitem[Donoho(2006)]{donoho2006compressed}
David~L. Donoho.
\newblock Compressed sensing.
\newblock \emph{{IEEE} Trans. Inf. Theory}, 52\penalty0 (4):\penalty0
  1289--1306, 2006.

\bibitem[Dwork and Lei(2009)]{dwork2009differential}
Cynthia Dwork and Jing Lei.
\newblock Differential privacy and robust statistics.
\newblock In \emph{Proceedings of the 41st Annual {ACM} Symposium on Theory of
  Computing (STOC)}, pages 371--380, Bethesda, MD, 2009.

\bibitem[Dwork and Roth(2014)]{dwork2014algorithmic}
Cynthia Dwork and Aaron Roth.
\newblock The algorithmic foundations of differential privacy.
\newblock \emph{Found. Trends Theor. Comput. Sci.}, 9\penalty0 (3-4):\penalty0
  211--407, 2014.

\bibitem[Dwork and Rothblum(2016)]{dwork2016concentrated}
Cynthia Dwork and Guy~N Rothblum.
\newblock Concentrated differential privacy.
\newblock \emph{arXiv preprint arXiv:1603.01887}, 2016.

\bibitem[Dwork et~al.(2006)Dwork, McSherry, Nissim, and
  Smith]{dwork2006calibrating}
Cynthia Dwork, Frank McSherry, Kobbi Nissim, and Adam~D. Smith.
\newblock Calibrating noise to sensitivity in private data analysis.
\newblock In \emph{Proceedings of the Third Theory of Cryptography Conference
  (TCC)}, pages 265--284, New York, NY, 2006.

\bibitem[Erlingsson et~al.(2014)Erlingsson, Pihur, and
  Korolova]{erlingsson2014rappor}
{\'{U}}lfar Erlingsson, Vasyl Pihur, and Aleksandra Korolova.
\newblock {RAPPOR:} randomized aggregatable privacy-preserving ordinal
  response.
\newblock In \emph{Proceedings of the 2014 {ACM} {SIGSAC} Conference on
  Computer and Communications Security (CCS)}, pages 1054--1067, Scottsdale,
  AZ, 2014.

\bibitem[Fan and Li(2022)]{fan2022distances}
Chenglin Fan and Ping Li.
\newblock Distances release with differential privacy in tree and grid graph.
\newblock In \emph{{IEEE} International Symposium on Information Theory
  (ISIT)}, pages 2190--2195, 2022.

\bibitem[Fan et~al.(2022)Fan, Li, and Li]{fan2022private}
Chenglin Fan, Ping Li, and Xiaoyun Li.
\newblock Private graph all-pairwise-shortest-path distance release with
  improved error rate.
\newblock In \emph{Advances in Neural Information Processing Systems
  (NeurIPS)}, New Orleans, LA, 2022.

\bibitem[Fang et~al.(2023)Fang, Li, Fan, and Li]{fang2023improved}
Huang Fang, Xiaoyun Li, Chenglin Fan, and Ping Li.
\newblock Improved convergence of differential private sgd with gradient
  clipping.
\newblock In \emph{Proceedings of the Eleventh International Conference on
  Learning Representations (ICLR)}, Kigali, Rwanda, 2023.

\bibitem[Feldman et~al.(2009)Feldman, Fiat, Kaplan, and
  Nissim]{feldman2009private}
Dan Feldman, Amos Fiat, Haim Kaplan, and Kobbi Nissim.
\newblock Private coresets.
\newblock In \emph{Proceedings of the 41st Annual {ACM} Symposium on Theory of
  Computing (STOC)}, pages 361--370, Bethesda, MD, 2009.

\bibitem[Fern and Brodley(2003)]{fern2003random}
Xiaoli~Zhang Fern and Carla~E. Brodley.
\newblock Random projection for high dimensional data clustering: A cluster
  ensemble approach.
\newblock In \emph{Proceedings of the Twentieth International Conference
  (ICML)}, pages 186--193, Washington, DC, 2003.

\bibitem[Freund et~al.(2007)Freund, Dasgupta, Kabra, and
  Verma]{frund2007leanring}
Yoav Freund, Sanjoy Dasgupta, Mayank Kabra, and Nakul Verma.
\newblock Learning the structure of manifolds using random projections.
\newblock In \emph{Advances in Neural Information Processing Systems (NIPS)},
  pages 473--480, Vancouver, Canada, 2007.

\bibitem[Friedman et~al.(1975)Friedman, Baskett, and
  Shustek]{friedman1975algorithm}
Jerome~H. Friedman, F.~Baskett, and L.~Shustek.
\newblock An algorithm for finding nearest neighbors.
\newblock \emph{IEEE Transactions on Computers}, 24:\penalty0 1000--1006, 1975.

\bibitem[Ge et~al.(2018)Ge, Wang, Wang, and Liu]{ge2018minimax}
Jason Ge, Zhaoran Wang, Mengdi Wang, and Han Liu.
\newblock Minimax-optimal privacy-preserving sparse {PCA} in distributed
  systems.
\newblock In \emph{Proceedings of the International Conference on Artificial
  Intelligence and Statistics (AISTATS)}, pages 1589--1598, Playa Blanca,
  Lanzarote, Canary Islands, Spain, 2018.

\bibitem[Goemans and Williamson(1995)]{goemans1995improved}
Michel~X. Goemans and David~P. Williamson.
\newblock Improved approximation algorithms for maximum cut and satisfiability
  problems using semidefinite programming.
\newblock \emph{J. {ACM}}, 42\penalty0 (6):\penalty0 1115--1145, 1995.

\bibitem[Gupta et~al.(2010)Gupta, Ligett, McSherry, Roth, and
  Talwar]{gupta2010differentially}
Anupam Gupta, Katrina Ligett, Frank McSherry, Aaron Roth, and Kunal Talwar.
\newblock Differentially private combinatorial optimization.
\newblock In \emph{Proceedings of the Twenty-First Annual {ACM-SIAM} Symposium
  on Discrete Algorithms (SODA)}, pages 1106--1125, Austin, TX, 2010.

\bibitem[Haddadpour et~al.(2020)Haddadpour, Karimi, Li, and
  Li]{haddadpour2020fedsketch}
Farzin Haddadpour, Belhal Karimi, Ping Li, and Xiaoyun Li.
\newblock Fedsketch: Communication-efficient and private federated learning via
  sketching.
\newblock \emph{arXiv preprint arXiv:2008.04975}, 2020.

\bibitem[Haeberlen and Khanna(2014)]{haeberlen2014differential}
Justin Hsu Marco Gaboardi~Andreas Haeberlen and Sanjeev Khanna.
\newblock Differential privacy: An economic method for choosing epsilon.
\newblock \emph{arXiv preprint arXiv:1402.3329}, 2014.

\bibitem[Indyk and Motwani(1998)]{indyk1998approximate}
Piotr Indyk and Rajeev Motwani.
\newblock Approximate nearest neighbors: Towards removing the curse of
  dimensionality.
\newblock In \emph{Proceedings of the Thirtieth Annual {ACM} Symposium on the
  Theory of Computing (STOC)}, pages 604--613, Dallas, TX, 1998.

\bibitem[Jain et~al.(2018)Jain, Thakkar, and Thakurta]{jain2018differentially}
Prateek Jain, Om~Dipakbhai Thakkar, and Abhradeep Thakurta.
\newblock Differentially private matrix completion revisited.
\newblock In \emph{Proceedings of the 35th International Conference on Machine
  Learning (ICML)}, pages 2220--2229, Stockholmsm{\"{a}}ssan, Stockholm,
  Sweden, 2018.

\bibitem[Jayaraman and Evans(2019)]{jayaraman2019evaluating}
Bargav Jayaraman and David Evans.
\newblock Evaluating differentially private machine learning in practice.
\newblock In \emph{Proceedings of the 28th {USENIX} Security Symposium (USENIX
  Security)}, pages 1895--1912, Santa Clara, CA, 2019.

\bibitem[Johnson and Lindenstrauss(1984)]{johnson1984extensions}
William~B. Johnson and Joram Lindenstrauss.
\newblock Extensions of \text{Lipschitz} mapping into \text{Hilbert} space.
\newblock \emph{Contemporary Mathematics}, 26:\penalty0 189--206, 1984.

\bibitem[Kairouz et~al.(2014)Kairouz, Oh, and Viswanath]{kairouz2014extremal}
Peter Kairouz, Sewoong Oh, and Pramod Viswanath.
\newblock Extremal mechanisms for local differential privacy.
\newblock In \emph{Advances in Neural Information Processing Systems
  (NeurIPS)}, pages 2879--2887, Montreal, Canada, 2014.

\bibitem[Kasiviswanathan et~al.(2013)Kasiviswanathan, Nissim, Raskhodnikova,
  and Smith]{kasiviswanathan2013analyzing}
Shiva~Prasad Kasiviswanathan, Kobbi Nissim, Sofya Raskhodnikova, and Adam~D.
  Smith.
\newblock Analyzing graphs with node differential privacy.
\newblock In \emph{Proceedings of the 10th Theory of Cryptography Conference
  (TCC)}, pages 457--476, Tokyo, Japan, 2013.

\bibitem[Kenny et~al.(2021)Kenny, Kuriwaki, McCartan, Rosenman, Simko, and
  Imai]{kenny2021use}
Christopher~T Kenny, Shiro Kuriwaki, Cory McCartan, Evan~TR Rosenman, Tyler
  Simko, and Kosuke Imai.
\newblock The use of differential privacy for census data and its impact on
  redistricting: The case of the 2020 us census.
\newblock \emph{Science advances}, 7\penalty0 (41):\penalty0 eabk3283, 2021.

\bibitem[Kenthapadi et~al.(2013)Kenthapadi, Korolova, Mironov, and
  Mishra]{kenthapadi2013privacy}
Krishnaram Kenthapadi, Aleksandra Korolova, Ilya Mironov, and Nina Mishra.
\newblock Privacy via the johnson-lindenstrauss transform.
\newblock \emph{J. Priv. Confidentiality}, 5\penalty0 (1), 2013.

\bibitem[Knudson et~al.(2016)Knudson, Saab, and Ward]{knudson2016one}
Karin Knudson, Rayan Saab, and Rachel Ward.
\newblock One-bit compressive sensing with norm estimation.
\newblock \emph{{IEEE} Trans. Inf. Theory}, 62\penalty0 (5):\penalty0
  2748--2758, 2016.

\bibitem[Krizhevsky and Hinton(2009)]{krizhevsky2009learning}
Alex Krizhevsky and Geoffrey Hinton.
\newblock Learning multiple layers of features from tiny images.
\newblock \emph{Technical Report, University of Toronto}, 2009.

\bibitem[Laurent and Massart(2000)]{laurent2000adaptive}
B{\'e}atrice Laurent and Pascal Massart.
\newblock Adaptive estimation of a quadratic functional by model selection.
\newblock \emph{The Annals of Statistics}, pages 1302--1338, 2000.

\bibitem[LeCun et~al.(1998)LeCun, Bottou, Bengio, and
  Haffner]{lecun1998gradient}
Yann LeCun, L{\'{e}}on Bottou, Yoshua Bengio, and Patrick Haffner.
\newblock Gradient-based learning applied to document recognition.
\newblock \emph{Proc. {IEEE}}, 86\penalty0 (11):\penalty0 2278--2324, 1998.

\bibitem[Leng et~al.(2014)Leng, Cheng, and Lu]{leng2014random}
Cong Leng, Jian Cheng, and Hanqing Lu.
\newblock Random subspace for binary codes learning in large scale image
  retrieval.
\newblock In \emph{Proceedings of the 37th International {ACM} {SIGIR}
  Conference on Research and Development in Information Retrieval (SIGIR)},
  pages 1031--1034, Gold Coast, Australia, 2014.

\bibitem[Li(2019)]{li2019sign}
Ping Li.
\newblock Sign-full random projections.
\newblock In \emph{Proceedings of the Thirty-Third {AAAI} Conference on
  Artificial Intelligence (AAAI)}, pages 4205--4212, Honolulu, HI, 2019.

\bibitem[Li and Li(2023{\natexlab{a}})]{li2023oporp}
Ping Li and Xiaoyun Li.
\newblock {OPORP}: One permutation + one random projection.
\newblock \emph{arXiv preprint arXiv:2302.03505}, 2023{\natexlab{a}}.

\bibitem[Li and Zhao(2022)]{li2022gcwsnet}
Ping Li and Weijie Zhao.
\newblock {GCWSNet}: Generalized consistent weighted sampling for scalable and
  accurate training of neural networks.
\newblock In \emph{Proceedings of the 31st {ACM} International Conference on
  Information and Knowledge Management (CIKM)}, Atlanta, GA, 2022.

\bibitem[Li et~al.(2006)Li, Hastie, and Church]{li2006very}
Ping Li, Trevor~J Hastie, and Kenneth~W Church.
\newblock Very sparse random projections.
\newblock In \emph{Proceedings of the 12th ACM SIGKDD international conference
  on Knowledge discovery and data mining (KDD)}, pages 287--296, Philadelphia,
  PA, 2006.

\bibitem[Li et~al.(2014)Li, Mitzenmacher, and Shrivastava]{li2014coding}
Ping Li, Michael Mitzenmacher, and Anshumali Shrivastava.
\newblock Coding for random projections.
\newblock In \emph{Proceedings of the 31th International Conference on Machine
  Learning (ICML)}, pages 676--684, Beijing, China, 2014.

\bibitem[Li and Li(2019{\natexlab{a}})]{li2019generalization}
Xiaoyun Li and Ping Li.
\newblock Generalization error analysis of quantized compressive learning.
\newblock In \emph{Advances in Neural Information Processing Systems
  (NeurIPS)}, pages 15124--15134, Vancouver, Canada, 2019{\natexlab{a}}.

\bibitem[Li and Li(2019{\natexlab{b}})]{li2019random}
Xiaoyun Li and Ping Li.
\newblock Random projections with asymmetric quantization.
\newblock In \emph{Advances in Neural Information Processing Systems
  (NeurIPS)}, pages 10857--10866, Vancouver, Canada, 2019{\natexlab{b}}.

\bibitem[Li and Li(2021)]{li2021quantization}
Xiaoyun Li and Ping Li.
\newblock Quantization algorithms for random {Fourier} features.
\newblock In \emph{Proceedings of the 38th International Conference on Machine
  Learning (ICML)}, pages 6369--6380, Virtual Event, 2021.

\bibitem[Li and Li(2023{\natexlab{b}})]{li2023differentially}
Xiaoyun Li and Ping Li.
\newblock Differentially private one permutation hashing and bin-wise
  consistent weighted sampling.
\newblock \emph{arXiv preprint}, 2023{\natexlab{b}}.

\bibitem[Mironov(2017)]{mironov2017renyi}
Ilya Mironov.
\newblock R{\'{e}}nyi differential privacy.
\newblock In \emph{Proceedings of the 30th {IEEE} Computer Security Foundations
  Symposium (CSF)}, pages 263--275, Santa Barbara, CA, 2017.

\bibitem[Nissim et~al.(2007)Nissim, Raskhodnikova, and Smith]{nissim2007smooth}
Kobbi Nissim, Sofya Raskhodnikova, and Adam~D. Smith.
\newblock Smooth sensitivity and sampling in private data analysis.
\newblock In \emph{Proceedings of the 39th Annual {ACM} Symposium on Theory of
  Computing (STOC)}, pages 75--84, San Diego, CA, 2007.

\bibitem[Owen(1980)]{owen1980table}
Donald~Bruce Owen.
\newblock A table of normal integrals: A table.
\newblock \emph{Communications in Statistics-Simulation and Computation},
  9\penalty0 (4):\penalty0 389--419, 1980.

\bibitem[Rabanser et~al.(2019)Rabanser, G{\"{u}}nnemann, and
  Lipton]{rabanser2019failing}
Stephan Rabanser, Stephan G{\"{u}}nnemann, and Zachary~C. Lipton.
\newblock Failing loudly: An empirical study of methods for detecting dataset
  shift.
\newblock In \emph{Advances in Neural Information Processing Systems
  (NeurIPS)}, pages 1394--1406, Vancouver, Canada, 2019.

\bibitem[Rothchild et~al.(2020)Rothchild, Panda, Ullah, Ivkin, Stoica,
  Braverman, Gonzalez, and Arora]{rothchild2020fetchsgd}
Daniel Rothchild, Ashwinee Panda, Enayat Ullah, Nikita Ivkin, Ion Stoica,
  Vladimir Braverman, Joseph Gonzalez, and Raman Arora.
\newblock {FetchSGD}: Communication-efficient federated learning with
  sketching.
\newblock In \emph{Proceedings of the 37th International Conference on Machine
  Learning (ICML)}, pages 8253--8265, Virtual Event, 2020.

\bibitem[Shrivastava and Li(2014)]{shrivastava2014defense}
Anshumali Shrivastava and Ping Li.
\newblock In defense of minhash over simhash.
\newblock In \emph{Proceedings of the Seventeenth International Conference on
  Artificial Intelligence and Statistics (AISTATS)}, pages 886--894, Reykjavik,
  Iceland, 2014.

\bibitem[Singhal et~al.(2021)Singhal, Sidahmed, Garrett, Wu, Rush, and
  Prakash]{singhal2021federated}
Karan Singhal, Hakim Sidahmed, Zachary Garrett, Shanshan Wu, John Rush, and
  Sushant Prakash.
\newblock Federated reconstruction: Partially local federated learning.
\newblock In \emph{Advances in Neural Information Processing Systems
  (NeurIPS)}, virtual, 2021.

\bibitem[Slawski and Li(2018)]{slawski2018trade}
Martin Slawski and Ping Li.
\newblock On the trade-off between bit depth and number of samples for a basic
  approach to structured signal recovery from b-bit quantized linear
  measurements.
\newblock \emph{{IEEE} Trans. Inf. Theory}, 64\penalty0 (6):\penalty0
  4159--4178, 2018.

\bibitem[Smith et~al.(2020)Smith, Song, and Thakurta]{smith2020flajolet}
Adam~D. Smith, Shuang Song, and Abhradeep Thakurta.
\newblock The flajolet-martin sketch itself preserves differential privacy:
  Private counting with minimal space.
\newblock In \emph{Advances in Neural Information Processing Systems}, virtual,
  2020.

\bibitem[Soria{-}Comas et~al.(2017)Soria{-}Comas, Domingo{-}Ferrer,
  S{\'{a}}nchez, and Meg{\'{\i}}as]{comas2017individual}
Jordi Soria{-}Comas, Josep Domingo{-}Ferrer, David S{\'{a}}nchez, and David
  Meg{\'{\i}}as.
\newblock Individual differential privacy: {A} utility-preserving formulation
  of differential privacy guarantees.
\newblock \emph{{IEEE} Trans. Inf. Forensics Secur.}, 12\penalty0 (6):\penalty0
  1418--1429, 2017.

\bibitem[Stausholm(2021)]{stausholm2021improved}
Nina~Mesing Stausholm.
\newblock Improved differentially private euclidean distance approximation.
\newblock In \emph{Proceedings of the 40th {ACM} {SIGMOD-SIGACT-SIGAI}
  Symposium on Principles of Database Systems (PODS)}, pages 42--56, Virtual
  Event, China, 2021.

\bibitem[Tomita et~al.(2020)Tomita, Browne, Shen, Chung, Patsolic, Falk,
  Priebe, Yim, Burns, Maggioni, and Vogelstein]{tomita2020sparse}
Tyler~M. Tomita, James Browne, Cencheng Shen, Jaewon Chung, Jesse Patsolic,
  Benjamin Falk, Carey~E. Priebe, Jason Yim, Randal~C. Burns, Mauro Maggioni,
  and Joshua~T. Vogelstein.
\newblock Sparse projection oblique randomer forests.
\newblock \emph{J. Mach. Learn. Res.}, 21:\penalty0 104:1--104:39, 2020.

\bibitem[Vempala(2005)]{vempala2005random}
Santosh~S Vempala.
\newblock \emph{The random projection method}, volume~65.
\newblock American Mathematical Soc., 2005.

\bibitem[Warner(1965)]{warner1965randomized}
Stanley~L Warner.
\newblock Randomized response: A survey technique for eliminating evasive
  answer bias.
\newblock \emph{Journal of the American Statistical Association}, 60\penalty0
  (309):\penalty0 63--69, 1965.

\bibitem[Wei et~al.(2020)Wei, Li, Ding, Ma, Yang, Farokhi, Jin, Quek, and
  Poor]{wei2020federated}
Kang Wei, Jun Li, Ming Ding, Chuan Ma, Howard~H. Yang, Farhad Farokhi, Shi Jin,
  Tony Q.~S. Quek, and H.~Vincent Poor.
\newblock Federated learning with differential privacy: Algorithms and
  performance analysis.
\newblock \emph{{IEEE} Trans. Inf. Forensics Secur.}, 15:\penalty0 3454--3469,
  2020.

\bibitem[Weinberger et~al.(2009)Weinberger, Dasgupta, Langford, Smola, and
  Attenberg]{weinberger2009feature}
Kilian~Q. Weinberger, Anirban Dasgupta, John Langford, Alexander~J. Smola, and
  Josh Attenberg.
\newblock Feature hashing for large scale multitask learning.
\newblock In \emph{Proceedings of the 26th Annual International Conference on
  Machine Learning (ICML)}, pages 1113--1120, Montreal, Canada, 2009.

\bibitem[Wu et~al.(2019)Wu, He, and Xu]{wu2019demo}
Jun Wu, Jingrui He, and Jiejun Xu.
\newblock {DEMO-Net}: Degree-specific graph neural networks for node and graph
  classification.
\newblock In \emph{Proceedings of the 25th {ACM} {SIGKDD} International
  Conference on Knowledge Discovery {\&} Data Mining (KDD)}, pages 406--415,
  Anchorage, AK, 2019.

\bibitem[Xu et~al.(2013)Xu, Zhang, Xiao, Yang, Yu, and
  Winslett]{xu2013differentially}
Jia Xu, Zhenjie Zhang, Xiaokui Xiao, Yin Yang, Ge~Yu, and Marianne Winslett.
\newblock Differentially private histogram publication.
\newblock \emph{{VLDB} J.}, 22\penalty0 (6):\penalty0 797--822, 2013.

\bibitem[Xu et~al.(2021)Xu, Chen, Li, Liu, Song, Lin, and
  Shrivastava]{xu2021locality}
Zhaozhuo Xu, Beidi Chen, Chaojian Li, Weiyang Liu, Le~Song, Yingyan Lin, and
  Anshumali Shrivastava.
\newblock Locality sensitive teaching.
\newblock In \emph{Advances in Neural Information Processing Systems
  (NeurIPS)}, pages 18049--18062, virtual, 2021.

\bibitem[Yang et~al.(2015)Yang, Sato, and Nakagawa]{yang2015bayesian}
Bin Yang, Issei Sato, and Hiroshi Nakagawa.
\newblock Bayesian differential privacy on correlated data.
\newblock In \emph{Proceedings of the 2015 {ACM} {SIGMOD} International
  Conference on Management of Data (SIGMOD)}, pages 747--762, Melbourne,
  Australia, 2015.

\bibitem[Zhang and Li(2020)]{zhang2020optimal}
Hang Zhang and Ping Li.
\newblock Optimal estimator for unlabeled linear regression.
\newblock In \emph{Proceedings of the 37th International Conference on Machine
  Learning (ICML)}, pages 11153--11162, Virtual Event, 2020.

\bibitem[Zhang et~al.(2012)Zhang, Zhang, Xiao, Yang, and
  Winslett]{zhang2012functional}
Jun Zhang, Zhenjie Zhang, Xiaokui Xiao, Yin Yang, and Marianne Winslett.
\newblock Functional mechanism: Regression analysis under differential privacy.
\newblock \emph{Proc. {VLDB} Endow.}, 5\penalty0 (11):\penalty0 1364--1375,
  2012.

\bibitem[Zhang et~al.(2014)Zhang, Cormode, Procopiuc, Srivastava, and
  Xiao]{zhang2014privbayes}
Jun Zhang, Graham Cormode, Cecilia~M. Procopiuc, Divesh Srivastava, and Xiaokui
  Xiao.
\newblock Privbayes: private data release via bayesian networks.
\newblock In \emph{International Conference on Management of Data (SIGMOD)},
  pages 1423--1434, Snowbird, UT, 2014.

\bibitem[Zhang et~al.(2022)Zhang, Wang, Murray, and
  Koniusz]{zhang2022kernelized}
Shan Zhang, Lei Wang, Naila Murray, and Piotr Koniusz.
\newblock Kernelized few-shot object detection with efficient integral
  aggregation.
\newblock In \emph{Proceedings of the {IEEE/CVF} Conference on Computer Vision
  and Pattern Recognition (CVPR)}, pages 19185--19194, New Orleans, LA, 2022.

\bibitem[Zhang et~al.(2021)Zhang, Yin, Chen, Huang, Cui, and
  Zhang]{zhang2021graph}
Shijie Zhang, Hongzhi Yin, Tong Chen, Zi~Huang, Lizhen Cui, and Xiangliang
  Zhang.
\newblock Graph embedding for recommendation against attribute inference
  attacks.
\newblock In \emph{Proceedings of the Web Conference (WWW)}, pages 3002--3014,
  Virtual Event / Ljubljana, Slovenia, 2021.

\bibitem[Zhang et~al.(2020)Zhang, Qi, and Wang]{zhang2020dynamic}
Zhaoqi Zhang, Panpan Qi, and Wei Wang.
\newblock Dynamic malware analysis with feature engineering and feature
  learning.
\newblock In \emph{Proceedings of the Thirty-Fourth {AAAI} Conference on
  Artificial Intelligence (AAAI)}, pages 1210--1217, New York, NY, 2020.

\bibitem[Zhao et~al.(2022)Zhao, Qiao, Redberg, Agrawal, Abbadi, and
  Wang]{zhao2022differentially}
Fuheng Zhao, Dan Qiao, Rachel Redberg, Divyakant Agrawal, Amr~El Abbadi, and
  Yu{-}Xiang Wang.
\newblock Differentially private linear sketches: Efficient implementations and
  applications.
\newblock In \emph{Advances in Neural Information Processing Systems
  (NeurIPS)}, 2022.

\bibitem[Zymnis et~al.(2010)Zymnis, Boyd, and
  Cand{\`{e}}s]{zymnis2010compressed}
Argyrios Zymnis, Stephen~P. Boyd, and Emmanuel~J. Cand{\`{e}}s.
\newblock Compressed sensing with quantized measurements.
\newblock \emph{{IEEE} Signal Process. Lett.}, 17\penalty0 (2):\penalty0
  149--152, 2010.

\end{thebibliography}
\bibliographystyle{plainnat}



\newpage

\appendix

\section{Deferred Proofs}

\subsection{Proof of Lemma~\ref{lemma:half-normal-tail}}

\noindent \textbf{Lemma 3.6} (Half-normal tail bound).
Let $X_1,...,X_n$ be iid $N(0,1)$ random variables and let $Y_i=|X_i|$, $\forall i$. Denote $Z=\sum_{i=1}^n Y_i$. Then for any $t>0$, it holds that
\begin{align*}
    Pr(Z\geq \sqrt{2n^2\log 2+2nt})\leq \exp(-t).
\end{align*}
\begin{proof}
By the Markov inequality, we write (with some $\zeta>0$)
\begin{align*}
    Pr(Z\geq t)=Pr(\exp(\zeta\sum_{i=1}^n Y_i)\geq \exp(\zeta t))\leq \frac{\mathbb E[\exp(\zeta\sum_{i=1}^n Y_i]}{\exp(\zeta t)}.
\end{align*}
By the independence of $Y_i$ and the moment generating function of half-normal distribution, we have
\begin{align*}
    \mathbb E[\exp(\zeta\sum_{i=1}^n Y_i]=\prod_{i=1}^n \mathbb E[\exp(\zeta Y_i)]=\exp(\frac{n\zeta^2}{2})(2\Phi(\zeta))^n&=\exp(\frac{n\zeta^2}{2}+n\log 2\Phi(\zeta)) \\
    &\leq \exp(\frac{n\zeta^2}{2}+n\log 2).
\end{align*}
Hence, we have
\begin{align*}
    Pr(Z\geq t)&\leq \exp(\frac{n\zeta^2}{2}-t\zeta+n\log 2),
\end{align*}
which is minimized at $\zeta=\frac{t}{n}$, leading to
\begin{align*}
    Pr(Z\geq t)&\leq \exp(-\frac{t^2}{2n}+n\log 2).
\end{align*}
Taking $t=\sqrt{2n^2\log 2+2nt}$ completes the proof.
\end{proof}

\subsection{Proof of Lemma~\ref{lemma:conditional-gaussian} }  \label{app-sec:proof}

\begin{proof}

The bivariate normal density function is
\begin{align*}
    f(x,y)&=\frac{1}{2\pi\sigma_x\sigma_y\sqrt{1-\rho^2}}\exp\left(-\frac{\frac{x^2}{\sigma_x^2}-\frac{2\rho xy}{\sigma_x\sigma_y}+\frac{y^2}{\sigma_y^2}}{2(1-\rho^2)} \right)\\
    &=\frac{1}{2\pi\sigma_x\sigma_y\sqrt{1-\rho^2}} \exp\left(-\frac{x^2}{2\sigma_x^2}\right)\exp\left( -\frac{(\frac{y}{\sigma_y}-\rho\frac{x}{\sigma_x})^2}{2(1-\rho^2)}\right).
\end{align*}
Therefore, we have $\mathbb E\left[|X|\big| |X|>|Y|\right]=\frac{A}{P}$, with
\begin{align*}
    A&=\int_{-\infty}^\infty \frac{|x|}{\sqrt{2\pi}\sigma_x\sqrt{1-\rho^2}} \exp\left(-\frac{x^2}{2\sigma_x^2}\right)dx\int_{-|x|}^{|x|}\frac{1}{\sqrt{2\pi}\sigma_y}\exp\left( -\frac{(\frac{y}{\sigma_y}-\rho\frac{x}{\sigma_x})^2}{2(1-\rho^2)}\right) dy, \\
    P&=\int_{-\infty}^\infty \frac{1}{\sqrt{2\pi}\sigma_x\sqrt{1-\rho^2}} \exp\left(-\frac{x^2}{2\sigma_x^2}\right)dx\int_{-|x|}^{|x|}\frac{1}{\sqrt{2\pi}\sigma_y}\exp\left( -\frac{(\frac{y}{\sigma_y}-\rho\frac{x}{\sigma_x})^2}{2(1-\rho^2)}\right) dy.
\end{align*}
Note that $P=Pr(|X|>|Y|)$ in the first statement of the theorem. Our calculation will use the following two identities involving the Gaussian functions~\citep{owen1980table}:
\begin{align}
    &\int_0^\infty \phi(ax)\Phi(bx)dx=\frac{1}{2\pi |a|}\left( \frac{\pi}{2} +\tan^{-1}\left( \frac{b}{|a|} \right)\right),  \label{eqn:gp} \\
    &\int_0^\infty x\phi(ax)\Phi(bx)dx=\frac{1}{2\sqrt{2\pi}}\left( 1 +\frac{b}{\sqrt{1+b^2}}\right),  \label{eqn:gx}
\end{align}
where $\phi(x)$ and $\Phi(x)$ are the pdf and cdf of the standard Gaussian distribution.

With a proper change of random variables, we can compute $A$ as
\begin{align*}
    A=&\int_{-\infty}^\infty \frac{|x|}{\sqrt{2\pi}\sigma_x\sqrt{1-\rho^2}} \exp\left(-\frac{x^2}{2\sigma_x^2}\right)dx\int_{\frac{\frac{-|x|}{\sigma_y}-\rho\frac{x}{\sigma_x}}{\sqrt{1-\rho^2}}}^{\frac{\frac{|x|}{\sigma_y}-\rho\frac{x}{\sigma_x}}{\sqrt{1-\rho^2}}} \sqrt{1-\rho^2}\frac{1}{\sqrt{2\pi}}e^{-s^2}ds \\
    =&\int_{-\infty}^\infty \frac{|x|}{\sqrt{2\pi}\sigma_x} \exp\left(-\frac{x^2}{2\sigma_x^2}\right)\left[ \Phi\left( \frac{\frac{|x|}{\sigma_y}-\rho\frac{x}{\sigma_x}}{\sqrt{1-\rho^2}}\right) - \Phi\left( \frac{\frac{-|x|}{\sigma_y}-\rho\frac{x}{\sigma_x}}{\sqrt{1-\rho^2}}\right) \right]dx \\
    =&\int_{-\infty}^\infty \frac{\sigma_x|t|}{\sqrt{2\pi}}\exp\left(-\frac{t^2}{2}\right)\left[ \Phi\left( \frac{\frac{\sigma_x}{\sigma_y}|t|-\rho t}{\sqrt{1-\rho^2}}\right) - \Phi\left( -\frac{\frac{\sigma_x}{\sigma_y}|t|-\rho t}{\sqrt{1-\rho^2}}\right) \right]dt\\
    \eqdef& A_1-A_2.
\end{align*}
For the first term we have
\begin{align*}
    A_1&=\sigma_x\left[ \int_{0}^\infty \frac{t}{\sqrt{2\pi}}\exp\left(-\frac{t^2}{2}\right)\Phi\left( \frac{\frac{\sigma_x}{\sigma_y}-\rho }{\sqrt{1-\rho^2}}t\right)dt + \int_{-\infty}^0 \frac{-t}{\sqrt{2\pi}}\exp\left(-\frac{t^2}{2}\right)\Phi\left( -\frac{\frac{\sigma_x}{\sigma_y}+\rho }{\sqrt{1-\rho^2}}t \right)dt \right] \\
    &=\sigma_x\left[ \int_{0}^\infty \frac{t}{\sqrt{2\pi}}\exp\left(-\frac{t^2}{2}\right)\Phi\left( \frac{\frac{\sigma_x}{\sigma_y}-\rho }{\sqrt{1-\rho^2}}t\right)dt + \int_0^\infty \frac{s}{\sqrt{2\pi}}\exp\left(-\frac{s^2}{2}\right)\Phi\left( \frac{\frac{\sigma_x}{\sigma_y}+\rho }{\sqrt{1-\rho^2}}s \right)ds \right] \\
    &=\sigma_x\left[ \frac{1}{2\sqrt{2\pi}}\left( 1+\frac{\frac{\frac{\sigma_x}{\sigma_y}-\rho }{\sqrt{1-\rho^2}}}{\sqrt{1+\frac{(\frac{\sigma_x}{\sigma_y}-\rho)^2 }{1-\rho^2}}} \right) +\frac{1}{2\sqrt{2\pi}} \left( 1+\frac{\frac{\frac{\sigma_x}{\sigma_y}+\rho }{\sqrt{1-\rho^2}}}{\sqrt{1+\frac{(\frac{\sigma_x}{\sigma_y}+\rho)^2 }{1-\rho^2}}} \right) \right]\\
    &=\sigma_x\left[ \frac{1}{\sqrt{2\pi}} + \frac{1}{2\sqrt{2\pi}}\left( \frac{r-\rho}{\sqrt{1+r^2-2r\rho}}+\frac{r+\rho}{\sqrt{1+r^2+2r\rho}} \right) \right],
\end{align*}
where we denote $r=\frac{\sigma_x}{\sigma_y}$ and use (\ref{eqn:gx}). Similarly, we have that
\begin{align*}
    A_2&=\sigma_x\left[ \int_{0}^\infty \frac{t}{\sqrt{2\pi}}\exp\left(-\frac{t^2}{2}\right)\Phi\left( -\frac{r+\rho }{\sqrt{1-\rho^2}}t\right)dt + \int_{-\infty}^0 \frac{-t}{\sqrt{2\pi}}\exp\left(-\frac{t^2}{2}\right)\Phi\left( \frac{r-\rho }{\sqrt{1-\rho^2}}t \right)dt \right] \\
    &=\sigma_x\left[ \int_{0}^\infty \frac{t}{\sqrt{2\pi}}\exp\left(-\frac{t^2}{2}\right)\Phi\left( -\frac{r+\rho }{\sqrt{1-\rho^2}}t\right)dt + \int_0^\infty \frac{s}{\sqrt{2\pi}}\exp\left(-\frac{s^2}{2}\right)\Phi\left( -\frac{r-\rho }{\sqrt{1-\rho^2}}s \right)ds \right] \\
    &=\sigma_x\left[ \frac{1}{\sqrt{2\pi}} - \frac{1}{2\sqrt{2\pi}}\left( \frac{r-\rho}{\sqrt{1+r^2-2r\rho}}+\frac{r+\rho}{\sqrt{1+r^2+2r\rho}} \right) \right].
\end{align*}
Therefore, we obtain
\begin{align}
    A(\rho,r) = A_1-A_2 = \frac{\sigma_x}{\sqrt{2\pi}}\left( \frac{r-\rho}{\sqrt{1+r^2-2r\rho}}+\frac{r+\rho}{\sqrt{1+r^2+2r\rho}} \right).  \label{eqn:A}
\end{align}

To compute $P$, by doing a similar change of variables, we have
\begin{align*}
    P=&\int_{-\infty}^\infty \frac{1}{\sqrt{2\pi}}\exp\left(-\frac{t^2}{2}\right)\left[ \Phi\left( \frac{r|t|-\rho t}{\sqrt{1-\rho^2}}\right) - \Phi\left( -\frac{r|t|-\rho t}{\sqrt{1-\rho^2}}\right) \right]dt\\
    \eqdef& P_1-P_2.
\end{align*}
Using (\ref{eqn:gp}), we obtain
\begin{align*}
    P_1&= \int_{0}^\infty \frac{1}{\sqrt{2\pi}}\exp\left(-\frac{t^2}{2}\right)\Phi\left( \frac{r-\rho }{\sqrt{1-\rho^2}}t\right)dt + \int_{-\infty}^0 \frac{1}{\sqrt{2\pi}}\exp\left(-\frac{t^2}{2}\right)\Phi\left( -\frac{r+\rho }{\sqrt{1-\rho^2}}t \right)dt  \\
    &=\int_{0}^\infty \frac{1}{\sqrt{2\pi}}\exp\left(-\frac{t^2}{2}\right)\Phi\left( \frac{r-\rho }{\sqrt{1-\rho^2}}t\right)dt + \int_0^\infty \frac{1}{\sqrt{2\pi}}\exp\left(-\frac{s^2}{2}\right)\Phi\left( \frac{r+\rho }{\sqrt{1-\rho^2}}s \right)ds  \\
    &=\frac{1}{2\pi}\left(  \frac{\pi}{2} + \tan^{-1}\left( \frac{r-\rho }{\sqrt{1-\rho^2}} \right)\right) + \frac{1}{2\pi}\left(  \frac{\pi}{2} + \tan^{-1}\left( \frac{r+\rho }{\sqrt{1-\rho^2}} \right)\right)\\
    &=\frac{1}{2}+\frac{1}{2\pi}\left[ \tan^{-1}\left( \frac{r-\rho }{\sqrt{1-\rho^2}} \right)+\tan^{-1}\left( \frac{r+\rho }{\sqrt{1-\rho^2}} \right) \right],\\
    P_2&=\frac{1}{2}-\frac{1}{2\pi}\left[ \tan^{-1}\left( \frac{r-\rho }{\sqrt{1-\rho^2}} \right)+\tan^{-1}\left( \frac{r+\rho }{\sqrt{1-\rho^2}} \right) \right],
\end{align*}
which leads to
\begin{align}
    P(\rho,r)=P_1-P_2=&\frac{1}{\pi}\left[ \tan^{-1}\left( \frac{r-\rho }{\sqrt{1-\rho^2}} \right)+\tan^{-1}\left( \frac{r+\rho }{\sqrt{1-\rho^2}} \right) \right] \label{eqn:P}
\end{align}
Therefore, we know that
\begin{align*}
    \mathbb E\left[|X|\big| |X|>|Y|\right]=\frac{A(\rho,r)}{P(\rho,r)}=\sigma_x\sqrt{\frac{\pi}{2}}\cdot \frac{\frac{r-\rho}{\sqrt{1+r^2-2r\rho}}+\frac{r+\rho}{\sqrt{1+r^2+2r\rho}}}{\tan^{-1}\left( \frac{r-\rho }{\sqrt{1-\rho^2}} \right)+\tan^{-1}\left( \frac{r+\rho }{\sqrt{1-\rho^2}} \right)},
\end{align*}
with $r=\sigma_x/\sigma_y$. We now investigate the derivative of $P$. By some algebra, we can show that
\begin{align*}
    \frac{\partial P(\rho,r)}{\partial\rho}=\frac{2r\rho(r^2-1)}{(1+r^2-2r\rho)(1+r^2+2r\rho)\sqrt{1-\rho^2}}.
\end{align*}
When $0<r\leq 1$, $\frac{\partial P(\rho,r)}{\partial\rho}\geq 0$ when $\rho\leq 0$ and $\frac{\partial P(\rho,r)}{\partial\rho}\leq 0$ when $\rho> 0$. Therefore, $\max_\rho P(\rho,r)=P(0,r)=\frac{2}{\pi}\tan^{-1}(r)$.

\vspace{0.1in}
\noindent\textbf{Tail bound.} By our previous calculations, the conditional distribution of $X$ given $|X|>|Y|$ is
\begin{align*}
    f(x\big| |X|>|Y|)=\frac{\frac{1}{\sqrt{2\pi}}\exp\left(-\frac{x^2}{2}\right)\left[ \Phi\left( \frac{r|x|-\rho x}{\sqrt{1-\rho^2}}\right) - \Phi\left( -\frac{r|x|-\rho x}{\sqrt{1-\rho^2}}\right) \right]}{P},\quad x\in \mathbb R,
\end{align*}
with $P=Pr(|X|>|Y|)$ in (\ref{eqn:P}) the normalizing constant to make the integral equal to 1.

The conditional tail probability can be computed as follows. For some $t>0$, by symmetry,
\begin{align*}
    &Pr(|X|>t, |X|>|Y|)\\
    =&2\int_{t}^\infty \frac{1}{\sqrt{2\pi}\sigma_x}\exp\left(-\frac{x^2}{2\sigma_x^2}\right)\left[ \Phi\left( \frac{r|x|-\rho x}{\sqrt{1-\rho^2}}\right) - \Phi\left( -\frac{r|x|-\rho x}{\sqrt{1-\rho^2}}\right) \right]dx \\
    =&2\int_{\frac{t}{\sigma_x}}^\infty \frac{1}{\sqrt{2\pi}}\exp\left(-\frac{x^2}{2}\right)\left[ \Phi\left( \frac{r|x|-\rho x}{\sqrt{1-\rho^2}}\right) - \Phi\left( -\frac{r|x|-\rho x}{\sqrt{1-\rho^2}}\right) \right]dx\\
    \eqdef& 2(\tilde P_1- \tilde P_2).
\end{align*}
For $\tilde P_1$, using polar coordinates we have
\begin{align*}
    \tilde P_1=&\frac{1}{2\pi}\int_{\frac{t}{\sigma_x}}^\infty e^{-\frac{x^2}{2}} dx \int_{-\infty}^{\frac{r-\rho}{\sqrt{1-\rho^2}}x} e^{-\frac{y^2}{2}}dy\\
    =&\frac{1}{2\pi}\int_{-\frac{\pi}{2}}^{\tan^{-1}( \frac{r-\rho}{\sqrt{1-\rho^2}})}d\theta \int_{\frac{t}{\sigma_x\cos(\theta)}}^\infty e^{-\frac{r^2}{2}}r dr \\
    =&\frac{1}{2\pi} \int_{-\frac{\pi}{2}}^{\tan^{-1}( \frac{r-\rho}{\sqrt{1-\rho^2}})} \exp\left(-\frac{t^2}{2\sigma_x^2\cos^2(\theta)}\right) d\theta.
\end{align*}
Similarly,
\begin{align*}
    \tilde P_2&=\frac{1}{2\pi} \int_{-\frac{\pi}{2}}^{\tan^{-1}( -\frac{r+\rho}{\sqrt{1-\rho^2}})} \exp\left(-\frac{t^2}{2\sigma_x^2\cos^2(\theta)}\right) d\theta.
\end{align*}
Therefore, we obtain
\begin{align*}
    Pr(|X|>t, |X|>|Y|)&=\frac{1}{\pi}\int_{\tan^{-1}( -\frac{r+\rho}{\sqrt{1-\rho^2}})}^{\tan^{-1}( \frac{r-\rho}{\sqrt{1-\rho^2}})} \exp\left(-\frac{t^2}{2\sigma_x^2\cos^2(\theta)}\right) d\theta \\
    &=\frac{1}{\pi}\int_{-\tan^{-1}( \frac{r+\rho}{\sqrt{1-\rho^2}})}^{\tan^{-1}( \frac{r-\rho}{\sqrt{1-\rho^2}})} \exp\left(-\frac{t^2}{2\sigma_x^2\cos^2(\theta)}\right) d\theta \\
    &\leq e^{-\frac{t^2}{2\sigma_x^2}}\frac{1}{\pi}\int_{-\tan^{-1}( \frac{r+\rho}{\sqrt{1-\rho^2}})}^{\tan^{-1}( \frac{r-\rho}{\sqrt{1-\rho^2}})} d\theta,
\end{align*}
since $\cos^2(\theta)\in [0,1]$. Notice that $P$ in (\ref{eqn:P}) can be written as $P=\frac{1}{\pi}\int_{-\tan^{-1}( \frac{r+\rho}{\sqrt{1-\rho^2}})}^{\tan^{-1}( \frac{r-\rho}{\sqrt{1-\rho^2}})} d\theta$. Hence, we know that the conditional tail probability is
\begin{align*}
    Pr(|X|>t\big| |X|>|Y|) = \frac{Pr(|X|>t, |X|>|Y|)}{Pr(|X|>|Y|)}\leq \exp\left( -\frac{t^2}{2\sigma_x^2} \right),\quad \forall r>0, \rho\in (-1,1).
\end{align*}
At the boundaries $\rho=1$, $\rho=-1$, one can verify $Pr(|X|>|Y|)=0$. This concludes the proof.
\end{proof}

\end{document}